\documentclass[12pt]{article} 
\pdfoutput=1

\usepackage[round]{natbib}
\renewcommand{\cite}{\citep}

\usepackage{amsmath, amsfonts, amsthm}
\usepackage[ruled, noline]{algorithm2e}
\usepackage{xspace}
\usepackage{fullpage}
\usepackage{verbatim}
\usepackage[pdftex]{graphicx}
\usepackage{grffile}

\newtheorem{theorem}{Theorem}
\newtheorem{lemma}{Lemma}

\newtheorem{remark}{Remark}
\newcommand{\eq}[1]{(\ref{eq:#1})}

\newcommand{\F}{\mathcal{F}}
\newcommand{\G}{\mathbb{G}}
\newcommand{\I}{\mathcal{I}}

\newcommand{\R}{\mathbb{R}}
\renewcommand{\P}{\mathbb{P}}
\newcommand{\Q}{\mathbb{Q}}
\newcommand{\X}{\mathcal{X}}
\newcommand{\1}{\mathbf{1}}
\newcommand{\Normal}{\text{Normal}}
\newcommand{\StudentT}{\text{StudentT}}
\newcommand{\GammaDist}{\text{Gamma}}
\newcommand{\Bernoulli}{\text{Bernoulli}}
\newcommand{\Multinomial}{\text{Multinomial}}
\newcommand{\pconv}{\overset{P}{\rightarrow}}

\newcommand{\dequal}{\overset{d}{=}}
\DeclareMathOperator{\Var}{Var}

\newcommand{\eout}[1]{#1^*}  
\newcommand{\ein}[1]{#1_*}   
\newcommand{\s}[1]{{#1}^*} 
\newcommand{\sest}[1]{\hat{#1}}	     

\newcommand{\numsub}{s}      
\newcommand{\numre}{r}         

\newcommand{\OurAlg}{Bag of Little Bootstraps\xspace}  
\newcommand{\OurAlgWithAbbrev}{Bag of Little Bootstraps (BLB)\xspace}  
\newcommand{\OuralgWithAbbrev}{The Bag of Little Bootstraps (BLB)\xspace}   
\newcommand{\OuralgAbbrev}{BLB\xspace}  
\newcommand{\ouralgWithAbbrev}{the Bag of Little Bootstraps (BLB)\xspace}  
\newcommand{\ouralgAbbrev}{BLB\xspace}  

\title{A Scalable Bootstrap for Massive Data}

\author{
Ariel Kleiner \\
Department of Electrical Engineering and Computer Science\\
University of California, Berkeley\\
\texttt{akleiner@eecs.berkeley.edu} \\
\and
Ameet Talwalkar \\
Department of Electrical Engineering and Computer Science\\
University of California, Berkeley\\
\texttt{ameet@eecs.berkeley.edu} \\
\and
Purnamrita Sarkar \\
Department of Electrical Engineering and Computer Science\\
University of California, Berkeley\\
\texttt{psarkar@eecs.berkeley.edu} \\
\and
Michael I.~Jordan \\
Department of Statistics\\
and Department of Electrical Engineering and Computer Science\\
University of California, Berkeley\\
\texttt{jordan@eecs.berkeley.edu}
}

\begin{document}
\maketitle

\begin{abstract}
The bootstrap provides a simple and powerful means of assessing the quality of 
estimators.  However, in settings involving large datasets---which are increasingly 
prevalent---the computation of bootstrap-based quantities can be prohibitively 
demanding computationally.  While variants such as subsampling and the $m$ out 
of $n$ bootstrap can be used in principle to reduce the cost of bootstrap 
computations, we find that these methods are generally not robust to specification 
of hyperparameters (such as the number of subsampled data points), and they often 
require use of more prior information (such as rates of convergence of estimators) 
than the bootstrap.  As an alternative, we introduce \ouralgWithAbbrev, a new 
procedure which incorporates features of both the bootstrap and subsampling 
to yield a robust, computationally efficient means of assessing the quality of estimators.  \OuralgAbbrev is well suited to modern parallel and distributed computing architectures 
and furthermore retains the generic applicability and statistical efficiency of 
the bootstrap.  We demonstrate \ouralgAbbrev's favorable statistical performance 
via a theoretical analysis elucidating the procedure's properties, as well as a 
simulation study comparing \ouralgAbbrev to the bootstrap, the $m$ out of $n$ 
bootstrap, and subsampling.  In addition, we present results from a large-scale 
distributed implementation of \ouralgAbbrev demonstrating its computational 
superiority on massive data, a method for adaptively selecting \ouralgAbbrev's 
hyperparameters, an empirical study applying \ouralgAbbrev to several real datasets,
and an extension of \ouralgAbbrev to time series data.
\end{abstract}

\section{Introduction}
\label{sec:introduction}

The development of the bootstrap and related resampling-based methods 
in the 1960s and 1970s heralded an era in statistics in which inference 
and computation became increasingly intertwined~\cite{bootstrap-efron,DiaconisEfron}.  
By exploiting the basic capabilities of the classical von Neumann computer 
to simulate and iterate, the bootstrap made it possible to use computers 
not only to compute estimates but also to assess the quality of estimators, 
yielding results that are quite generally consistent~\cite{bickel-freedman, 
gine-zinn, vdv-wellner} and often more accurate than those based upon 
asymptotic approximation~\cite{hall-edgeworth}.  Moreover, the bootstrap
aligned statistics to computing technology, such that advances in speed and storage 
capacity of computers could immediately allow statistical methods to scale 
to larger datasets.

Two recent trends are worthy of attention in this regard.  First, 
the growth in size of datasets is accelerating, with ``massive''  
datasets becoming increasingly prevalent.  Second, computational resources are 
shifting toward parallel and distributed architectures, with multicore and 
cloud computing platforms providing access to hundreds or thousands 
of processors.  The second trend is seen as a mitigating factor with 
respect to the first, in that parallel and distributed architectures
present new capabilities for storage and manipulation of data.  
However, from an inferential point of view, it is not yet clear how 
statistical methodology will transport to a world involving 
massive data on parallel and distributed computing platforms.

While massive data bring many statistical issues to the fore, including 
issues in exploratory data analysis and data visualization, there remains the
core inferential need to assess the quality of estimators.
Indeed, the uncertainty and biases in estimates based on large data can remain quite significant, as large datasets are often high dimensional, are frequently used to fit complex models with large numbers of parameters, and can have many potential sources of bias.
Furthermore, even if sufficient data are available to allow highly accurate estimation, the ability to efficiently assess estimator quality remains essential to allow efficient use of available resources by processing only as much data as is necessary to achieve a desired accuracy or confidence.

The bootstrap brings to bear various desirable features in the massive data setting, notably its relatively automatic nature and its applicability to a wide variety of inferential problems.  It can be used to assess bias, to quantify the uncertainty in an estimate (e.g., via a standard error or a confidence interval), or to assess risk.  However, these virtues are realized at the expense of a substantial computational burden.  Bootstrap-based quantities typically must be computed via a form of Monte Carlo approximation in which the estimator in question is repeatedly applied to resamples of the entire original observed dataset.

Because these resamples have size on the order of that of the original data, 
with approximately 63\% of data points appearing at least once in each resample, 
the usefulness of the bootstrap is severely blunted by the large datasets increasingly 
encountered in practice.  In the massive data setting, computation of even a single 
point estimate on the full dataset can be quite computationally demanding, and so 
repeated computation of an estimator on comparably sized resamples can be prohibitively 
costly.  To mitigate this problem, one might naturally attempt to exploit the modern
trend toward parallel and distributed computing.  Indeed, at first glance, the bootstrap would seem ideally suited
to straightforwardly leveraging parallel and distributed computing architectures: one might imagine using
different processors or compute nodes to process different bootstrap resamples independently in parallel.
However, the large size of bootstrap resamples in the massive 
data setting renders this approach problematic, as the cost of transferring data
to independent processors or compute nodes can be overly high, as is the cost of operating on even a single resample using an independent set of computing resources.

While the literature does contain some discussion of techniques for improving 
the computational efficiency of the bootstrap, that work is largely devoted 
to reducing the number of resamples required~\cite{efficient-bootstrap-efron, 
efron-intro-bootstrap}.  These techniques in general introduce significant 
additional complexity of implementation and do not eliminate the crippling 
need for repeated computation of the estimator on resamples having size 
comparable to that of the original dataset.

Another landmark in the development of simulation-based inference is 
subsampling~\cite{politis-subsampling-book} and the closely related  
$m$ out of $n$ bootstrap~\cite{bootstrap-moutofn}.  These methods (which 
were introduced to achieve statistical consistency in edge cases in which 
the bootstrap fails) initially appear to remedy the bootstrap's key 
computational shortcoming, as they only require repeated computation 
of the estimator under consideration on resamples (or subsamples) that 
can be significantly smaller than the original dataset.  However, these 
procedures also have drawbacks.  As we show in our simulation study, their success is sensitive to the choice of resample (or subsample) size (i.e., $m$ in the $m$ out of $n$ bootstrap).  Additionally, because the variability of an estimator on a subsample differs from its variability on the 
full dataset, these procedures must perform a rescaling of their output, and this rescaling requires knowledge and explicit use of the convergence rate of the estimator in question; these methods are thus less automatic and easily deployable than the bootstrap.  While schemes have been proposed for data-driven selection of an optimal resample size~\cite{bickel-sakov-choice-of-m}, they require significantly greater computation which would eliminate any computational gains.  Also, there has been work on the $m$ out of $n$ bootstrap that has sought to reduce computational costs using two different values of $m$ in conjunction with extrapolation~\cite{bickel-yahav-extrapolation, bickel-sakov-extrapolation}.  However, these approaches explicitly utilize series expansions of the estimator's sampling distribution and hence are less automatically usable; they also require execution of the $m$ out of $n$ bootstrap for multiple values of $m$.

Motivated by the need for an automatic, accurate means of assessing estimator 
quality that is scalable to large datasets, we introduce a new procedure, 
\ouralgWithAbbrev, which functions by combining the results of bootstrapping 
multiple small subsets of a larger original dataset. Instead of applying an 
estimator directly to each small subset, as in the $m$ out of $n$ bootstrap 
and subsampling, \ouralgWithAbbrev applies the bootstrap to each small subset, 
where in the resampling process of each individual bootstrap run, weighted 
samples are formed such that the effect is that of sampling the small subset 
$n$ times with replacement, but the computational cost is that associated 
with the size of the small subset.  This has the effect that, despite operating 
only on subsets of the original dataset, \ouralgAbbrev does not require analytical
rescaling of its output.  Overall, \OuralgAbbrev has a significantly
more favorable computational profile than the bootstrap, as it only
requires repeated computation of the estimator under consideration on
quantities of data that can be much smaller than the original dataset.
As a result, \ouralgAbbrev is well suited to implementation on modern
distributed and parallel computing architectures which are often used
to process large datasets.  Also, our procedure maintains the bootstrap's
generic applicability, favorable statistical properties (i.e., consistency
and higher-order correctness), and simplicity of implementation.
Finally, as we show in experiments, \ouralgAbbrev is consistently more
robust than alternatives such as the $m$ out of $n$ bootstrap and subsampling.

The remainder of our presentation is organized as follows.  In Section~\ref{sec:blb}, we formalize our statistical setting and notation, present \ouralgAbbrev in detail, and discuss the procedure's computational characteristics.  Subsequently, in Section~\ref{sec:stat-perform}, we elucidate \ouralgAbbrev's statistical properties via a theoretical analysis (Section~\ref{sec:theory}) showing that \ouralgAbbrev shares the bootstrap's consistency and higher-order correctness, as well as a simulation study (Section~\ref{sec:simulation}) which compares \ouralgAbbrev to the bootstrap, the $m$ out of $n$ bootstrap, and subsampling.  Section~\ref{sec:scalability} discusses a large-scale implementation of \ouralgAbbrev on a distributed computing system and presents results illustrating the procedure's superior computational performance in the massive data setting.  We present a method for adaptively selecting \ouralgAbbrev's hyperparameters in Section~\ref{sec:hyperparam}.  Finally, we apply \ouralgAbbrev (as well as the bootstrap and the $m$ out of $n$ bootstrap, for comparison) to several real datasets in Section~\ref{sec:real-data}, we present an extension of \ouralgAbbrev to time series data in Section~\ref{sec:time-series}, and we conclude in Section~\ref{sec:conclusion}.

\section{\OurAlgWithAbbrev}
\label{sec:blb}

\subsection{Setting and Notation}
We assume that we observe a sample $X_1, \ldots, X_n \in \X$ drawn i.i.d.\ from some (unknown) underlying distribution $P \in \mathcal{P}$; we denote by $\P_n = n^{-1} \sum_{i=1}^n \delta_{X_i}$ the corresponding empirical distribution.  Based only on this observed data, we compute an estimate $\sest{\theta}_n \in \Theta$ of some (unknown) population value $\theta \in \Theta$ associated with $P$.  For example, $\sest{\theta}_n$ might estimate a measure of correlation, the parameters of a regressor, or the prediction accuracy of a trained classification model.  When we wish to explicitly indicate the data used to compute an estimate, we shall write $\sest{\theta}(\P_n)$.  Noting that $\sest{\theta}_n$ is a random quantity because it is based on $n$ random observations, we define $Q_n(P) \in \mathcal{Q}$ as the true underlying distribution of $\sest{\theta}_n$, which is determined by both $P$ and the form of the estimator.  Our end goal is the computation of an estimator quality assessment $\xi(Q_n(P), P): \mathcal{Q} \times \mathcal{P} \rightarrow \Xi$, for $\Xi$ a vector space; to lighten notation, we shall interchangeably write $\xi(Q_n(P))$ in place of $\xi(Q_n(P), P)$.  For instance, $\xi$ might compute a quantile, a confidence region, a standard error, or a bias.  In practice, we do not have direct knowledge of $P$ or $Q_n(P)$, and so we must estimate $\xi(Q_n(P))$ itself based only on the observed data and knowledge of the form of the estimator under consideration.

Note that we allow $\xi$ to depend directly on $P$ in addition to $Q_n(P)$ because $\xi$ might operate on the distribution of a centered and normalized version of $\sest{\theta}_n$.  For example, if $\xi$ computes a confidence region, it might manipulate the distribution of the statistic $\sqrt{n}(\sest{\theta}_n - \theta)$, which is determined by both $Q_n(P)$ and $\theta$; because $\theta$ cannot in general be obtained directly from $Q_n(P)$, a direct dependence on $P$ is required in this case.  Nonetheless, given knowledge of $Q_n(P)$, any direct dependence of $\xi$ on $P$ generally has a simple form, often only involving the parameter $\theta$.  Additionally, rather than restricting $Q_n(P)$ to be the distribution of $\sest{\theta}_n$, we could instead allow it to be the distribution of a more general statistic, such as $(\sest{\theta}_n, \sest{\sigma}_n)$, where $\sest{\sigma}_n$ is an estimate of the standard deviation of $\sest{\theta}_n$ (e.g., this would apply when constructing confidence intervals based on the distribution of the studentized statistic $(\sest{\theta}_n - \theta)/\sest{\sigma}_n$).  Our subsequent development generalizes straightforwardly to this setting, but to simplify the exposition, we will largely assume that $Q_n(P)$ is the distribution of $\sest{\theta}_n$.

Under our notation, the bootstrap simply computes the data-driven plugin approximation $\xi(Q_n(P)) \approx \xi(Q_n(\P_n))$.  Although $\xi(Q_n(\P_n))$ cannot be computed exactly in most cases, it is generally amenable to straightforward Monte Carlo approximation via the following algorithm~\cite{efron-intro-bootstrap}: repeatedly resample $n$ points i.i.d.\ from $\P_n$, compute the estimate on each resample, form the empirical distribution $\s{\Q}_n$ of the computed estimates, and approximate $\xi(Q_n(P)) \approx \xi(\s{\Q}_n)$.

Similarly, using our notation, the $m$ out of $n$ bootstrap (and subsampling) functions as follows, for $m < n$~\cite{bootstrap-moutofn, politis-subsampling-book}: repeatedly resample $m$ points i.i.d.\ from $\P_n$ (subsample $m$ points without replacement from $X_1, \ldots, X_n$), compute the estimate on each resample (subsample), form the empirical distribution $\s{\Q}_{m}$ of the computed estimates, approximate $\xi(Q_m(P)) \approx \xi(\s{\Q}_{m})$, and apply an analytical correction to in turn approximate $\xi(Q_n(P))$.  This final analytical correction uses prior knowledge of the convergence rate of $\sest{\theta}_n$ as $n$ increases and is necessary because each value of the estimate is computed based on only $m$ rather than $n$ points.

We use $\1_d$ to denote the $d$-dimensional vector of ones, and we let $I_d$ denote the $d \times d$ identity matrix.

\subsection{\OurAlg}

\OuralgWithAbbrev functions by averaging the results of bootstrapping multiple small subsets of $X_1, \ldots, X_n$.  More formally, given a subset size $b < n$, \ouralgAbbrev samples $\numsub$ subsets of size $b$ from the original $n$ data points, uniformly at random (one can also impose the constraint that the subsets be disjoint).  Let $\I_1, \ldots, \I_{\numsub} \subset \{1, \ldots, n\}$ be the corresponding index multisets (note that $| \I_j | = b, \forall j$), and let $\P_{n,b}^{(j)} = b^{-1} \sum_{i \in \I_j} \delta_{X_i}$ be the empirical distribution corresponding to subset $j$.  \OuralgAbbrev's estimate of $\xi(Q_n(P))$ is then given by
\begin{equation}
\label{eq:ouralg}
\numsub^{-1} \sum_{j=1}^{\numsub} \xi(Q_n(\P_{n,b}^{(j)})).
\end{equation}
Although the terms $\xi(Q_n(\P_{n,b}^{(j)}))$ in~\eq{ouralg} cannot be computed analytically in general, they can be computed numerically via straightforward Monte Carlo approximation in the manner of the bootstrap: for each term $j$, repeatedly resample $n$ points i.i.d.\ from $\P_{n,b}^{(j)}$, compute the estimate on each resample, form the empirical distribution $\s{\Q}_{n,j}$ of the computed estimates, and approximate $\xi(Q_n(\P_{n,b}^{(j)})) \approx \xi(\s{\Q}_{n,j})$.

Now, to realize the substantial computational benefits afforded by \ouralgAbbrev, we utilize the following crucial fact: each \ouralgAbbrev resample, despite having nominal size $n$, contains at most $b$ distinct data points.  In particular, to generate each resample, it suffices to draw a vector of counts from an $n$-trial uniform multinomial distribution over $b$ objects.  We can then represent each resample by simply maintaining the at most $b$ distinct points present within it, accompanied by corresponding sampled counts (i.e., each resample requires only storage space in $O(b)$).  In turn, if the estimator can work directly with this weighted data representation, then the computational requirements of the estimator---with respect to both time and storage space---scale only in $b$, rather than $n$.  Fortunately, this property does indeed hold for many if not most commonly used estimators, such as general M-estimators.  The resulting \ouralgAbbrev algorithm, including Monte Carlo resampling, is shown in Algorithm~\ref{alg:ouralg}.

\begin{algorithm}[tb]
	\caption{\OurAlgWithAbbrev}
	\label{alg:ouralg}
	\DontPrintSemicolon
	\SetAlgoLined

	\KwIn{\parbox[t]{0.39\linewidth}{Data $X_1, \ldots, X_n$
							\newline $\sest{\theta}$: estimator of interest
							\newline $\xi$: estimator quality assessment}
        	 \parbox[t]{0.45\linewidth}{$b$: subset size
	 						\newline $\numsub$: number of sampled subsets
        						\newline $\numre$: number of Monte Carlo iterations}}
	\KwOut{An estimate of $\xi(Q_n(P))$}
	\BlankLine
	\For{$j \leftarrow 1$ \KwTo $\numsub$}{
		\tcp{Subsample the data}
		Randomly sample a set $\I = \{i_1, \ldots, i_b\}$ of $b$ indices from $\{1, \ldots, n\}$ without replacement\;
		[or, choose $\I$ to be a disjoint subset of size $b$ from a predefined random partition of $\{1, \ldots, n\}$]\;
		\tcp{Approximate $\xi(Q_n(\P_{n,b}^{(j)}))$}
		\For{$k \leftarrow 1$ \KwTo $\numre$}{
			Sample $(n_1, \ldots, n_b) \sim \Multinomial(n, \1_b / b)$\;
			$\s{\P}_{n,k} \leftarrow n^{-1} \sum_{a=1}^b n_a \delta_{X_{i_a}}$\;
			$\s{\sest{\theta}}_{n,k} \leftarrow \sest{\theta}(\s{\P}_{n,k})$\;
		}
		$\s{\Q}_{n,j} \leftarrow {\numre}^{-1} \sum_{k=1}^{\numre} \delta_{\s{\sest{\theta}}_{n,k}}$\;
		$\s{\xi}_{n,j} \leftarrow \xi(\s{\Q}_{n,j})$\;
	}
	\tcp{Average values of $\xi(Q_n(\P_{n,b}^{(j)}))$ computed for different data subsets}
	\Return{$\numsub^{-1} \sum_{j=1}^{\numsub} \s{\xi}_{n,j}$}
\end{algorithm}

Thus, \ouralgAbbrev only requires repeated computation on small subsets of the original dataset and avoids the bootstrap's problematic need for repeated computation of the estimate on resamples having size comparable to that of the original dataset.  A simple and standard calculation~\cite{efron-intro-bootstrap} shows that each bootstrap resample contains approximately $0.632n$ distinct points, which is large if $n$ is large.  In contrast, as discussed above, each \ouralgAbbrev resample contains at most $b$ distinct points, and $b$ can be chosen to be much smaller than $n$ or $0.632n$.  For example, we might take $b = n^\gamma$ where $\gamma \in [0.5, 1]$.  More concretely, if $n = 1,000,000$, then each bootstrap resample would contain approximately $632,000$ distinct points, whereas with $b = n^{0.6}$ each \ouralgAbbrev subsample and resample would contain at most $3,981$ distinct points.  If each data point occupies 1~MB of storage space, then the original dataset would occupy 1~TB, a bootstrap resample would occupy approximately 632~GB, and each \ouralgAbbrev subsample or resample would occupy at most 4~GB.  As a result, the cost of computing the estimate on each \ouralgAbbrev resample is generally substantially lower than the cost of computing the estimate on each bootstrap resample, or on the full dataset.  Furthermore, as we show in our simulation study and scalability experiments below, \ouralgAbbrev typically requires less total computation (across multiple data subsets and resamples) than the bootstrap to reach comparably high accuracy; fairly modest values of $\numsub$ and $\numre$ suffice.

Due to its much smaller subsample and resample sizes, \ouralgAbbrev is also significantly more amenable than the bootstrap to distribution of different subsamples and resamples and their associated computations to independent compute nodes; therefore, \ouralgAbbrev allows for simple distributed and parallel implementations, enabling additional large computational gains.  In the large data setting, computing a single full-data point estimate often requires simultaneous distributed computation across multiple compute nodes, among which the observed dataset is partitioned.  Given the large size of each bootstrap resample, computing the estimate on even a single such resample in turn also requires the use of a comparably large cluster of compute nodes; the bootstrap requires repetition of this computation for multiple resamples.  Each computation of the estimate is thus quite costly, and the aggregate computational costs of this repeated distributed computation are quite high  (indeed, the computation for each bootstrap resample requires use of an entire cluster of compute nodes and incurs the associated overhead).

In contrast, \ouralgAbbrev straightforwardly permits computation on multiple (or even all) subsamples and resamples simultaneously in parallel: because \ouralgAbbrev subsamples and resamples can be significantly smaller than the original dataset, they can be transferred to, stored by, and processed on individual (or very small sets of) compute nodes.  For example, we could naturally leverage modern hierarchical distributed architectures by distributing subsamples to different compute nodes and subsequently using intra-node parallelism to compute across different resamples generated from the same subsample.  Thus, relative to the bootstrap, \ouralgAbbrev both decreases the total computational cost of assessing estimator quality and allows more natural use of parallel and distributed computational resources.  Moreover, even if only a single compute node is available, \ouralgAbbrev allows the following somewhat counterintuitive possibility: even if it is prohibitive to actually compute a point estimate for the full observed data using a single compute node (because the full dataset is large), it may still be possible to efficiently assess such a point estimate's quality using only a single compute node by processing one subsample (and the associated resamples) at a time.

Returning to equation~\eq{ouralg}, unlike the plugin approximation $\xi(Q_n(\P_n))$ used by the bootstrap, the plugin approximations $\xi(Q_n(\P_{n,b}^{(j)}))$ used by \ouralgAbbrev are based on empirical distributions $\P_{n,b}^{(j)}$ which are more compact and hence, as we have seen, less computationally demanding than the full empirical distribution $\P_n$.  However, each $\P_{n,b}^{(j)}$ is inferior to $\P_n$ as an approximation to the true underlying distribution $P$, and so \ouralgAbbrev averages across multiple different realizations of $\P_{n,b}^{(j)}$ to improve the quality of the final result.  This procedure yields significant computational benefits over the bootstrap (as discussed above and demonstrated empirically in Section~\ref{sec:scalability}), while having the same generic applicability and favorable statistical properties as the bootstrap (as shown in the next section), in addition to being more robust than the $m$ out of $n$ bootstrap and subsampling to the choice of subset size (see our simulation study below).

\section{Statistical Performance}
\label{sec:stat-perform}

\subsection{Consistency and Higher-Order Correctness}
\label{sec:theory}

We now show that \ouralgAbbrev has statistical properties---in particular, asymptotic consistency and higher-order correctness---which are identical to those of the bootstrap, under the same conditions that have been used in prior analysis of the bootstrap.  Note that if $\sest{\theta}_n$ is consistent (i.e., approaches $\theta$ in probability) as $n \rightarrow \infty$, then it has a degenerate limiting distribution.  Thus, in studying the asymptotics of the bootstrap and related procedures, it is typical to assume that $\xi$ manipulates the distribution of a centered and normalized version of $\sest{\theta}_n$ (though this distribution is still determined by $Q_n(P)$ and $P$).
Additionally, as in standard analyses of the bootstrap, we do not explicitly account here for error introduced by use of Monte Carlo approximation to compute the individual plugin approximations $\xi(Q_n(\P_{n,b}^{(j)}))$.

The following theorem states that (under standard assumptions) as $b, n \rightarrow \infty$, the estimates $\numsub^{-1} \sum_{j=1}^{\numsub} \xi(Q_n(\P_{n,b}^{(j)}))$ returned by \ouralgAbbrev approach the population value $\xi(Q_n(P))$ in probability.  Interestingly, the only assumption about $b$ required for this result is that $b \rightarrow \infty$, though in practice we would generally take $b$ to be a slowly growing function of $n$.

\begin{theorem}
\label{thm:consistency}
Suppose that $\sest{\theta}_n = \phi(\P_n)$ and $\theta = \phi(P)$, where $\phi$ is Hadamard differentiable at $P$ tangentially to some subspace, with $P$, $\P_n$, and $\P_{n,b}^{(j)}$ viewed as maps from some Donsker class $\F$ to $\R$ such that $\F_\delta$ is measurable for every $\delta > 0$, where $\F_{\delta} = \{f-g:f,g\in\F,\rho_P(f-g)<\delta\}$ and $\rho_P(f)=\left(P(f-Pf)^2\right)^{1/2}$.  Additionally, assume that $\xi(Q_n(P))$ is a function of the distribution of $\sqrt{n}(\phi(\P_n) - \phi(P))$ which is continuous in the space of such distributions with respect to a metric that metrizes weak convergence.  Then,
$$\numsub^{-1} \sum_{j=1}^{\numsub} \xi(Q_n(\P_{n,b}^{(j)})) - \xi(Q_n(P))\pconv 0$$
as $n \rightarrow \infty$, for any sequence $b \rightarrow \infty$ and for any fixed $\numsub$.
\end{theorem}

See the appendix for a proof of this theorem, as well as for proofs of all other results in this section.  Note that the assumptions of Theorem~\ref{thm:consistency} are standard in analysis of the bootstrap and in fact hold in many practically interesting cases.  For example, M-estimators are generally Hadamard differentiable (under some regularity conditions)~\cite{vdv, vdv-wellner}, and the assumptions on $\xi$ are satisfied if, for example, $\xi$ computes a cdf value.  Theorem~\ref{thm:consistency} can also be generalized to hold for sequences $\numsub \rightarrow \infty$ and more general forms of $\xi$, but such generalization appears to require stronger assumptions, such as uniform integrability of the $\xi(Q_n(\P_{n,b}^{(j)}))$; the need for stronger assumptions in order to obtain more general consistency results has also been noted in prior work on the bootstrap (e.g., see~\citet{hahn}).

Moving beyond analysis of the asymptotic consistency of \ouralgAbbrev, we now characterize its higher-order correctness (i.e., the rate of convergence of its output to $\xi(Q_n(P))$).  A great deal of prior work has been devoted to showing that the bootstrap is higher-order correct in many cases (e.g., see the seminal book by~\citet{hall-edgeworth}), meaning that it converges to the true value $\xi(Q_n(P))$ at a rate of $O_P(1/n)$ or faster.  In contrast, methods based on analytical asymptotic approximation are generally correct only at order $O_P(1/\sqrt{n})$.  The bootstrap converges more quickly due to its more data-driven nature, which allows it to better capture finite-sample deviations of the distribution of $\sest{\theta}_n$ from its asymptotic limiting distribution.

As shown by the following theorem, \ouralgAbbrev shares the same degree of higher-order correctness as the bootstrap, assuming that $\numsub$ and $b$ are chosen to be sufficiently large.  Importantly, sufficiently large values of $b$ here can still be significantly smaller than $n$, with $b/n \rightarrow 0$ as $n \rightarrow \infty$.  Following prior analyses of the bootstrap, we now make the standard assumption that $\xi$ can be represented via an asymptotic series expansion in powers of $1/\sqrt{n}$.  In fact, prior work provides such expansions in a variety of settings.  When $\xi$ computes a cdf value, these expansions are termed Edgeworth expansions; if $\xi$ computes a quantile, then the relevant expansions are Cornish-Fisher expansions.  See~\citet{hall-edgeworth} for a full development of such expansions both in generality as well as for specific forms of the estimator, including smooth functions of mean-like statistics and curve estimators.

\begin{theorem}
\label{thm:higher-order}
Suppose that $\xi(Q_n(P))$ admits an expansion as an asymptotic series
\begin{equation}
\label{eq:pop-expansion}
\xi(Q_n(P)) = z + \frac{p_1}{\sqrt{n}} + \cdots + \frac{p_k}{n^{k/2}} + o\left(\frac{1}{n^{k/2}}\right),
\end{equation}
where $z$ is a constant independent of $P$ and the $p_k$ are polynomials in the moments of $P$.  Additionally, assume that the empirical version of $\xi(Q_n(P))$ for any $j$ admits a similar expansion
\begin{equation}
\label{eq:sample-expansion}
\xi(Q_n(\P_{n,b}^{(j)})) = z + \frac{\sest{p}^{(j)}_1}{\sqrt{n}} + \cdots + \frac{\sest{p}^{(j)}_k}{n^{k/2}} + o_P\left(\frac{1}{n^{k/2}}\right),
\end{equation}
where $z$ is as defined above and the $\sest{p}^{(j)}_k$ are polynomials in the moments of $\P_{n,b}^{(j)}$ obtained by replacing the moments of $P$ in the $p_k$ with those of $\P_{n,b}^{(j)}$.  Then, assuming that $b \leq n$ and $E(\sest{p}^{(1)}_k)^2 < \infty$ for $k \in \{1,2\}$,
\begin{equation}
\label{eq:higher-order-result}
\left| \numsub^{-1} \sum_{j=1}^{\numsub} \xi(Q_n(\P_{n,b}^{(j)})) - \xi(Q_n(P)) \right| = O_P\left(\frac{\sqrt{\Var(\sest{p}^{(1)}_k - p_k | \P_n)}}{\sqrt{n\numsub}}\right) + O_P\left(\frac{1}{n}\right) + O\left(\frac{1}{b\sqrt{n}}\right).
\end{equation}
Therefore, taking $\numsub = \Omega(n \Var(\sest{p}^{(1)}_k - p_k | \P_n))$ and $b = \Omega(\sqrt{n})$ yields
$$\left| \numsub^{-1} \sum_{j=1}^{\numsub} \xi(Q_n(\P_{n,b}^{(j)})) - \xi(Q_n(P)) \right| = O_P\left(\frac{1}{n}\right),$$
in which case \ouralgAbbrev enjoys the same level of higher-order correctness as the bootstrap.
\end{theorem}

Note that it is natural to assume above that $\xi(Q_n(\P_{n,b}^{(j)}))$ can be expanded in powers of $1/\sqrt{n}$, rather than $1/\sqrt{b}$, because $Q_n(\P_{n,b}^{(j)})$ is the distribution of the estimate computed on $n$ points sampled from $\P_{n,b}^{(j)}$.  The fact that only $b$ points are represented in $\P_{n,b}^{(j)}$ enters via the $\sest{p}^{(j)}_k$, which are polynomials in the sample moments of those $b$ points.

Theorem~\ref{thm:higher-order} indicates that, like the bootstrap, \ouralgAbbrev can converge at rate $O_P(1/n)$ (assuming that $\numsub$ and $b$ grow at a sufficient rate).  Additionally, because $\Var(\sest{p}^{(1)}_k - p_k | \P_n)$ is decreasing in probability as $b$ and $n$ increase, $\numsub$ can grow significantly more slowly than $n$ (indeed, unconditionally, $\sest{p}^{(j)}_k - p_k = O_P(1/\sqrt{b})$).  While $\Var(\sest{p}^{(1)}_k - p_k | \P_n)$ can in principle be computed given an observed dataset, as it depends only on $\P_n$ and the form of the estimator under consideration, we can also obtain a general upper bound (in probability) on the rate of decrease of this conditional variance:
\begin{remark}
\label{remark:condvar-rate}
Assuming that $E(\sest{p}^{(1)}_k)^4 < \infty$, $\Var(\sest{p}^{(1)}_k - p_k | \P_n) = O_P(1/\sqrt{n}) + O(1/b).$
\end{remark}

The following result, which applies to the alternative variant of \ouralgAbbrev that constrains the $\numsub$ randomly sampled subsets to be disjoint, also highlights the fact that $s$ can grow substantially more slowly than $n$:

\begin{theorem}
\label{thm:higher-order-disjoint}
Under the assumptions of Theorem~\ref{thm:higher-order}, and assuming that \ouralgAbbrev uses disjoint random subsets of the observed data (rather than simple random subsamples), we have
\begin{equation}
\label{eq:higher-order-disjoint-result}
\left| \numsub^{-1} \sum_{j=1}^{\numsub} \xi(Q_n(\P_{n,b}^{(j)})) - \xi(Q_n(P)) \right| = O_P\left(\frac{1}{\sqrt{nb\numsub}}\right) + O\left(\frac{1}{b\sqrt{n}}\right).
\end{equation}
Therefore, if $\numsub \sim (n/b)$ and $b = \Omega(\sqrt{n})$, then
$$\left| \numsub^{-1} \sum_{j=1}^{\numsub} \xi(Q_n(\P_{n,b}^{(j)})) - \xi(Q_n(P)) \right| = O_P\left(\frac{1}{n}\right),$$
in which case \ouralgAbbrev enjoys the same level of higher-order correctness as the bootstrap.
\end{theorem}

Finally, while the assumptions of the two preceding theorems generally require that $\xi$ studentizes the estimator under consideration (which involves dividing by an estimate of standard error), similar results hold even if the estimator is not studentized.  In particular, not studentizing slows the convergence rate of both the bootstrap and \ouralgAbbrev by the same factor, generally causing the loss of a factor of $O_P(1/\sqrt{n})$~\cite{vdv}.

\subsection{Simulation Study}
\label{sec:simulation}

We investigate empirically the statistical performance characteristics of \ouralgAbbrev and compare to the statistical performance of existing methods via experiments on simulated data.  Use of simulated data is necessary here because it allows knowledge of $P$, $Q_n(P)$, and hence $\xi(Q_n(P))$; this ground truth is required for evaluation of statistical correctness.  For different datasets and estimation tasks, we study the convergence properties of \ouralgAbbrev as well as the bootstrap, the $m$ out of $n$ bootstrap, and subsampling.

We consider two different settings: regression and classification.  For both settings, the data have the form $X_i = (\tilde{X}_i, Y_i) \sim P$, i.i.d.\ for $i=1, \ldots, n$, where $\tilde{X}_i \in \R^d$; $Y_i \in \R$ for regression, whereas $Y_i \in \{0,1\}$ for classification.  In each case, $\sest{\theta}_n$ estimates a parameter vector in $\R^d$ for a linear or generalized linear model of the mapping between $\tilde{X}_i$ and $Y_i$.  We define $\xi$ as a procedure that computes a set of marginal 95\% confidence intervals, one for each element of the estimated parameter vector.  In particular, given an estimator's sampling distribution $Q$ (or an approximation thereof), $\xi$ computes the boundaries of the relevant confidence intervals as the 2.5th and 97.5th percentiles of the marginal component-wise distributions defined by $Q$ (averaging across $\xi$'s simply consists of averaging these percentile estimates).

To evaluate the various quality assessment procedures on a given estimation task and true underlying data distribution $P$, we first compute the ground truth $\xi(Q_n(P))$ by generating $2,000$ realizations of datasets of size $n$ from $P$, computing $\sest{\theta}_n$ on each, and using this collection of $\sest{\theta}_n$'s to form a high-fidelity approximation to $Q_n(P)$.  Then, for an independent dataset realization of size $n$ from the true underlying distribution, we run each quality assessment procedure (without parallelization) until it converges and record the estimate of $\xi(Q_n(P))$ produced after each iteration (e.g., after each bootstrap resample or \ouralgAbbrev subsample is processed), as well as the cumulative processing time required to produce that estimate.  Every such estimate is evaluated based on the average (across dimensions) relative deviation of its component-wise confidence intervals' widths from the corresponding true widths; given an estimated confidence interval width $c$ and a true width $c_o$, the relative deviation of $c$ from $c_o$ is defined as $|c - c_o|/c_o$.  We repeat this process on five independent dataset realizations of size $n$ and average the resulting relative errors and corresponding processing times across these five datasets to obtain a trajectory of relative error versus time for each quality assessment procedure.  The relative errors' variances are small relative to the relevant differences between their means, and so these variances are not shown in our plots.  Note that we evaluate based on confidence interval widths, rather than coverage probabilities, to control the running times of our experiments: in our experimental setting, even a single run of a quality assessment procedure requires non-trivial time, and computing coverage probabilities would require a large number of such runs.  All experiments in this section were implemented and executed using MATLAB on a single processor.  To maintain consistency of notation, we refer to the $m$ out of $n$ bootstrap as the $b$ out of $n$ bootstrap throughout the remainder of this section.  For \ouralgAbbrev, the $b$ out of $n$ bootstrap, and subsampling, we consider $b = n^\gamma$ with $\gamma \in \{0.5, 0.6, 0.7, 0.8, 0.9\}$; we use $\numre = 100$ in all runs of \ouralgAbbrev.

In the regression setting, we generate each dataset from a true underlying distribution $P$ consisting of either a linear model $Y_i = \tilde{X}_i^T \1_d + \epsilon_i$ or a model $Y_i = \tilde{X}_i^T \1_d + \tilde{X}_i^T \tilde{X}_i + \epsilon_i$ having a quadratic term, with $d=100$ and $n=20,000$.  The $\tilde{X}_i$ and $\epsilon_i$ are drawn independently from one of the following pairs of distributions: $\tilde{X}_i \sim \Normal(0, I_d)$ with $\epsilon_i \sim \Normal(0, 10)$; $\tilde{X}_{i,j} \sim \StudentT(3)$ i.i.d.\ for $j=1, \ldots, d$ with $\epsilon_i \sim \Normal(0, 10)$; or $\tilde{X}_{i,j} \sim \GammaDist(1 + 5(j-1)/\max(d-1,1), 2) - 2[1 + 5(j-1)/\max(d-1,1), 2]$ independently for $j=1, \ldots, d$ with $\epsilon_i \sim \GammaDist(1,2) - 2$.  All of these distributions have $E\tilde{X}_i = E\epsilon_i = 0$, and the last $\tilde{X}_i$ distribution has non-zero skewness which varies among the dimensions.  In the regression setting under both the linear and quadratic data generating distributions, our estimator $\sest{\theta}_n$ consists of a linear (in $\tilde{X}_i$) least squares regression with a small $L_2$ penalty on the parameter vector to encourage numerical stability (we set the weight on this penalty term to $10^{-5}$).  The true average (across dimensions) marginal confidence interval width for the estimated parameter vector is approximately 0.1 under the linear data generating distributions (for all $\tilde{X_i}$ distributions) and approximately 1 under the quadratic data generating distributions.

\begin{figure*}
\includegraphics[width=0.32\linewidth]{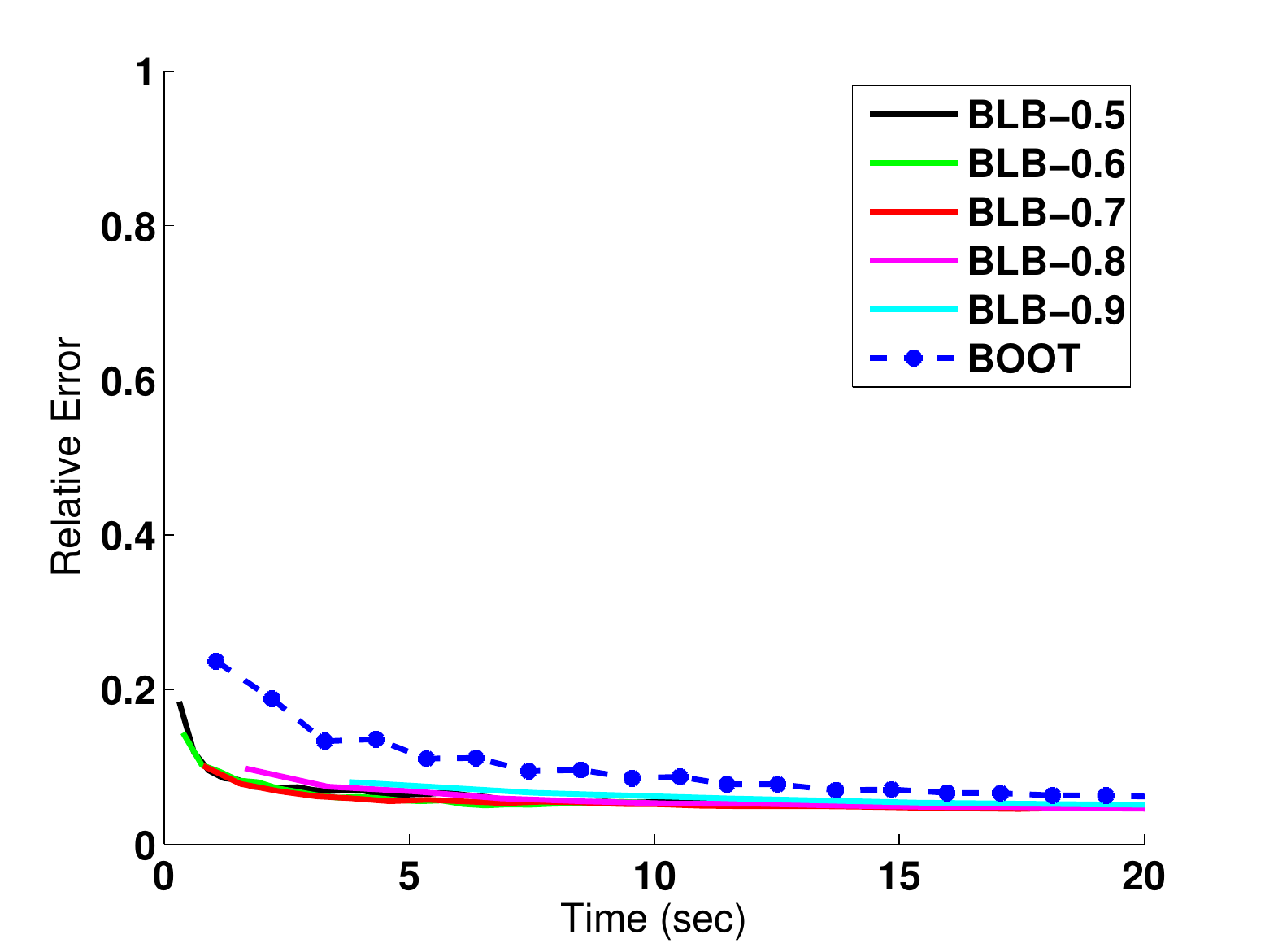}
\includegraphics[width=0.32\linewidth]{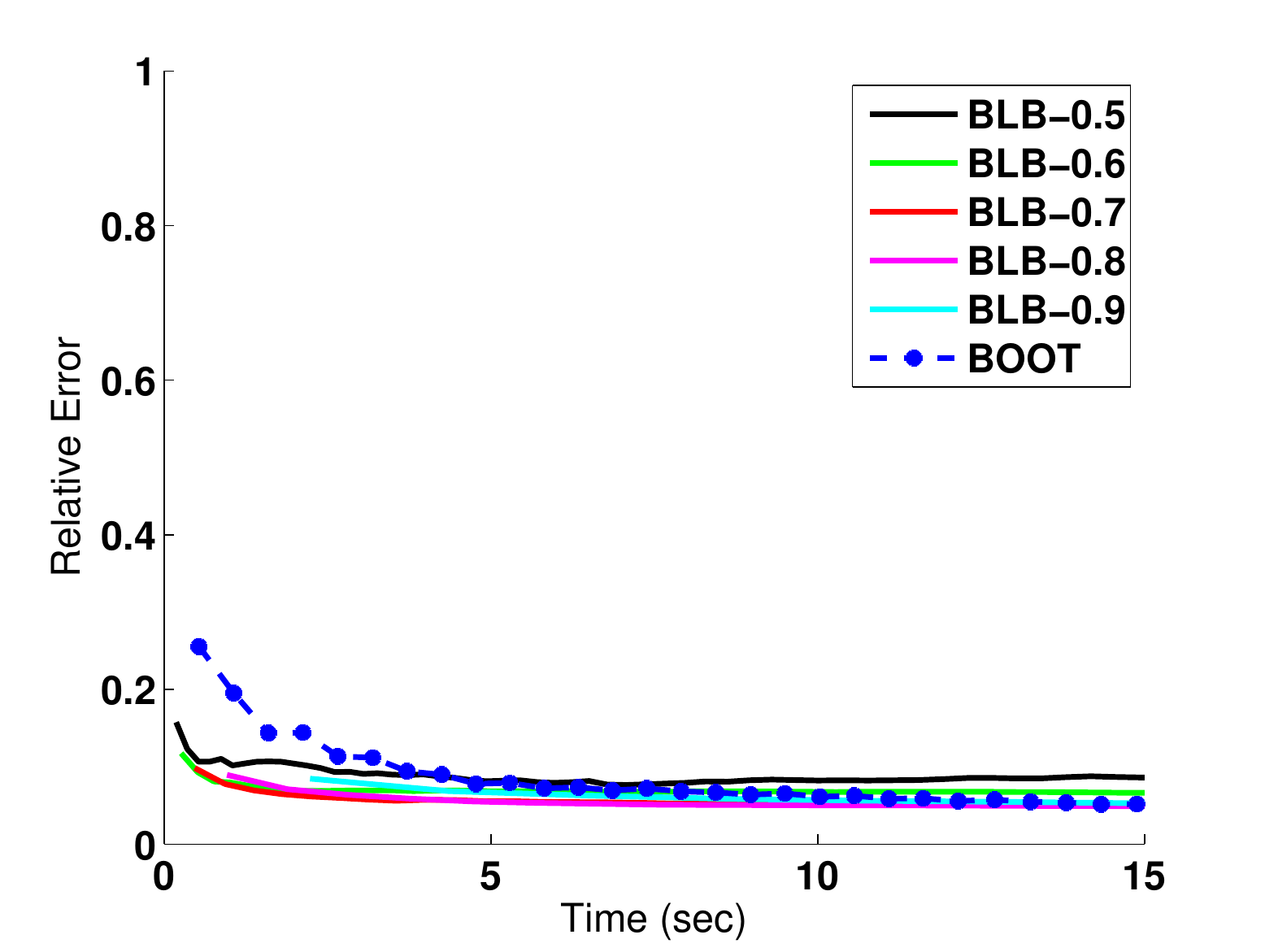}
\includegraphics[width=0.32\linewidth]{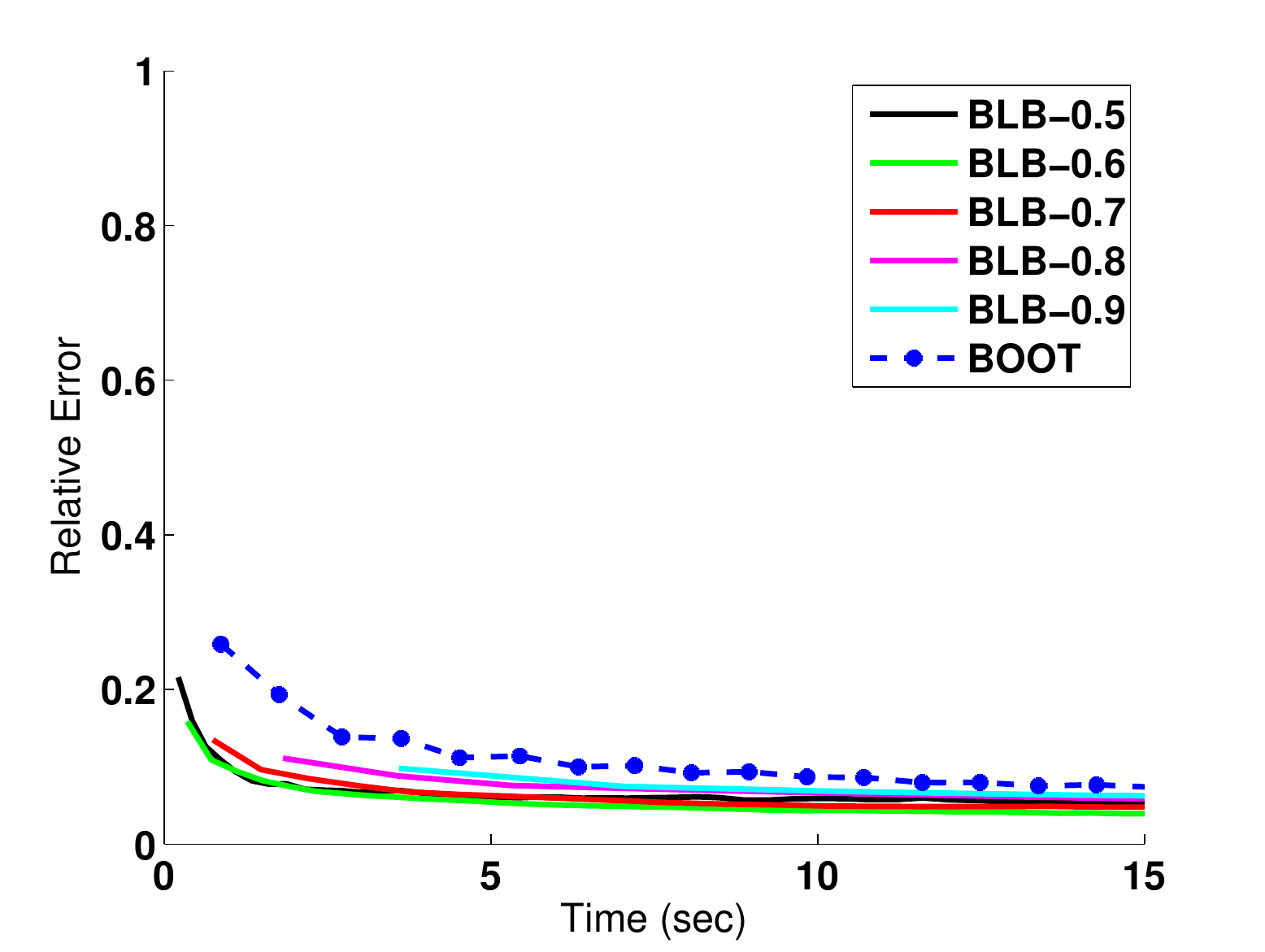}
\\
\includegraphics[width=0.32\linewidth]{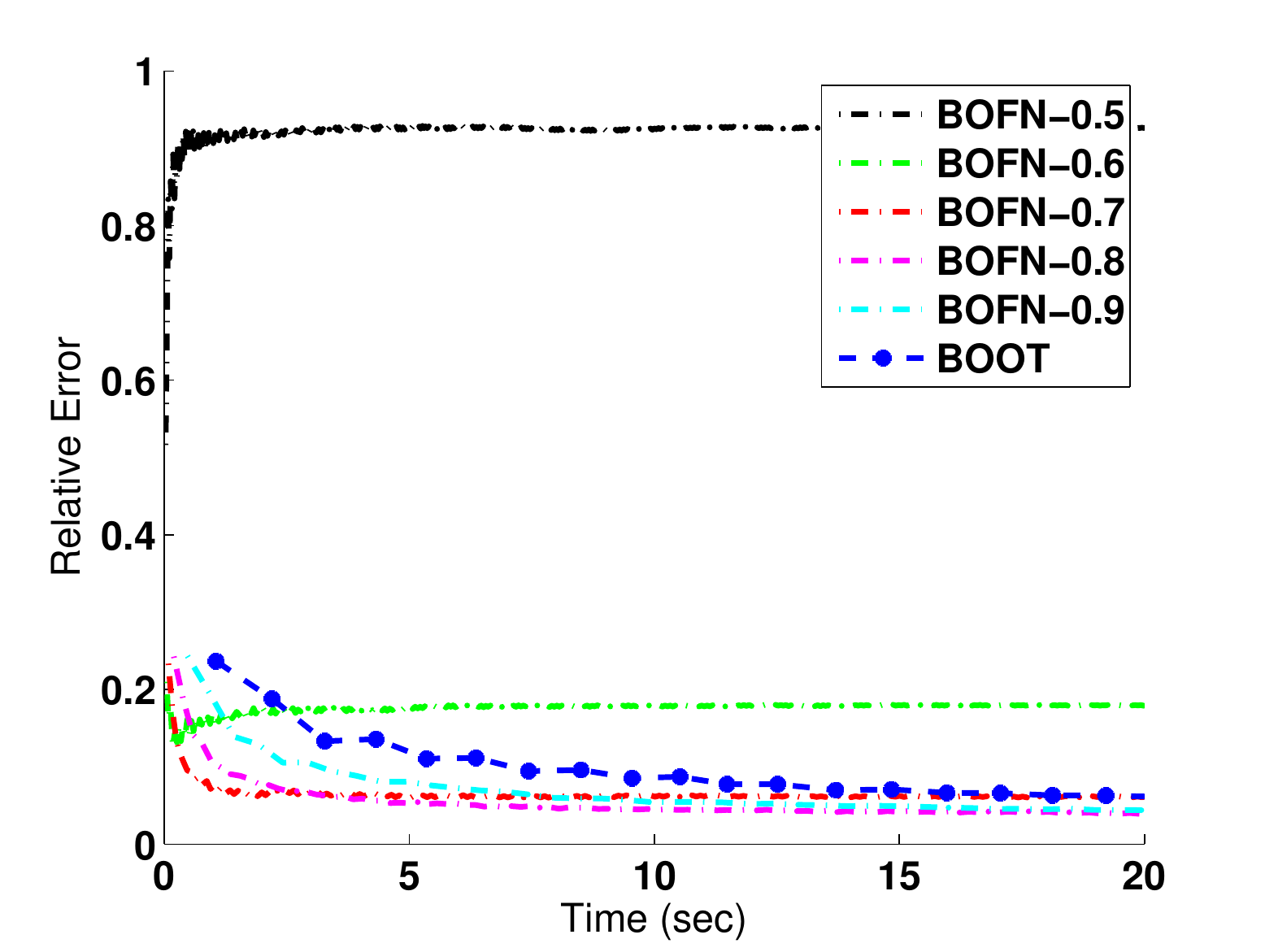}
\includegraphics[width=0.32\linewidth]{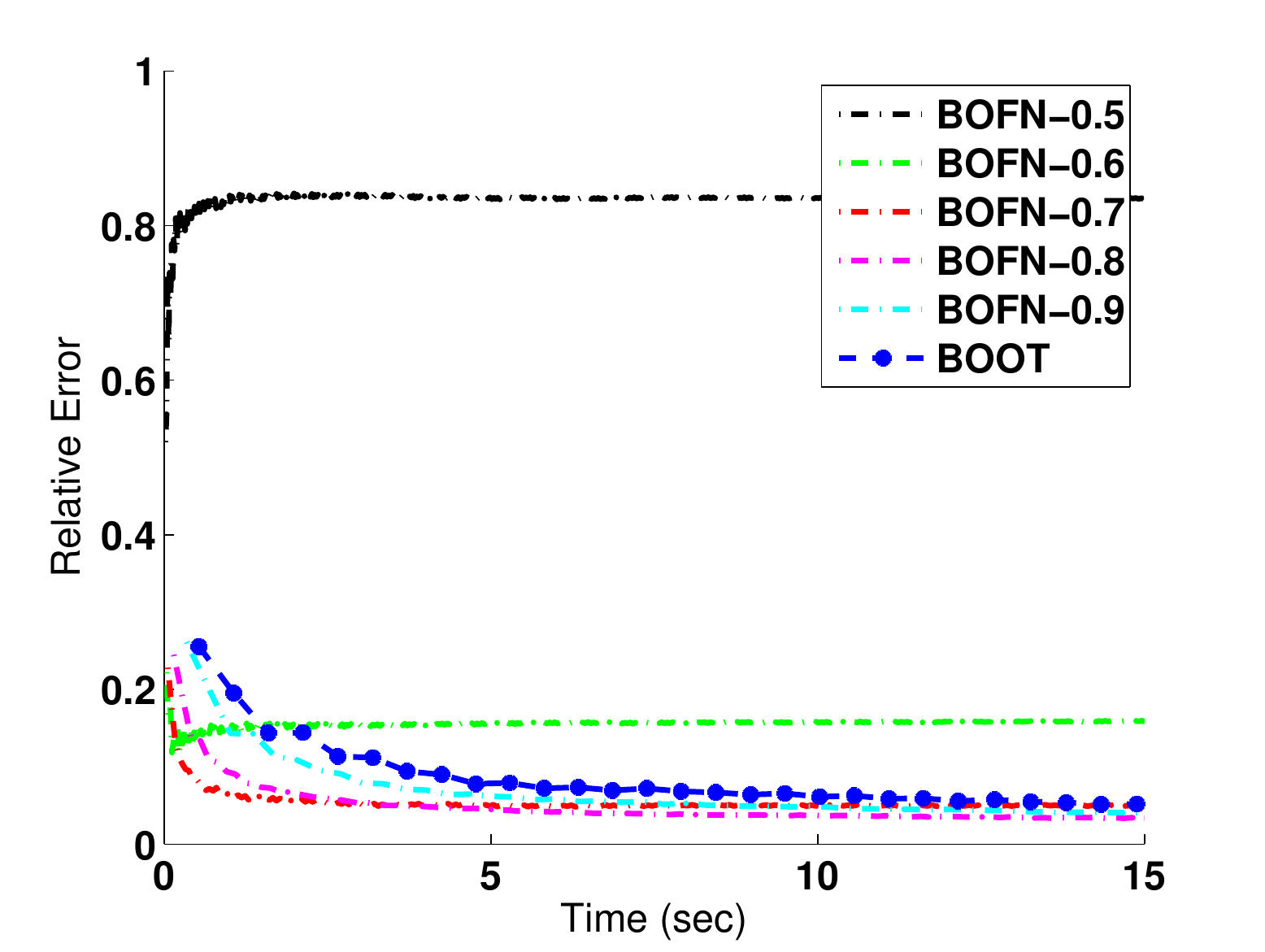}
\includegraphics[width=0.32\linewidth]{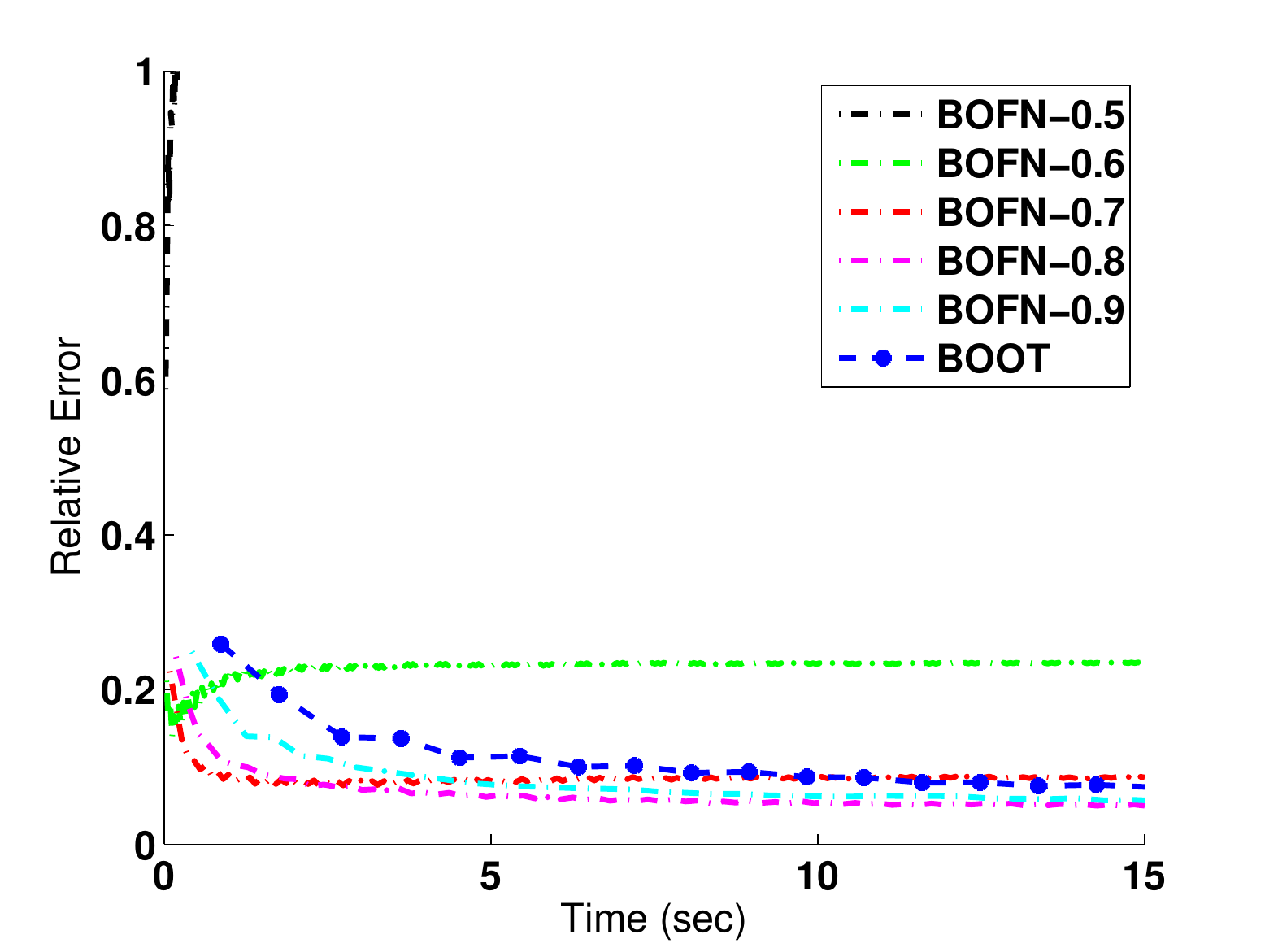}
\\
\includegraphics[width=0.32\linewidth]{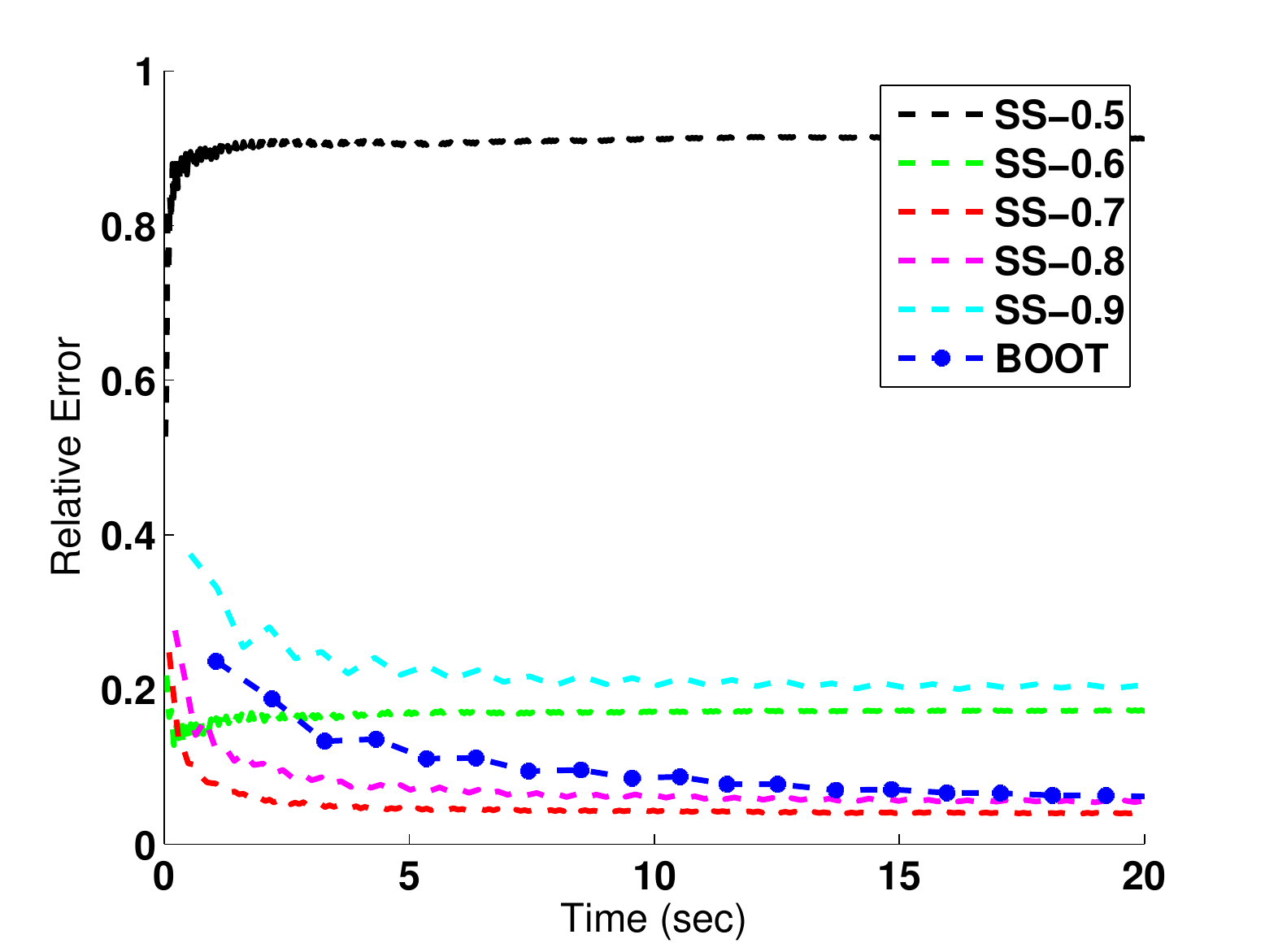}
\includegraphics[width=0.32\linewidth]{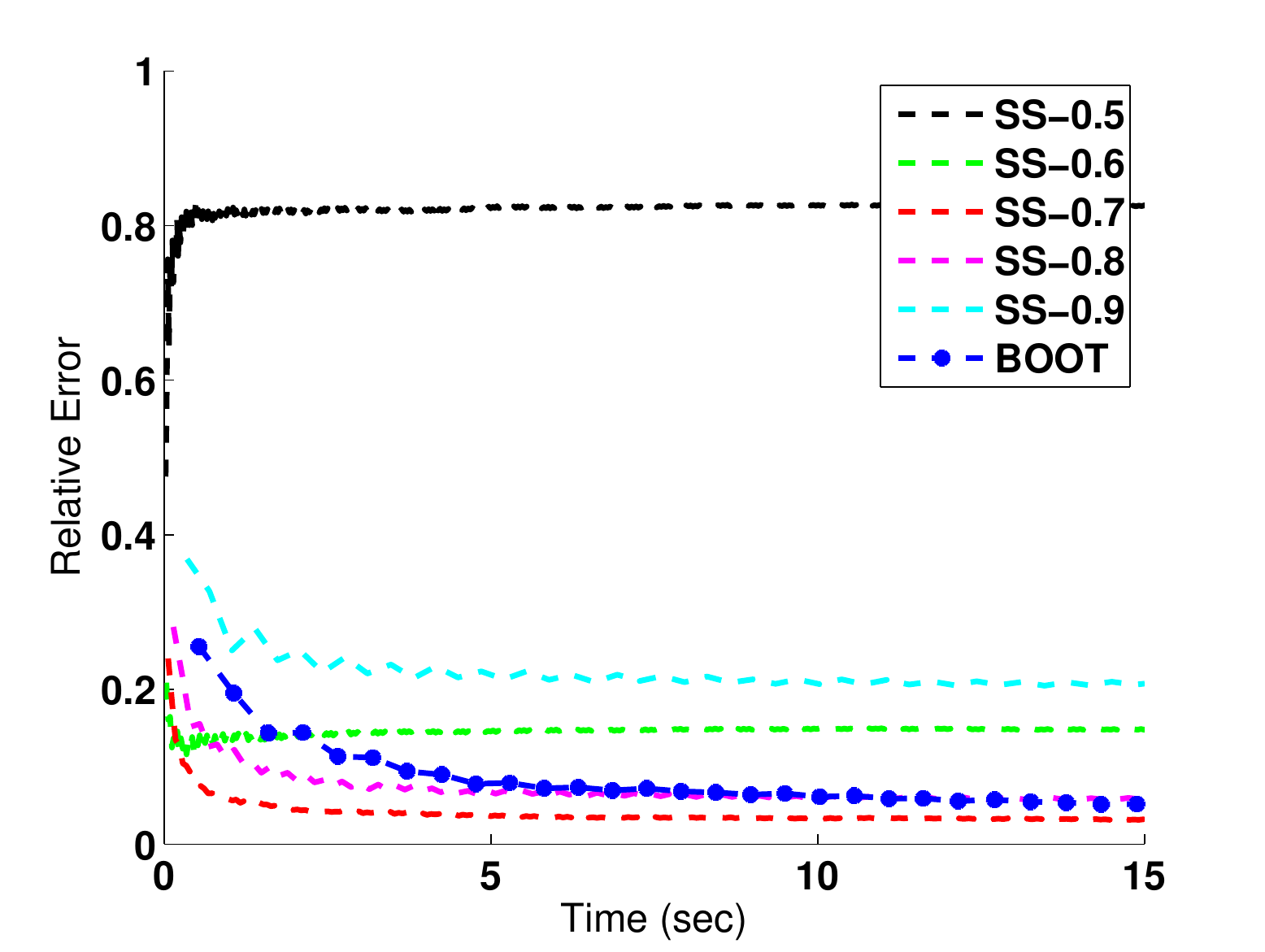}
\includegraphics[width=0.32\linewidth]{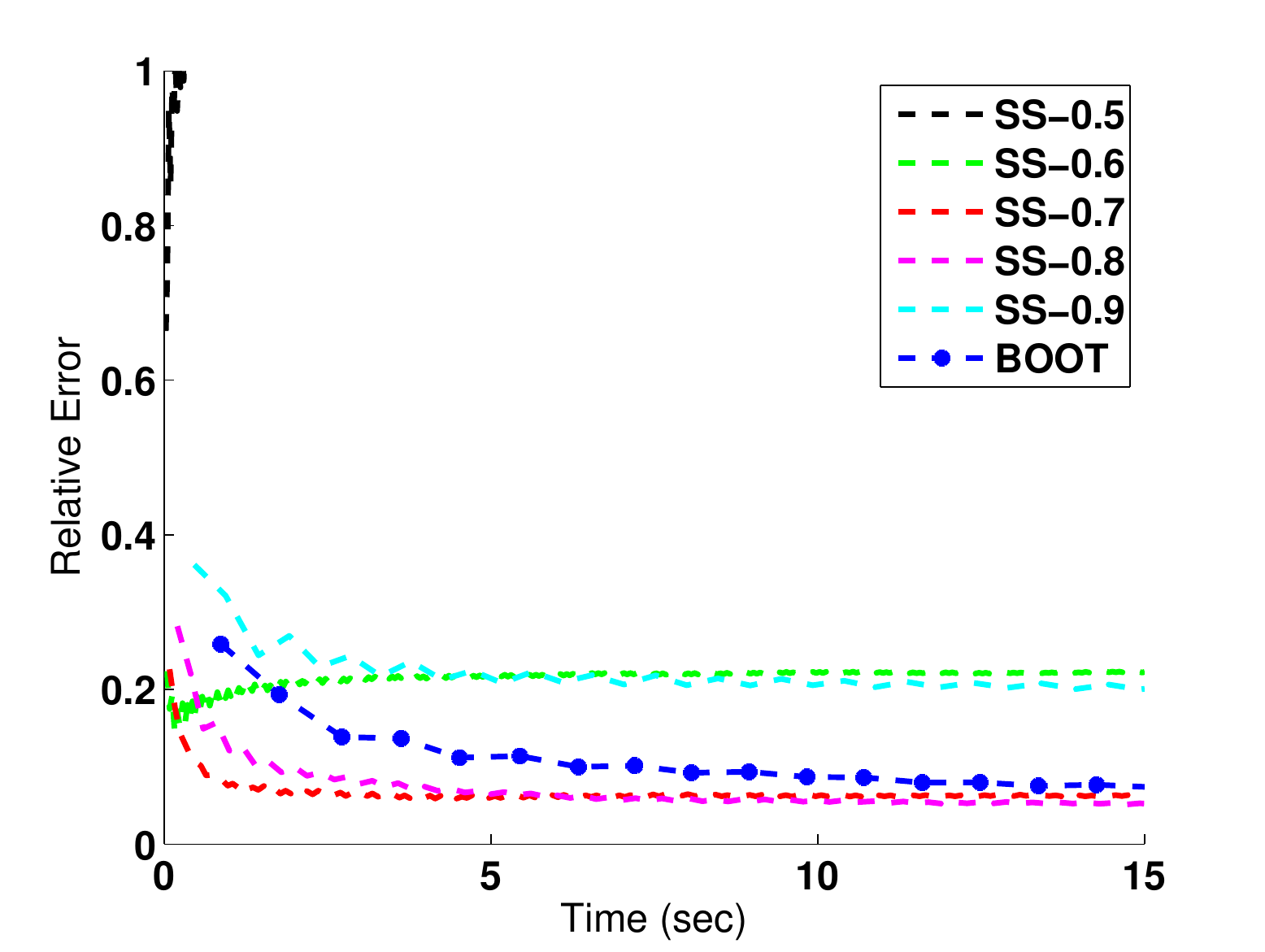}
\caption{Relative error vs.\ processing time for regression setting.  The top row shows \ouralgAbbrev with bootstrap (BOOT), the middle row shows $b$ out of $n$ bootstrap (BOFN), and the bottom row shows subsampling (SS).  For \ouralgAbbrev, BOFN, and SS, $b=n^\gamma$ with the value of $\gamma$ for each trajectory given in the legend.  The left column shows results for linear regression with linear data generating distribution and $\GammaDist$ $\tilde{X}_i$ distribution.  The middle column shows results for linear regression with quadratic data generating distribution and $\GammaDist$ $\tilde{X}_i$ distribution.  The right column shows results for linear regression with linear data generating distribution and $\StudentT$ $\tilde{X}_i$ distribution.}
\label{fig:linreg}
\end{figure*}

Figure~\ref{fig:linreg} shows results for the regression setting under the linear and quadratic data generating distributions with the $\GammaDist$ and $\StudentT$ $\tilde{X}_i$ distributions; similar results hold for the $\Normal$ $\tilde{X}_i$ distribution.  In all cases, \ouralgAbbrev (top row) succeeds in converging to low relative error significantly more quickly than the bootstrap, for all values of $b$ considered.  In contrast, the $b$ out of $n$ bootstrap (middle row) fails to converge to low relative error for smaller values of $b$ (below $n^{0.7}$).  Additionally, subsampling (bottom row) performs strictly worse than the $b$ out of $n$ bootstrap, as subsampling fails to converge to low relative error for both smaller and larger values of $b$ (e.g., for $b = n^{0.9}$).  Note that fairly modest values of $\numsub$ suffice for convergence of \ouralgAbbrev (recall that $\numsub$ values are implicit in the time axes of our plots), with $\numsub$ at convergence ranging from 1-2 for $b=n^{0.9}$ up to 10-14 for $b=n^{0.5}$, in the experiments shown in Figure~\ref{fig:linreg}; larger values of $\numsub$ are required for smaller values of $b$, which accords with both intuition and our theoretical analysis.  Under the quadratic data generating distribution with $\StudentT$ $\tilde{X}_i$ distribution (plots not shown), none of the procedures (including the bootstrap) converge to low relative error, which is unsurprising given the $\StudentT(3)$ distribution's lack of moments beyond order two.

For the classification setting, we generate each dataset considered from either a linear model $Y_i \sim \Bernoulli((1 + \exp(-\tilde{X}_i^T \1))^{-1})$ or a model $Y_i \sim \Bernoulli((1 + \exp(-\tilde{X}_i^T \1 - \tilde{X}_i^T \tilde{X}_i))^{-1})$ having a quadratic term, with $d=10$.  We use the three different distributions on $\tilde{X}_i$ defined in the regression setting.  Our estimator, under both the linear and quadratic data generating distributions, consists of a linear (in $\tilde{X}_i$) logistic regression fit via Newton's method, again using an $L_2$ penalty term with weight $10^{-5}$ to encourage numerical stability.  For this estimation task with $n=20,000$, the true average (across dimensions) marginal confidence interval width for the estimated parameter vector is approximately 0.1 under the linear data generating distributions (for all $\tilde{X_i}$ distributions) and approximately 0.02 under the quadratic data generating distributions.

\begin{figure*}
\includegraphics[width=0.32\linewidth]{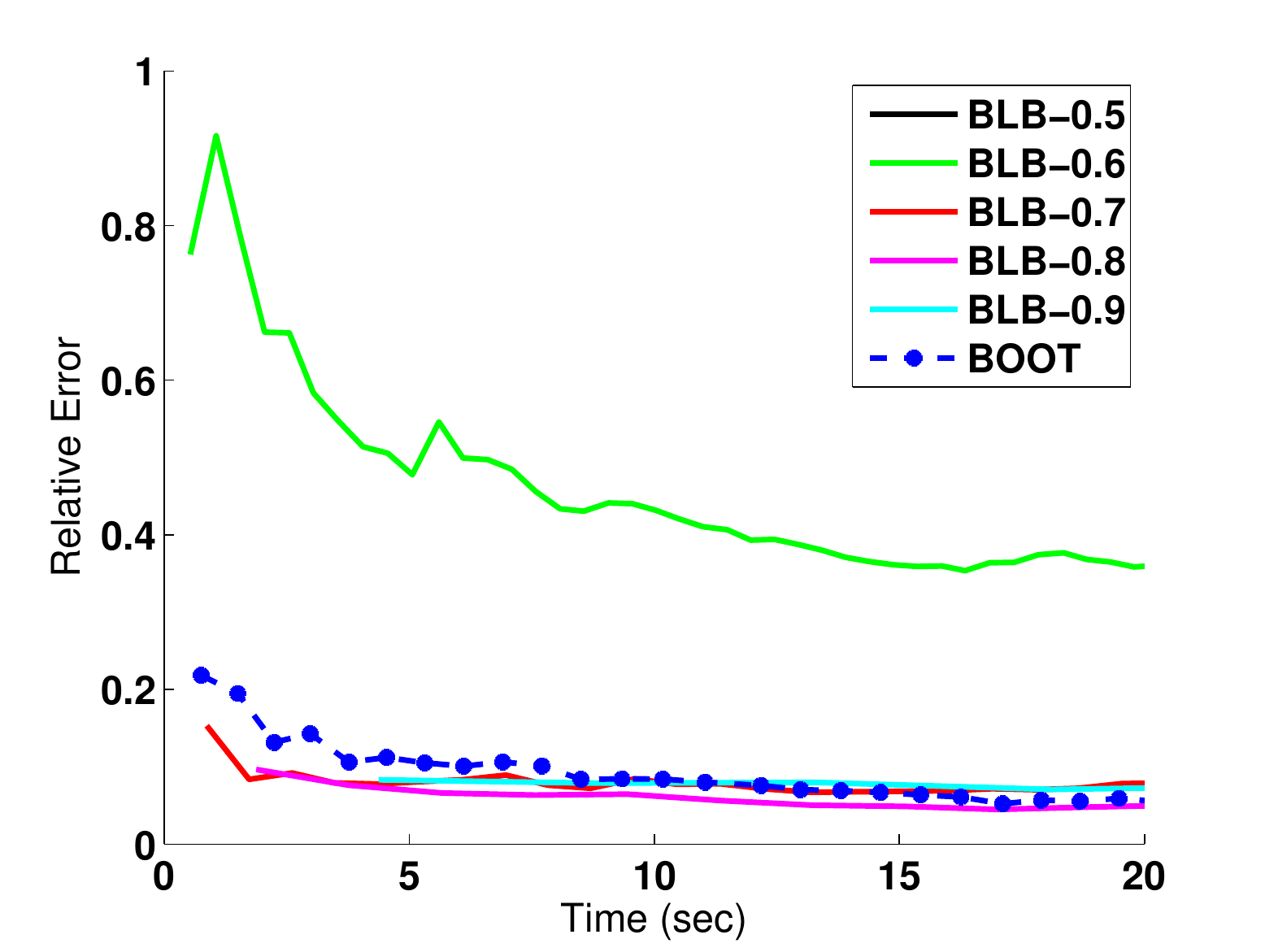}
\includegraphics[width=0.32\linewidth]{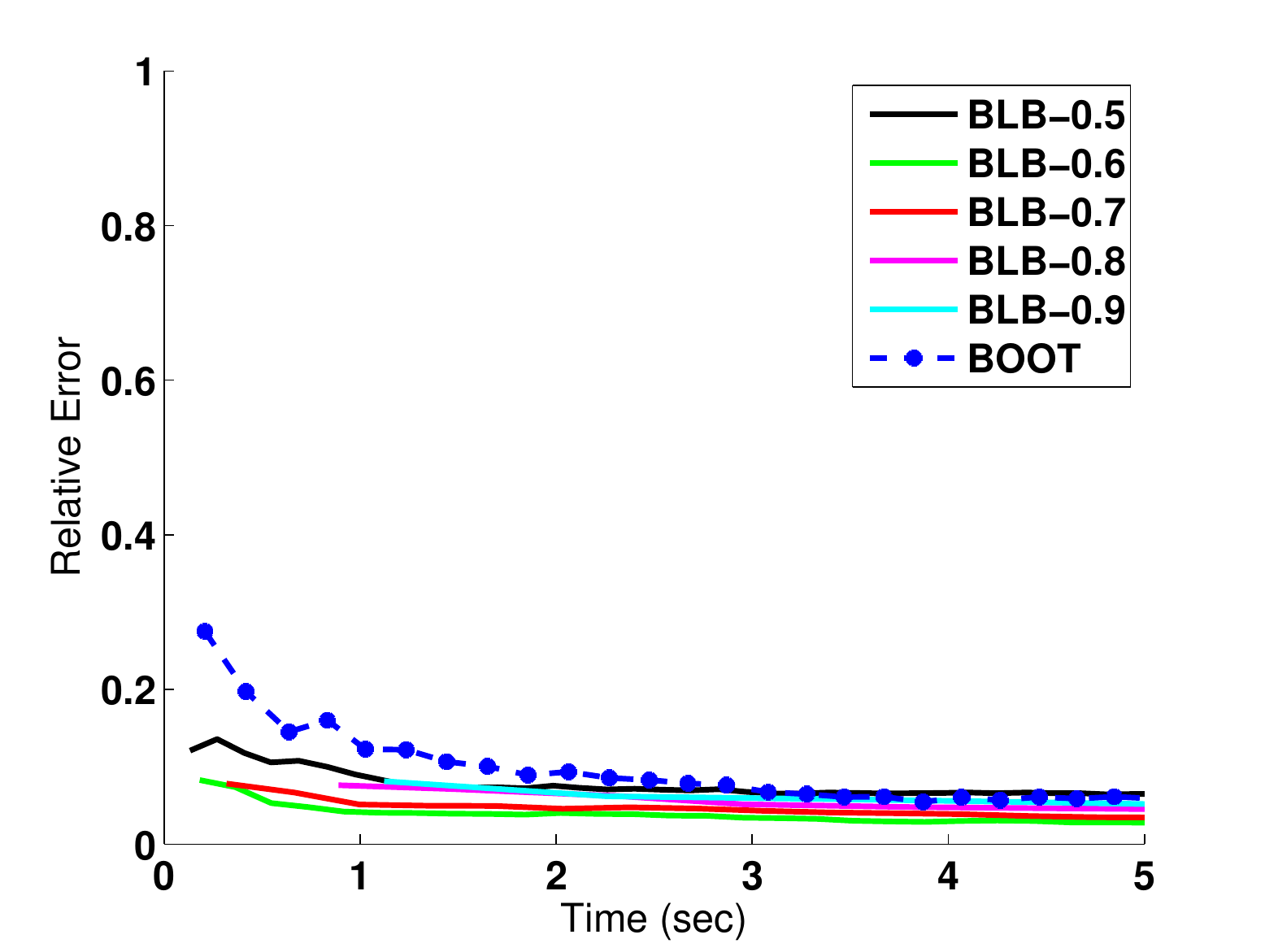}
\includegraphics[width=0.32\linewidth]{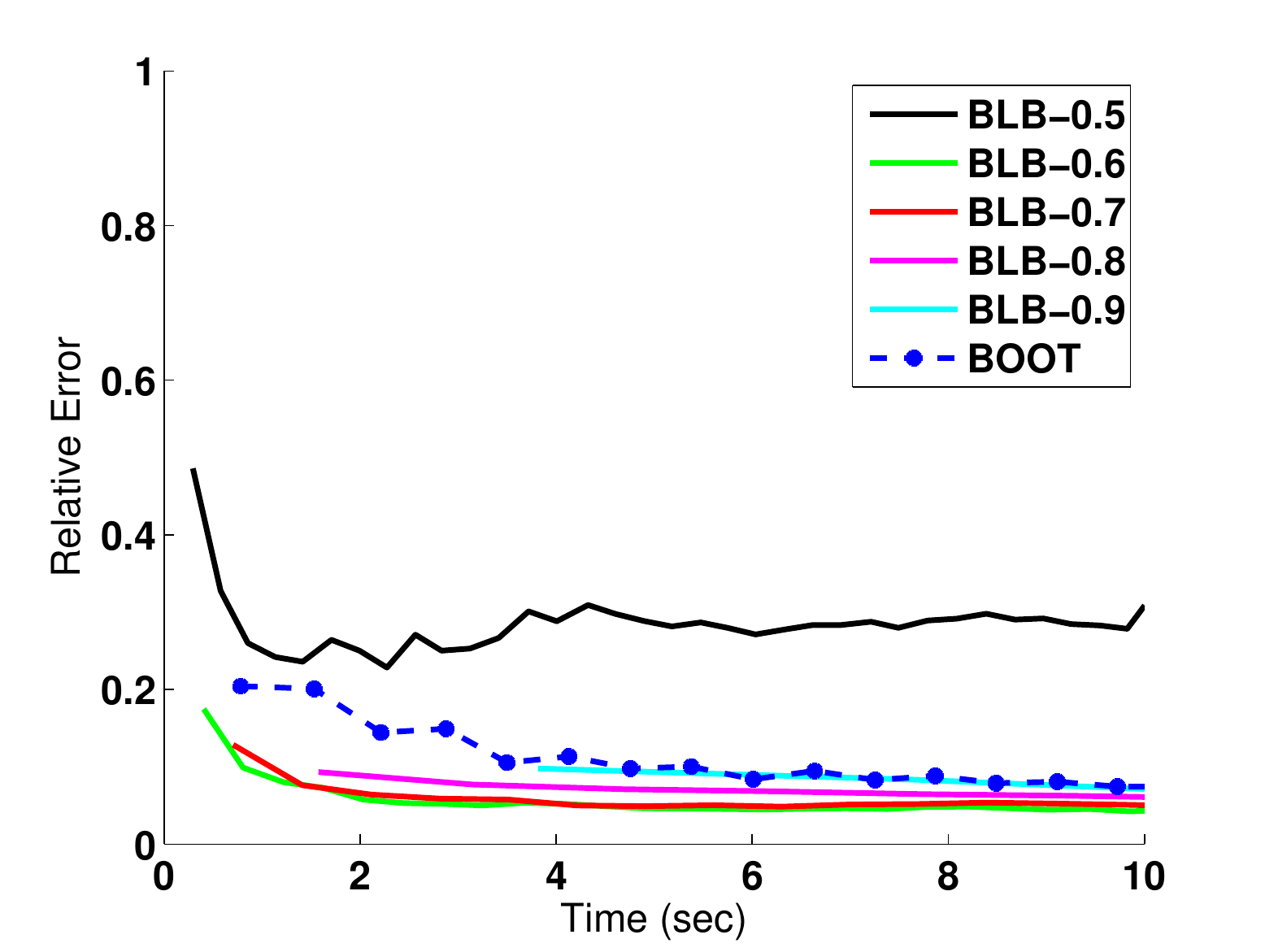}
\\
\includegraphics[width=0.32\linewidth]{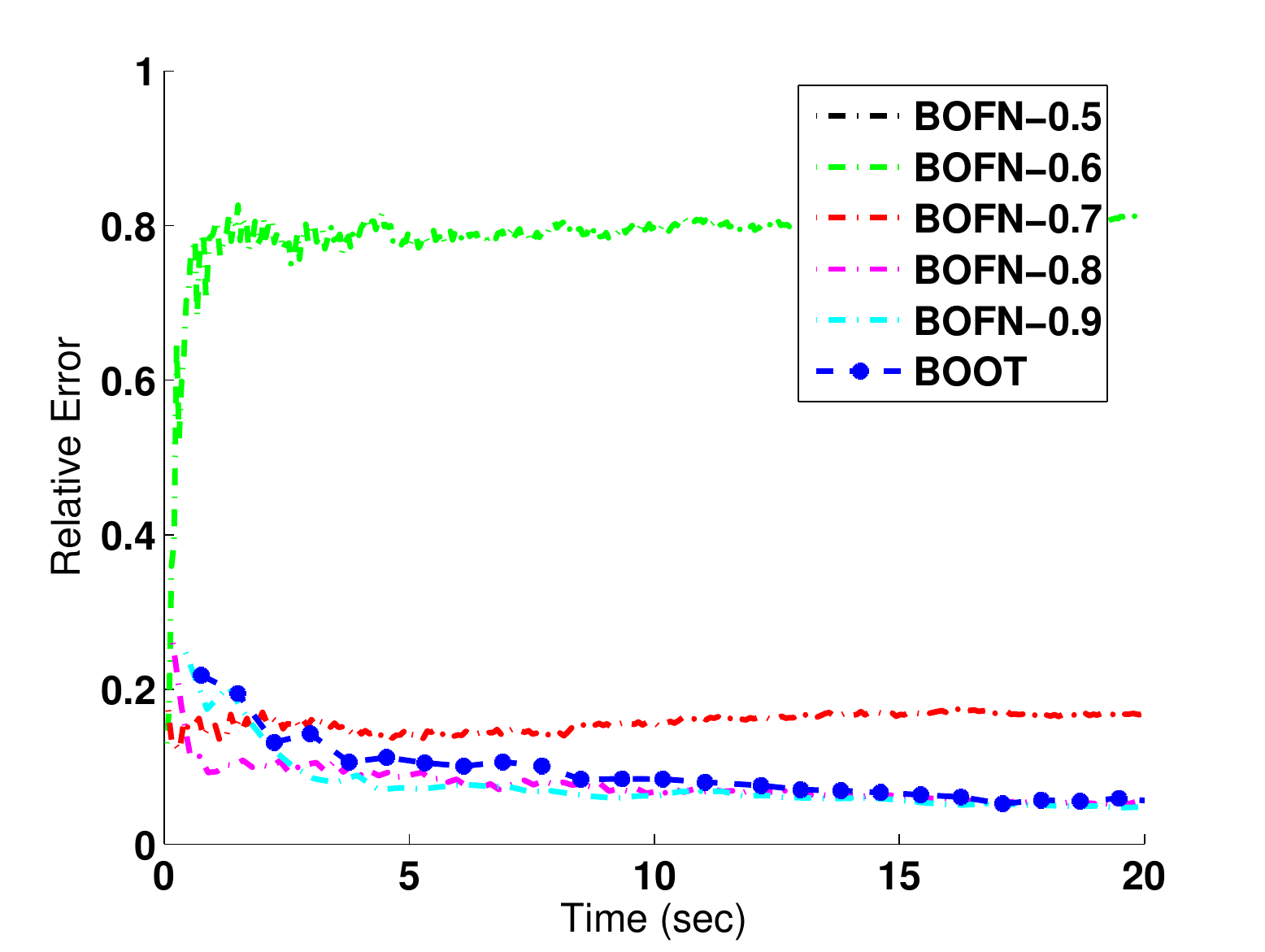}
\includegraphics[width=0.32\linewidth]{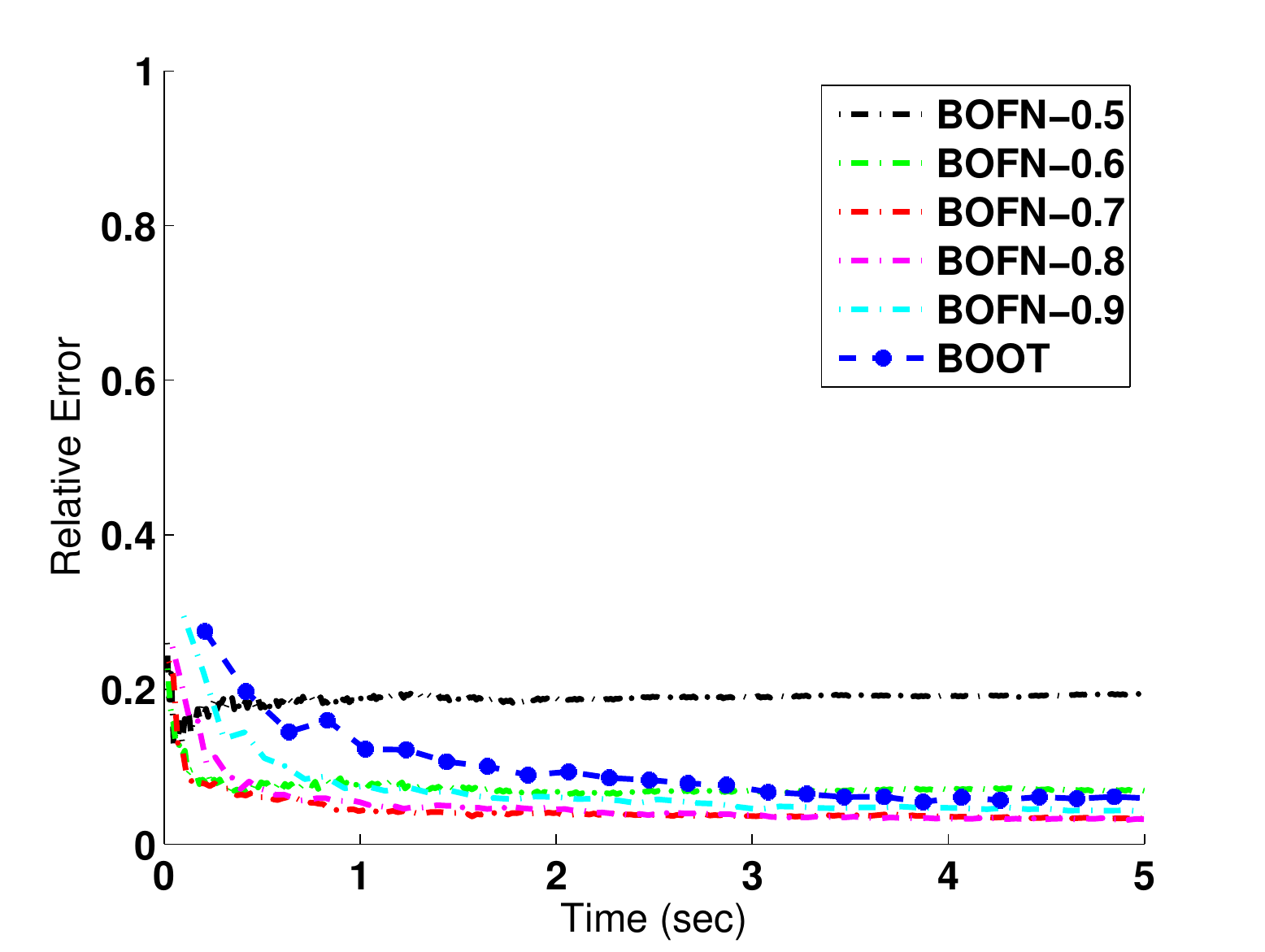}
\includegraphics[width=0.32\linewidth]{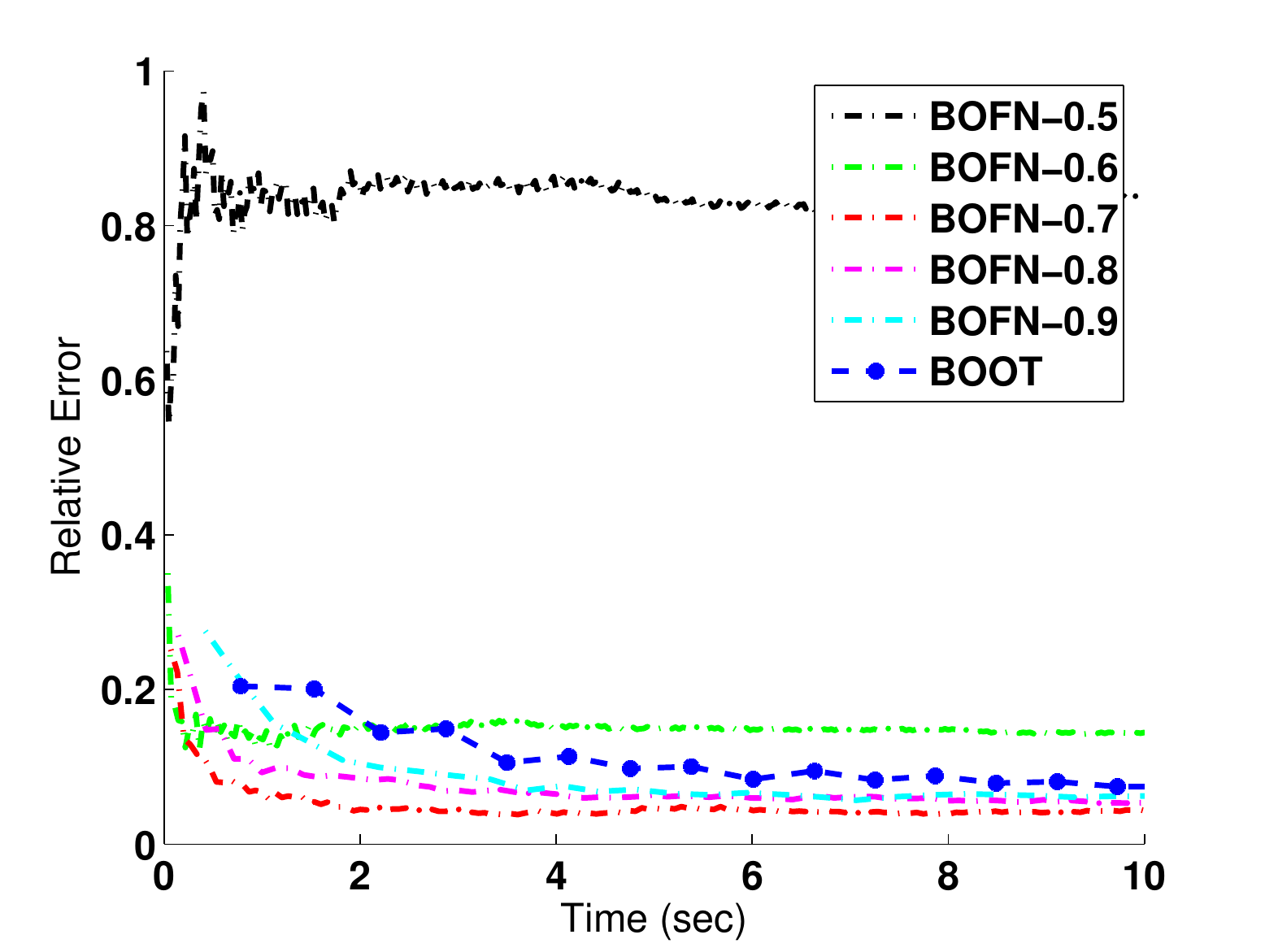}
\caption{Relative error vs.\ processing time for classification setting with $n=20,000$.  The top row shows \ouralgAbbrev with bootstrap (BOOT); bottom row shows $b$ out of $n$ bootstrap (BOFN).  For both \ouralgAbbrev and BOFN, $b=n^\gamma$ with the value of $\gamma$ for each trajectory given in the legend.  The left column shows results for logistic regression with linear data generating distribution and $\GammaDist$ $\tilde{X}_i$ distribution.  The middle column shows results for logistic regression with quadratic data generating distribution and $\GammaDist$ $\tilde{X}_i$ distribution.  The right column shows results for logistic regression with linear data generating distribution and $\StudentT$ $\tilde{X}_i$ distribution.}
\label{fig:logreg}
\end{figure*}

Figure~\ref{fig:logreg} shows results for the classification setting under the linear and quadratic data generating distributions with the $\GammaDist$ and $\StudentT$ $\tilde{X}_i$ distributions, and $n = 20,000$ (as in Figure~\ref{fig:linreg}); results for the $\Normal$ $\tilde{X_i}$ distribution are qualitatively similar.  Here, the performance of the various procedures is more varied than in the regression setting.  The case of the linear data generating distribution with $\GammaDist$ $\tilde{X}_i$ distribution (left column of Figure~\ref{fig:logreg}) appears to be the most challenging.  In this setting, \ouralgAbbrev converges to relative error comparable to that of the bootstrap for $b > n^{0.6}$, while converging to higher relative errors for the smallest values of $b$ considered.  For the larger values of $b$, which are still significantly smaller than $n$, we again converge to low relative error faster than the bootstrap.  We are also once again more robust than the $b$ out of $n$ bootstrap, which fails to converge to low relative error for $b \leq n^{0.7}$.  In fact, even for $b \leq n^{0.6}$, \ouralgAbbrev's performance is superior to that of the $b$ out of $n$ bootstrap.  Qualitatively similar results hold for the other data generating distributions, but with \ouralgAbbrev and the $b$ out of $n$ bootstrap both performing better relative to the bootstrap.  In the experiments shown in Figure~\ref{fig:logreg}, the values of $\numsub$ (which are implicit in the time axes of our plots) required for convergence of \ouralgAbbrev range from 1-2 for $b=n^{0.9}$ up to 10-20 for $b \leq n^{0.6}$ (for cases in which \ouralgAbbrev converges to low relative error).  As in the regression setting, subsampling (plots not shown) has performance strictly worse than that of the $b$ out of $n$ bootstrap in all cases.

\begin{figure*}
\includegraphics[width=0.5\linewidth]{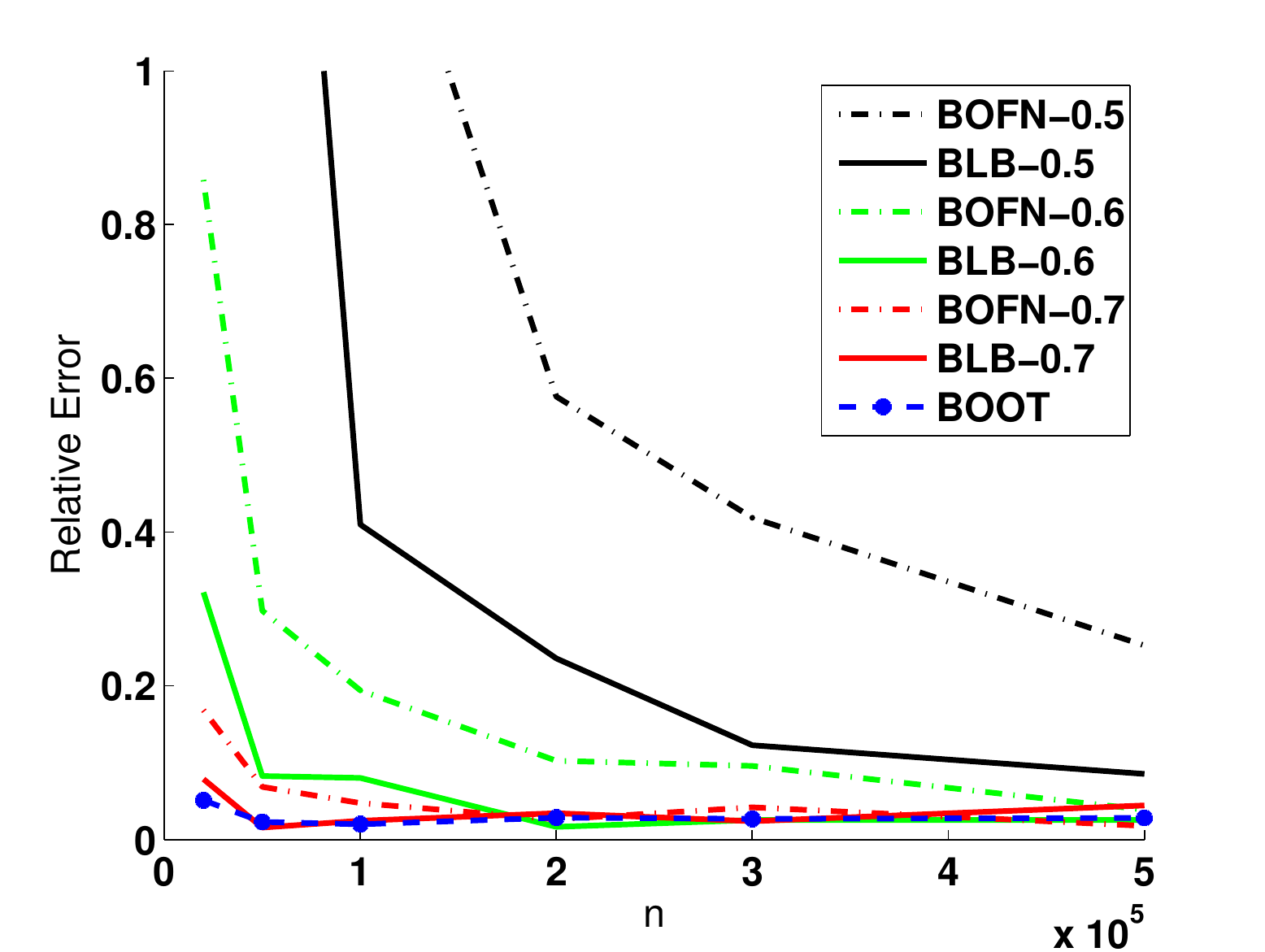}
\includegraphics[width=0.5\linewidth]{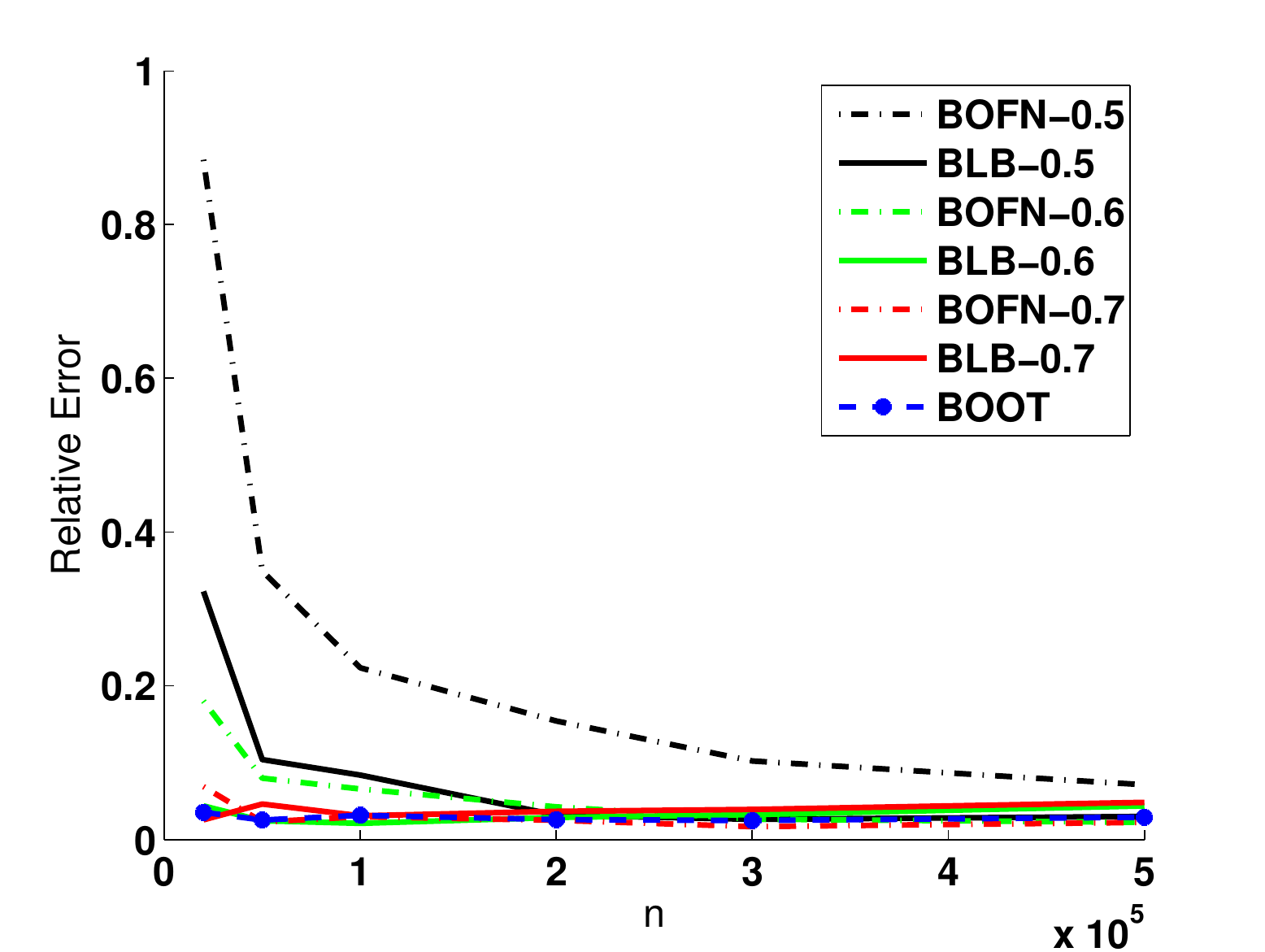}
\caption{Relative error (after convergence) vs.\ $n$ for \ouralgAbbrev, the $b$ out of $n$ bootstrap (BOFN), and the bootstrap (BOOT) in the classification setting.  For both \ouralgAbbrev and BOFN, $b=n^\gamma$ with the relevant values of $\gamma$ given in the legend.  The left plot shows results for logistic regression with linear data generating distribution and $\GammaDist$ $\tilde{X}_i$ distribution.  The right plot shows results for logistic regression with linear data generating distribution and $\StudentT$ $\tilde{X}_i$ distribution.}
\label{fig:logreg-vary-n}
\end{figure*}

To further examine the cases in which \ouralgAbbrev (when using small values of $b$) does not converge to relative error comparable to that of the bootstrap, we explore how the various procedures' relative errors vary with $n$.  In particular, for different values of $n$ (and $b$), we run each procedure as described above and report the relative error that it achieves after it converges (i.e., after it has processed sufficiently many subsets, in the case of \ouralgAbbrev, or resamples, in the case of the $b$ out of $n$ bootstrap and the bootstrap, to allow its output to stabilize).  Figure~\ref{fig:logreg-vary-n} shows results for the classification setting under the linear data generating distribution with the $\GammaDist$ and $\StudentT$ $\tilde{X}_i$ distributions; qualitatively similar results hold for the $\Normal$ $\tilde{X}_i$ distribution.  As expected based on our previous results for fixed $n$, \ouralgAbbrev's relative error here is higher than that of the bootstrap for the smallest values of $b$ and $n$ considered.  Nonetheless, \ouralgAbbrev's relative error decreases to that of the bootstrap as $n$ increases---for all considered values of $\gamma$, with $b = n^\gamma$---in accordance with our theoretical analysis; indeed, as $n$ increases, we can set $b$ to progressively more slowly growing functions of $n$ while still achieving low relative error.  Furthermore, \ouralgAbbrev's relative error is consistently substantially lower than that of the $b$ out of $n$ bootstrap and decreases more quickly to the low relative error of the bootstrap as $n$ increases.

\section{Computational Scalability}
\label{sec:scalability}

The experiments of the preceding section, though primarily intended to investigate statistical performance, also provide some insight into computational performance: as seen in Figures~\ref{fig:linreg} and~\ref{fig:logreg}, when computing on a single processor, \ouralgAbbrev generally requires less time, and hence less total computation, than the bootstrap to attain comparably high accuracy.  Those results only hint at \ouralgAbbrev's superior ability to scale computationally to large datasets, which we now demonstrate in full in the following discussion and via large-scale experiments on a distributed computing platform.

As discussed in Section~\ref{sec:blb}, modern massive datasets often exceed both the processing and storage capabilities of individual processors or compute nodes, thus necessitating the use of parallel and distributed computing architectures.
As a result, the scalability of a quality assessment method is closely tied to its ability to effectively utilize such computing resources.

Recall from our exposition in preceding sections that, due to the large size of bootstrap resamples,
the following is the most natural avenue for applying the bootstrap to large-scale data using distributed computing: given data partitioned across a cluster of compute nodes, parallelize the estimate computation on each resample across the cluster, and compute on one resample at a time.  This approach, while at least potentially feasible, remains quite problematic.  Each computation of the estimate will require the use of an entire cluster of compute nodes, and the bootstrap repeatedly incurs the associated overhead, such as the cost of repeatedly communicating intermediate data among nodes.  Additionally, many cluster computing systems currently in widespread use (e.g., Hadoop MapReduce~\cite{hadoop-website}) store data only on disk, rather than in memory, due to physical size constraints (if the dataset size exceeds the amount of available memory) or architectural constraints (e.g., the need for fault tolerance).  In that case, the bootstrap incurs the extreme costs associated with repeatedly reading a very large dataset from disk---reads from disk are orders of magnitude slower than reads from memory.  Though disk read costs may be acceptable when (slowly) computing only a single full-data point estimate, they easily become prohibitive when computing many estimates on one hundred or more resamples.  Furthermore, as we have seen, executing the bootstrap at scale requires implementing the estimator such that it can be run on data distributed over a cluster of compute nodes.

In contrast, \ouralgAbbrev permits computation on multiple (or even all) subsamples and resamples simultaneously in parallel, allowing for straightforward distributed and parallel implementations which enable effective scalability and large computational gains.  Because \ouralgAbbrev subsamples and resamples can be significantly smaller than the original dataset, they can be transferred to, stored by, and processed independently on individual (or very small sets of) compute nodes.  For instance, we can distribute subsamples to different compute nodes and subsequently use intra-node parallelism to compute across different resamples generated from the same subsample.  Note that generation and distribution of the subsamples requires only a single pass over the full dataset (i.e., only a single read of the full dataset from disk, if it is stored only on disk), after which all required data (i.e., the subsamples) can potentially be stored in memory.  Beyond this significant architectural benefit, we also achieve implementation and algorithmic benefits: we do not need to parallelize the estimator internally to take advantage of the available parallelism, as \ouralgAbbrev uses this available parallelism to compute on multiple resamples simultaneously, and exposing the estimator to only $b$ rather than $n$ distinct points significantly reduces the computational cost of estimation, particularly if the estimator computation scales super-linearly.

Given the shortcomings of the $m$ out of $n$ bootstrap and subsampling illustrated in the preceding section, we do not include these methods in the scalability experiments of this section.  However, it is worth noting that these procedures have a significant computational shortcoming in the setting of large-scale data: the $m$ out of $n$ bootstrap and subsampling require repeated access to many different random subsets of the original dataset (in contrast to the relatively few, potentially disjoint, subsamples required by \ouralgAbbrev), and this access can be quite costly when the data is distributed across a cluster of compute nodes.

We now detail our large-scale experiments on a distributed computing platform.  For this empirical study, we use the experimental setup of Section~\ref{sec:simulation}, with some modification to accommodate larger scale and distributed computation.  First, we now use $d = 3,000$ and $n = 6,000,000$ so that the size of a full observed dataset is approximately 150~GB.  The full dataset is partitioned across a number of compute nodes.  We again use simulated data to allow knowledge of ground truth; due to the substantially larger data size and attendant higher running times, we now use 200 independent realizations of datasets of size $n$ to numerically compute the ground truth.  As our focus is now computational (rather than statistical) performance, we present results here for a single data generating distribution which yields representative statistical performance based on the results of the previous section; for a given dataset size, changing the underlying data generating distribution does not alter the computational resources required for storage and processing.  For the experiments in this section, we consider the classification setting with $\StudentT$ $\tilde{X}_i$ distribution.  The mapping between $\tilde{X}_i$ and $Y_i$ remains similar to that of the linear data generating distribution in Section~\ref{sec:simulation}, but with the addition of a normalization factor to prevent degeneracy when using larger $d$: $Y_i \sim \Bernoulli((1 + \exp(-\tilde{X}_i^T \1 / \sqrt{d}))^{-1})$.  We implement the logistic regression using L-BFGS~\cite{nocedal-wright} due to the significantly larger value of $d$.

We compare the performance of \ouralgAbbrev and the bootstrap, both implemented as described above.  That is, our implementation of \ouralgAbbrev processes all subsamples simultaneously in parallel on independent compute nodes; we use $r=50$, $s=5$, and $b=n^{0.7}$.  Our implementation of the bootstrap uses all available processors to compute on one resample at a time, with computation of the logistic regression parameter estimates parallelized across the available compute nodes by simply distributing the relevant gradient computations among the different nodes upon which the data is partitioned.  We utilize Poisson resampling~\cite{vdv-wellner} to generate bootstrap resamples, thereby avoiding the complexity of generating a random multinomial vector of length $n$ in a distributed fashion.  Due to high running times, we show results for a single trial of each method, though we have observed little variability in qualitative outcomes during development of these experiments.  All experiments in this section are run on Amazon EC2 and implemented in the Scala programming language using the Spark cluster computing framework~\cite{spark-paper}, which provides the ability to either read data from disk (in which case performance is similar to that of Hadoop MapReduce) or cache it in memory across a cluster of compute nodes (provided that sufficient memory is available) for faster repeated access.

\begin{figure*}
\includegraphics[width=0.5\linewidth]{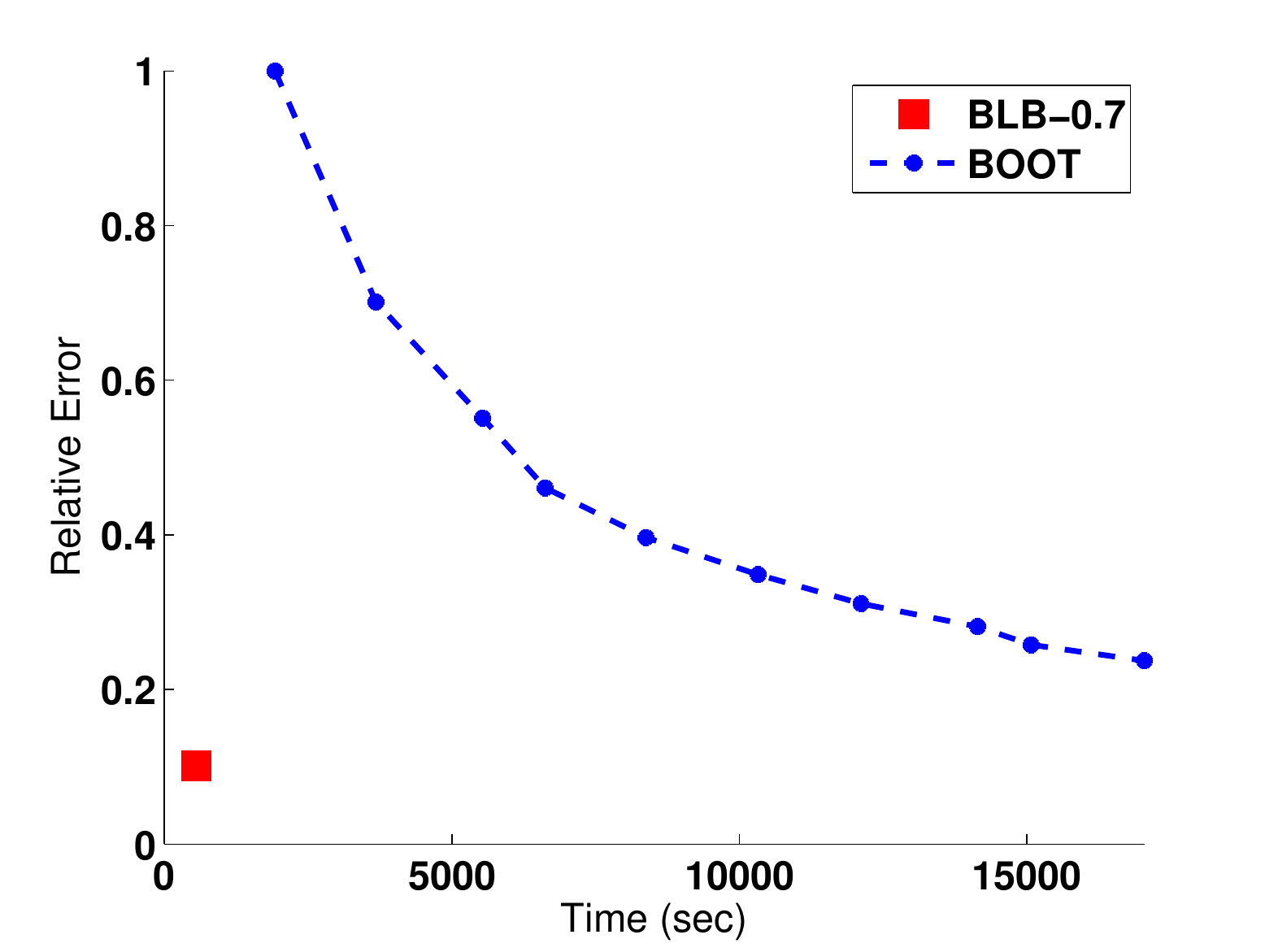}
\includegraphics[width=0.5\linewidth]{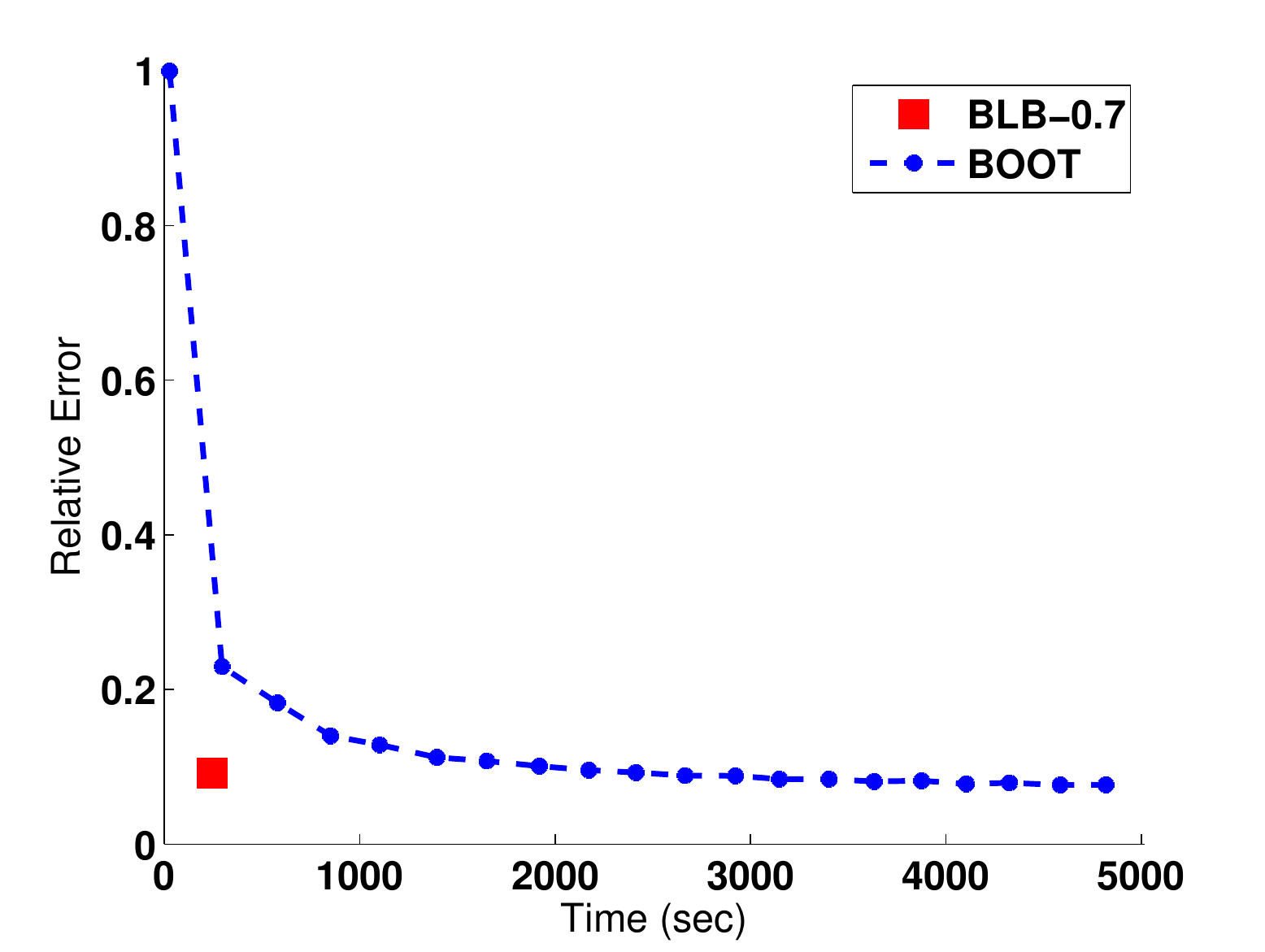}
\caption{Relative error vs.\ processing time for \ouralgAbbrev (with $b=n^{0.7}$) and the bootstrap (BOOT) on 150~GB of data in the classification setting.  The left plot shows results with the full dataset stored only on disk; the right plot shows results with the full dataset cached in memory.  Because \ouralgAbbrev's computation is fully parallelized across all subsamples, we show only the processing time and relative error of \ouralgAbbrev's final output.}
\label{fig:scalability}
\end{figure*}

In the left plot of Figure~\ref{fig:scalability}, we show results obtained using a cluster of 10 worker nodes, each having 6~GB of memory and 8~compute cores; thus, the total memory of the cluster is 60~GB, and the full dataset (150~GB) can only be stored on disk (the available disk space is ample and far exceeds the dataset size).  As expected, the time required by the bootstrap to produce even a low-accuracy output is prohibitively high, while \ouralgAbbrev provides a high-accuracy output quite quickly, in less than the time required to process even a single bootstrap resample.  In the right plot of Figure~\ref{fig:scalability}, we show results obtained using a cluster of 20 worker nodes, each having 12~GB of memory and 4~compute cores; thus, the total memory of the cluster is 240~GB, and we cache the full dataset in memory for faster repeated access.  Unsurprisingly, the bootstrap's performance improves significantly with respect to the previous disk-bound experiment.  However, the performance of \ouralgAbbrev (which also improves), remains substantially better than that of the bootstrap.

\section{Hyperparameter Selection}
\label{sec:hyperparam}

Like existing resampling-based procedures such as the bootstrap, \ouralgAbbrev requires the specification of hyperparameters controlling the number of subsamples and resamples processed.  Setting such hyperparameters to be sufficiently large is necessary to ensure good statistical performance; however, setting them to be unnecessarily large results in wasted computation.  Prior work on the bootstrap and related procedures---which largely does not address computational issues---generally assumes that a procedure's user will simply select a priori a large, constant number of resamples to be processed (with the exception of~\citet{tibs-how-many}, who does not provide a general solution for this issue).  However, this approach reduces the level of automation of these methods and can be quite inefficient in the large data setting, in which each subsample or resample can require a substantial amount of computation.

Thus, we now examine the dependence of \ouralgAbbrev's performance on the choice of $r$ and $s$, with the goal of better understanding their influence and providing guidance toward achieving adaptive methods for their selection.  For any particular application of \ouralgAbbrev, we seek to select the minimal values of $r$ and $s$ which are sufficiently large to yield good statistical performance.

Recall that in the simulation study of Section~\ref{sec:simulation}, across all of the settings considered, fairly modest values of $r$ (100 for confidence intervals) and $s$ (from 1-2 for $b=n^{0.9}$ up to 10-20 for $b = n^{0.6}$) were sufficient.  The left plot of Figure~\ref{fig:hyperparam} provides further insight into the influence of $r$ and $s$, giving the relative errors achieved by \ouralgAbbrev with $b=n^{0.7}$ for different $r, s$ pairs in the classification setting with linear data generating distribution and $\StudentT$ $\tilde{X}_i$ distribution.  In particular, note that for all but the smallest values of $r$ and $s$, it is possible to choose these values independently such that \ouralgAbbrev achieves low relative error; in this case, selecting $s \geq 3, r \geq 50$ is sufficient.

\begin{figure*}
\centering
\begin{minipage}[c]{0.6\linewidth}
	\includegraphics[width=0.48\linewidth]{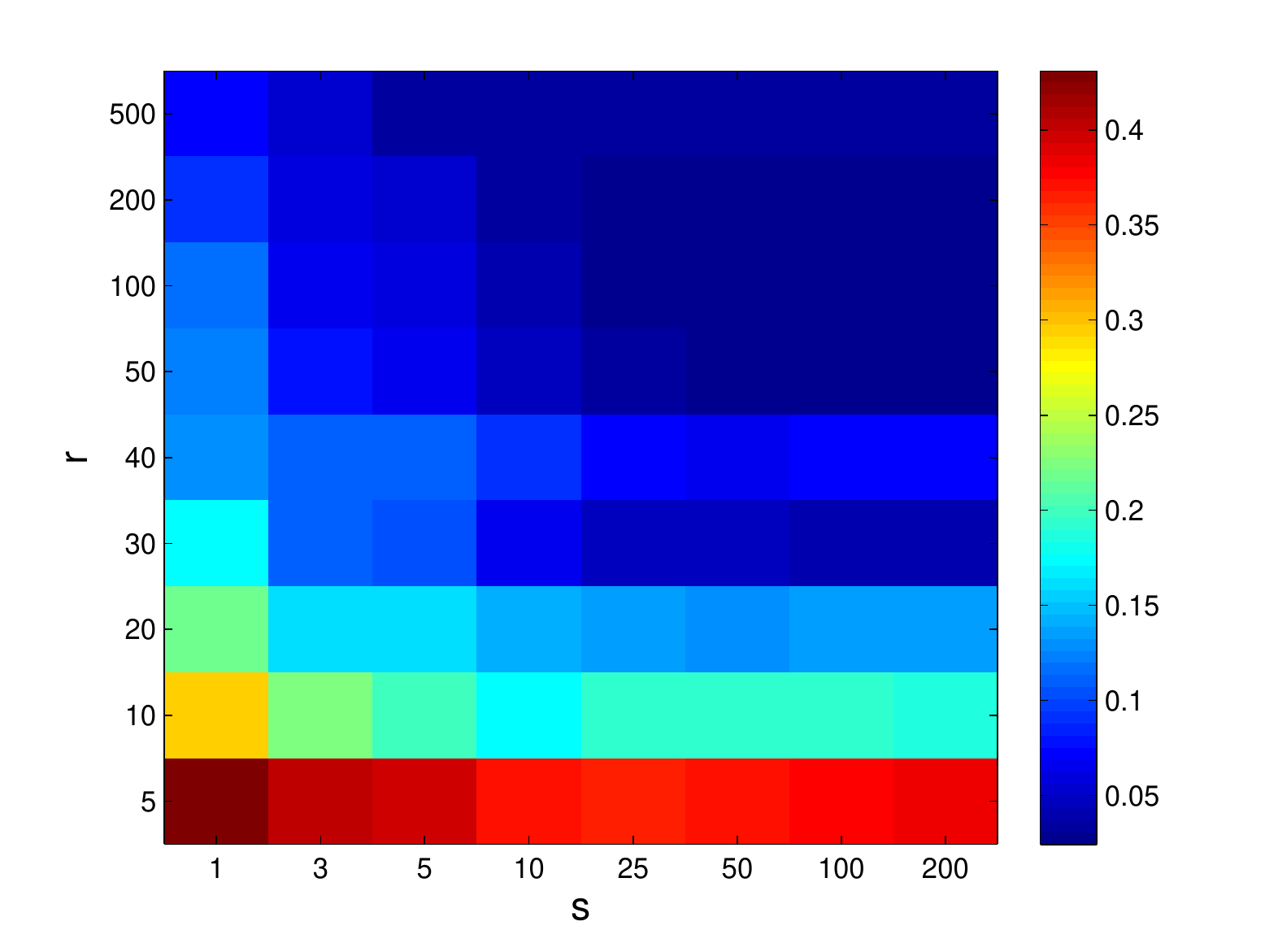}
	\includegraphics[width=0.48\linewidth]{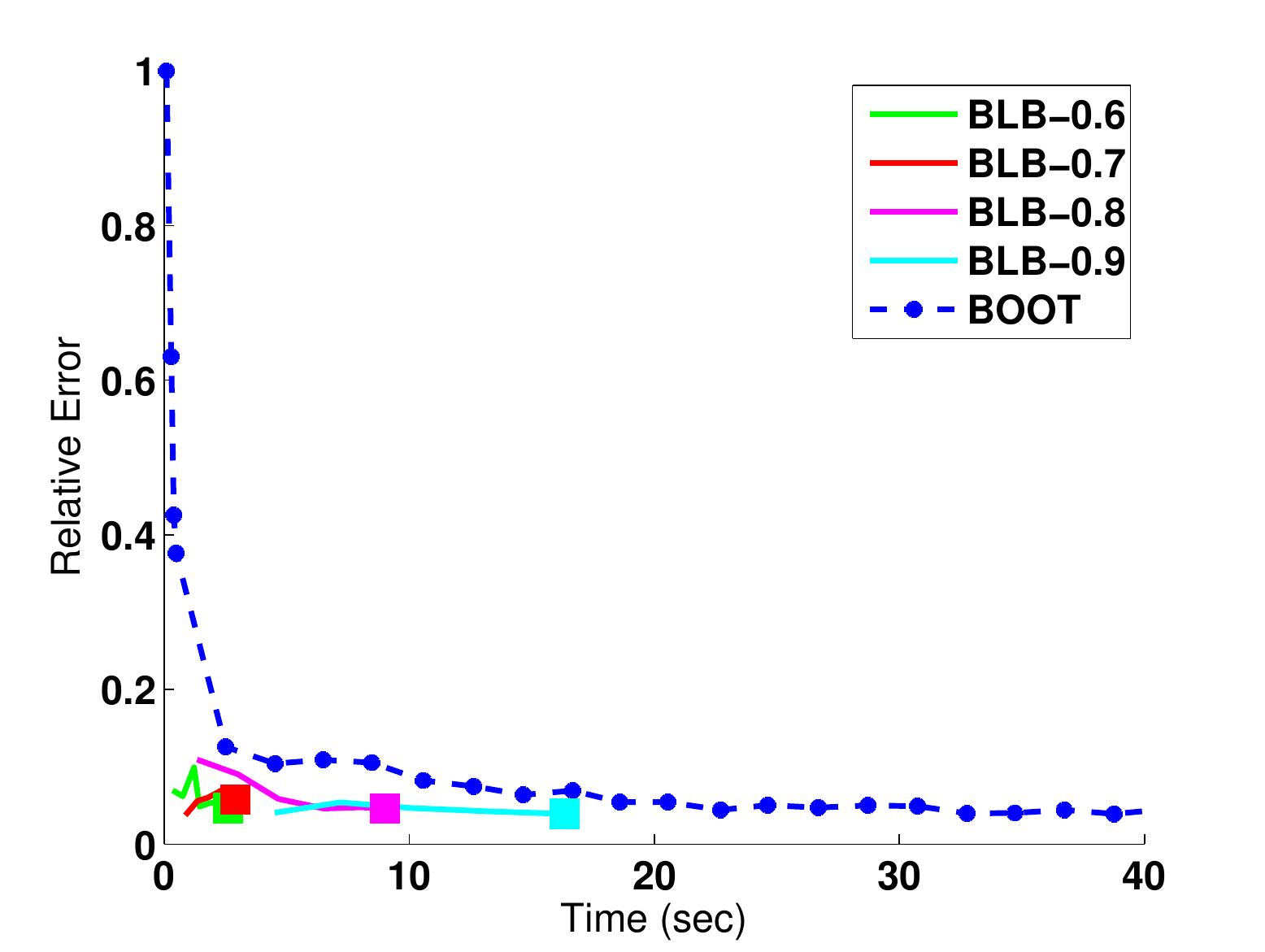}
\end{minipage}
\begin{minipage}[c]{0.3\linewidth}
{
\begin{small}
\begin{sc}
\begin{tabular}{lcc}
\hline
$r$ Stats & CI  & STDERR \\
\hline
mean & 89.6 & 67.7 \\
min & 50 & 40 \\
max & 150 & 110 \\
\hline
\end{tabular}
\end{sc}
\end{small}
} \end{minipage}
\caption{Results for \ouralgAbbrev hyperparameter selection in the classification setting with linear data generating distribution and $\StudentT$ $\tilde{X}_i$ distribution.  The left plot shows the relative error achieved by \ouralgAbbrev for different values of $r$ and $s$, with $b=n^{0.7}$.  The right plot shows relative error vs.\ processing time (without parallelization) for \ouralgAbbrev using adaptive selection of $r$ and $s$ (the resulting stopping times of the \ouralgAbbrev trajectories are marked by large squares) and the bootstrap (BOOT); for \ouralgAbbrev, $b=n^\gamma$ with the value of $\gamma$ for each trajectory given in the legend.  The table gives statistics of the different values of $r$ selected by \ouralgAbbrev's adaptive hyperparameter selection (across multiple subsamples, with $b=n^{0.7}$) when $\xi$ is either our usual confidence interval width-based quality measure (CI), or a component-wise standard error (STDERR); the relative errors achieved by \ouralgAbbrev and the bootstrap are comparable in both cases.}
\label{fig:hyperparam}
\end{figure*}

While these results are useful and provide some guidance for hyperparameter selection, we expect the sufficient values of $r$ and $s$ to change based on the identity of $\xi$ (e.g., we expect a confidence interval to be harder to compute and hence to require larger $r$ than a standard error) and the properties of the underlying data.  Thus, to help avoid the need to set $r$ and $s$ to be conservatively and inefficiently large, we now provide a means for adaptive hyperparameter selection, which we validate empirically.

Concretely, to select $r$ adaptively in the inner loop of Algorithm~\ref{alg:ouralg}, we propose an iterative scheme whereby, for any given subsample $j$, we continue to process resamples and update $\s{\xi}_{n,j}$ until it has ceased to change significantly.  Noting that the values $\s{\sest{\theta}}_{n,k}$ used to compute $\s{\xi}_{n,j}$ are conditionally i.i.d.\ given a subsample, for most forms of $\xi$ the series of computed $\s{\xi}_{n,j}$ values will be well behaved and will converge (in many cases at rate $O(1/\sqrt{r})$, though with unknown constant) to a constant target value as more resamples are processed.  Therefore, it suffices to process resamples (i.e., to increase $r$) until we are satisfied that $\s{\xi}_{n,j}$ has ceased to fluctuate significantly; we propose using Algorithm~\ref{alg:heuristic} to assess this convergence.  The same scheme can be used to select $s$ adaptively by processing more subsamples (i.e., increasing $s$) until \ouralgAbbrev's output value $\numsub^{-1} \sum_{j=1}^{\numsub} \s{\xi}_{n,j}$ has stabilized; in this case, one can simultaneously also choose $r$ adaptively and independently for each subsample.  When parallelizing across subsamples and resamples, one can simply process batches of subsamples and resamples (with batch size determined by the available parallelism) until the output stabilizes.

\begin{algorithm}[tb]
   \caption{Convergence Assessment}
   \label{alg:heuristic}
   \DontPrintSemicolon

   \KwIn{\parbox[t]{0.8\linewidth}{A series $z^{(1)}, z^{(2)}, \ldots, z^{(t)} \in \R^d$
						\newline $w \in \mathbb{N}$: window size ($< t$)
						\newline $\epsilon \in \R$: target relative error ($>0$)}}
   \KwOut{true if and only if the input series is deemed to have ceased to fluctuate beyond the target relative error}
   \BlankLine
   \eIf{$\forall j \in [1, w]$, $\frac{1}{d} \sum_{i=1}^{d} \frac{| z_i^{(t-j)} - z_i^{(t)} |}{| z_i^{(t)} |} \leq \epsilon$}{\Return{true}}{\Return{false}}
\end{algorithm}

The right plot of Figure~\ref{fig:hyperparam} shows the results of applying such adaptive hyperparameter selection in a representative empirical setting from our earlier simulation study (without parallelization).  For selection of $r$ we use $\epsilon=0.05$ and $w=20$, and for selection of $s$ we use $\epsilon=0.05$ and $w=3$.  As illustrated in the plot, the adaptive hyperparameter selection allows \ouralgAbbrev to cease computing shortly after it has converged (to low relative error), limiting the amount of unnecessary computation that is performed without degradation of statistical performance.  Though selected a priori, $\epsilon$ and $w$ are more intuitively interpretable and less dependent on the details of $\xi$ and the underlying data generating distribution than $r$ and $s$.  Indeed, the aforementioned specific values of $\epsilon$ and $w$ yield results of comparably good quality when also used for the other data generation settings considered in Section~\ref{sec:simulation}, when applied to a variety of real datasets in Section~\ref{sec:real-data} below, and when used in conjunction with different forms of $\xi$ (see the table in Figure~\ref{fig:hyperparam}, which shows that smaller values of $r$ are selected when $\xi$ is easier to compute).  Thus, our scheme significantly helps to alleviate the burden of a priori hyperparameter selection.

Automatic selection of a value of $b$ in a computationally efficient manner would also be desirable but is more difficult due to the inability to easily reuse computations performed for different values of $b$.  One could consider similarly increasing $b$ from some small value until the output of \ouralgAbbrev stabilizes (an approach reminiscent of the method proposed in~\citet{bickel-sakov-choice-of-m} for the $m$ out of $n$ bootstrap); devising a means of doing so efficiently is the subject of future work.  Nonetheless, based on our fairly extensive empirical investigation, it seems that $b = n^{0.7}$ is a reasonable and effective choice in many situations.

\section{Real Data}
\label{sec:real-data}

In this section, we present the results of applying \ouralgAbbrev to several different real datasets.  In this case, given the absence of ground truth, it is not possible to objectively evaluate the statistical correctness of any particular estimator quality assessment method; rather, we are reduced to comparing the outputs of various methods (in this case, \ouralgAbbrev, the bootstrap, and the $b$ out of $n$ bootstrap) to each other.  Because we cannot determine the relative error of each procedure's output without knowledge of ground truth, we now instead report the average (across dimensions) absolute confidence interval width yielded by each procedure.

Figure~\ref{fig:real-data-classif} shows results for \ouralgAbbrev, the bootstrap, and the $b$ out of $n$ bootstrap on the UCI connect4 dataset~\cite{uci-repo}, where the model is logistic regression (as in the classification setting of our simulation study above), $d=42$, and $n=67,557$.  We select the \ouralgAbbrev hyperparameters $r$ and $s$ using the adaptive method described in the preceding section.  Notably, the outputs of \ouralgAbbrev for all values of $b$ considered, and the output of the bootstrap, are tightly clustered around the same value; additionally, as expected, \ouralgAbbrev converges more quickly than the bootstrap.  However, the values produced by the $b$ out of $n$ bootstrap vary significantly as $b$ changes, thus further highlighting this procedure's lack of robustness.  We have obtained qualitatively similar results on six additional datasets from the UCI dataset repository (ct-slice, magic, millionsong, parkinsons, poker, shuttle)~\cite{uci-repo} with different estimators (linear regression and logistic regression) and a range of different values of $n$ and $d$ (see the appendix for plots of these results).

\begin{figure*}
\includegraphics[width=0.5\linewidth]{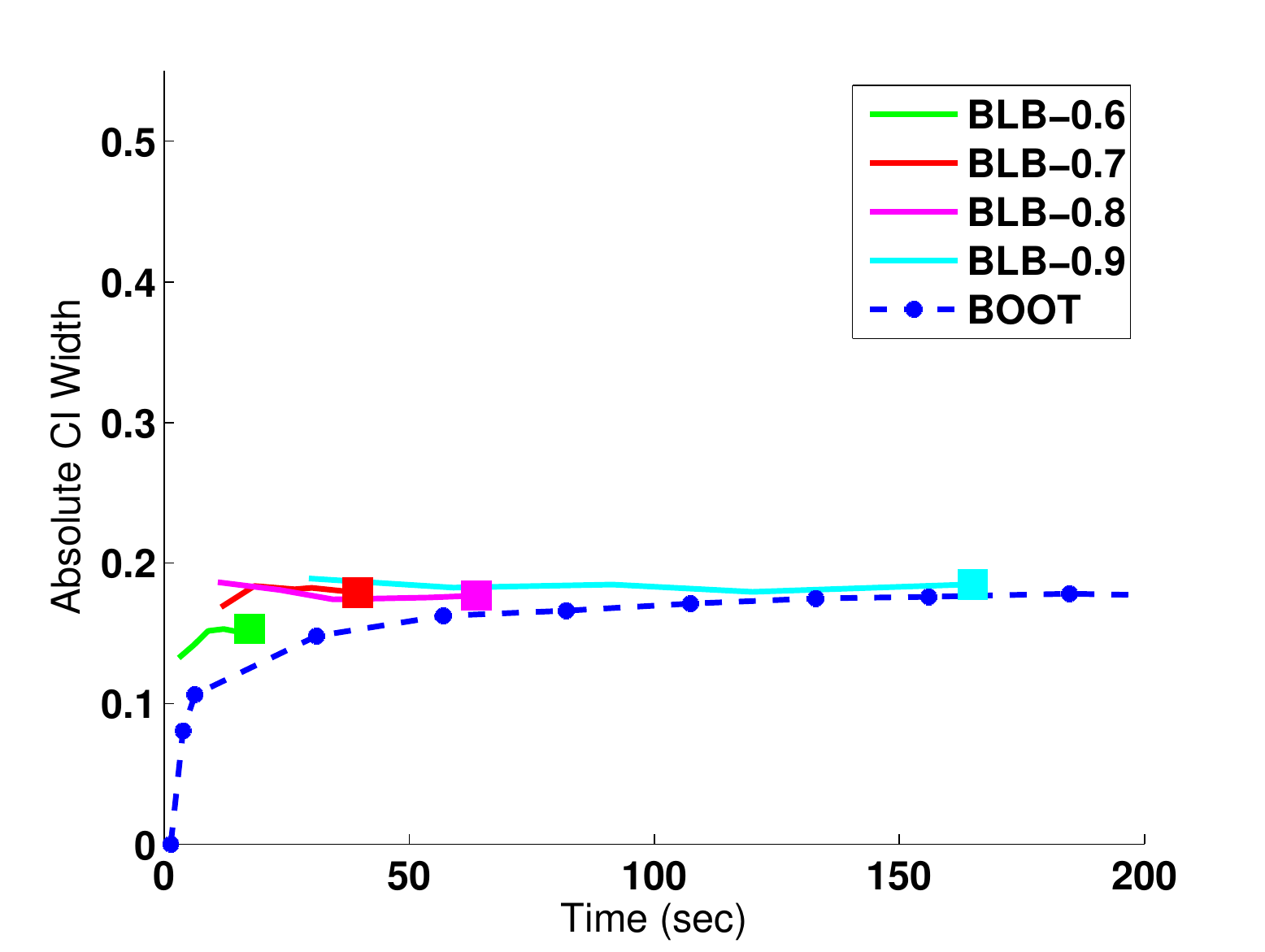}
\includegraphics[width=0.5\linewidth]{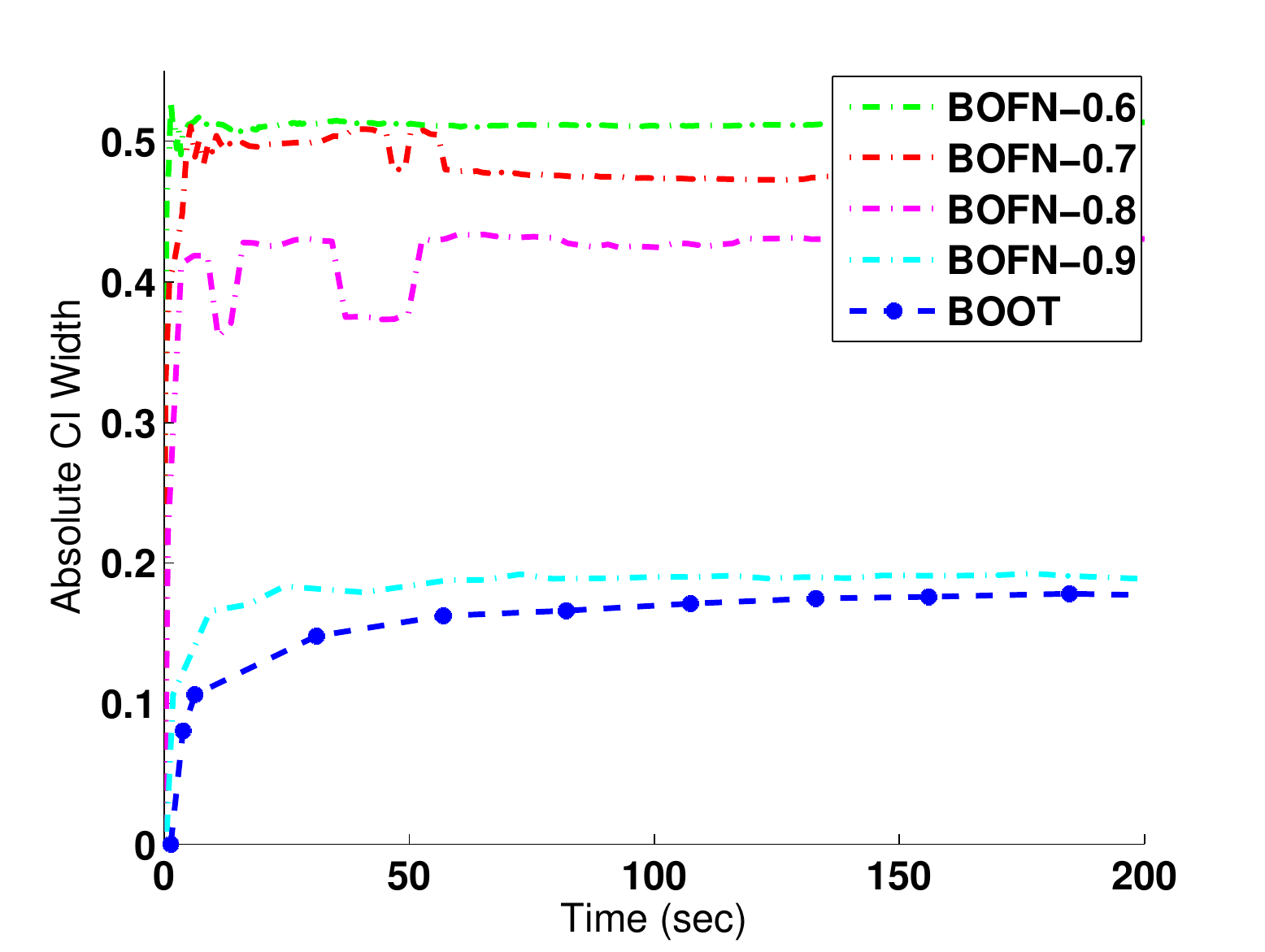}
\caption{Average (across dimensions) absolute confidence interval width vs.\ processing time on the UCI connect4 dataset (logistic regression, $d=42, n= 67,557$).  The left plot shows results for \ouralgAbbrev (using adaptive hyperparameter selection, with the output at convergence marked by large squares) and the bootstrap (BOOT).  The right plot shows results for the $b$ out of $n$ bootstrap (BOFN).  For both \ouralgAbbrev and BOFN, $b=n^\gamma$ with the value of $\gamma$ for each trajectory given in the legend.}
\label{fig:real-data-classif}
\end{figure*}

\section{Time Series}
\label{sec:time-series}

While we have focused thus far on the setting of i.i.d.\ data, variants of the bootstrap---such as the moving block bootstrap and the stationary bootstrap---have been proposed to handle other data analysis settings such as that of time series~\cite{efron-intro-bootstrap, hall-mammen, kunsch-mbb, liu-singh-mbb, politis-romano-statboot}.  These bootstrap variants can be used within \ouralgAbbrev, in computing the requisite plugin approximations $\xi(Q_n(\P_{n,b}^{(j)}))$, to obtain variants of our procedure which are applicable in non-i.i.d.\ settings.  The advantages (e.g., with respect to scalability) of such \ouralgAbbrev variants over variants of the bootstrap (and its relatives) remain identical to the advantages discussed above in the context of large-scale i.i.d.\ data.  We briefly demonstrate the extensibility of \ouralgAbbrev by combining our procedure with the stationary bootstrap~\cite{politis-romano-statboot} to obtain a ``stationary \ouralgAbbrev'' which is suitable for assessing the quality of estimators applied to large-scale stationary time series data.

To extend \ouralgAbbrev in this manner, we must simply alter both the subsample selection mechanism and the resample generation mechanism such that both of these processes respect the underlying data generating process.  In particular, for stationary time series data it suffices to select each subsample as a (uniformly) randomly positioned block of length $b$ within the observed time series of length $n$.  Given a subsample of size $b$, we generate each resample by applying the stationary bootstrap to the subsample to obtain a series of length $n$.  That is, given $p \in [0,1]$ (a hyperparameter of the stationary bootstrap), we first select uniformly at random a data point in the subsample series and then repeat the following process until we have amassed a new series of length $n$: with probability $1-p$ we append to our resample the next point in the subsample series (wrapping around to the beginning if we reach the end of the subsample series), and with probability $p$ we (uniformly at random) select and append a new point in the subsample series.  Given subsamples and resamples generated in this manner, we execute the remainder of the \ouralgAbbrev procedure as described in Algorithm~\ref{alg:ouralg}.

We now present simulation results comparing the performance of the bootstrap, BLB, the stationary bootstrap, and stationary BLB.  In this experiment, initially introduced by~\citet{politis-romano-statboot}, we generate observed data consisting of a stationary time series $X_1, \ldots, X_{n} \in \R$ where $X_t = Z_t + Z_{t-1} + Z_{t-2} + Z_{t-3} + Z_{t-4}$ and the $Z_t$ are drawn independently from $\Normal(0,1)$.  We consider the  task of estimating the standard deviation of the rescaled mean $\sum_{i=1}^n X_t / \sqrt{n}$, which is approximately 5; we set $p = 0.1$ for the stationary bootstrap and stationary BLB.  The results in Table~\ref{tbl:time-series} (for $n=5,000$) show the improvement of the stationary bootstrap over the bootstrap, the similar improvement of stationary \ouralgAbbrev over \ouralgAbbrev, and the fact that the statistical performance of stationary BLB is comparable to that of the stationary bootstrap for $b \ge n^{0.7}$.  Note that this exploration of stationary \ouralgAbbrev is intended as a proof of concept, and additional investigation would help to further elucidate and perhaps improve the performance characteristics of this \ouralgAbbrev extension.

\begin{table}
\begin{center} \begin{small} \begin{sc}
  \begin{tabular}{ccc}
    \hline
    Method & Standard & Stationary \\
    \hline
    BLB-0.6 & $2.2 \pm .1$  & $4.2 \pm .1$ \\
    BLB-0.7 & $2.2 \pm .04$ & $4.5 \pm .1$ \\
    BLB-0.8 & $2.2 \pm .1$  & $4.6 \pm .2$ \\
    BLB-0.9 & $2.2 \pm .1$  & $4.6 \pm .1$ \\
    BOOT & $2.2 \pm .1$  & $4.6 \pm .2$ \\
    \hline
  \end{tabular}
\end{sc} \end{small} \end{center}
\caption{Comparison of standard and stationary bootstrap (BOOT) and \ouralgAbbrev on stationary time series data with $n = 5,000$.  We report the average and standard deviation of estimates (after convergence) of the standard deviation of the rescaled mean aggregated over 10 trials.  The true population value of the standard deviation of the rescaled mean is approximately~5.}
\label{tbl:time-series}
\end{table}

\section{Conclusion}
\label{sec:conclusion}

We have presented a new procedure, \ouralgAbbrev, which provides a powerful new alternative for automatic, accurate assessment of estimator quality that is well suited to large-scale data and modern parallel and distributed computing architectures.  \OuralgAbbrev shares the favorable statistical properties (i.e., consistency and higher-order correctness) and generic applicability of the bootstrap, while typically having a markedly better computational profile, as we have demonstrated via large-scale experiments on a distributed computing platform.  Additionally, \ouralgAbbrev is consistently more robust than the $m$ out of $n$ bootstrap and subsampling to the choice of subset size and does not require the use of analytical corrections.  To enhance our procedure's computational efficiency and render it more automatically usable, we have introduced a means of adaptively selecting its hyperparameters.  We have also applied \ouralgAbbrev to several real datasets and presented an extension to non-i.i.d.\ time series data.

A number of open questions and possible extensions remain.  Though we have constructed an adaptive hyperparameter selection method based on the properties of the subsampling and resampling processes used in \ouralgAbbrev, as well as empirically validated the method, it would be useful to develop a more precise theoretical characterization of its behavior.  Additionally, as discussed in Section~\ref{sec:hyperparam}, it would be beneficial to develop a computationally efficient means of adaptively selecting $b$.  It may also be possible to further reduce $\numre$ by using methods that have been proposed for reducing the number of resamples required by the bootstrap~\cite{efficient-bootstrap-efron, efron-intro-bootstrap}.

Furthermore, it is worth noting that, while \ouralgAbbrev shares the statistical strengths of the bootstrap, we conversely do not expect our procedure to be applicable in cases in which the bootstrap fails~\cite{bootstrap-moutofn}.  Indeed, it was such edge cases that originally motivated development of the $m$ out of $n$ bootstrap and subsampling, which are consistent in various settings that are problematic for the bootstrap.  It would be interesting to investigate the performance of \ouralgAbbrev in such settings and perhaps use ideas from the $m$ out of $n$ bootstrap and subsampling to improve the applicability of \ouralgAbbrev in these edge cases while maintaining computational efficiency and robustness.  Finally, note that averaging the plugin approximations $\xi(Q_n(\P_{n,b}^{(j)}))$ in equation~\eq{ouralg} implicitly corresponds to minimizing the squared error of \ouralgAbbrev's output.  It would be possible to specifically optimize for other losses on our estimator quality assessments by combining the $\xi(Q_n(\P_{n,b}^{(j)}))$ in other ways (e.g., by using medians rather than averages).

{
\bibliographystyle{plainnat}
\bibliography{refs}
}

\appendix
\section{Appendix: Proofs}
\label{sec:appendix-proofs}

We provide here full proofs of the theoretical results included in Section~\ref{sec:theory} above.

\subsection{Consistency}

We first define some additional notation, following that used by~\citet{vdv-wellner}.  Let $l^\infty (\F)$ be the set of all uniformly bounded real functions on $\F$, and let $\text{BL}_1(l^\infty (\F))$ denote the set of all functions $h:l^\infty (\F)\rightarrow [0,1]$ such that $|h(z_1)-h(z_2)|\leq \|z_1-z_2\|_\F, \forall z_1, z_2 \in l^\infty (\F)$, where $\|\cdot\|_\F$ is the uniform norm for maps from $\F$ to $\R$.  We define $Pf$ to be the expectation of $f(X)$ when $X \sim P$; as suggested by this notation, throughout this section we will view distributions such as $P$, $\P_n$, and $\P_{n,b}^{(j)}$ as maps from some function class $\F$ to $\R$.  $\eout{E(\cdot)}$ and $\ein{E(\cdot)}$ denote the outer and inner expectation of $(\cdot)$, respectively, and we indicate outer probability via $P^*$.
$X \dequal Y$ denotes that the random variables $X$ and $Y$ are equal in distribution, and $\F_{\delta}$ is defined as the set $\{f-g:f,g\in\F,\rho_P(f-g)<\delta\}$, where $\rho_P(\cdot)$ is the variance semimetric: $\rho_P(f)=\left(P(f-Pf)^2\right)^{1/2}$.

Following prior analyses of the bootstrap~\cite{gine-zinn, vdv-wellner}, we first observe that, conditioning on $\P_{n,b}^{(j)}$ for any $j$ as $b,n \rightarrow \infty$, resamples from the subsampled empirical distribution $\P_{n,b}^{(j)}$ behave asymptotically as though they were drawn directly from $P$, the true underlying distribution:

\begin{lemma}
\label{lemma:ep-consistency}
Given $\P_{n,b}^{(j)}$ for any $j$, let $\s{X}_1, \ldots, \s{X}_n \sim \P_{n,b}^{(j)}$ i.i.d., and define $\s{\P}_{n,b} = n^{-1} \sum_{i=1}^n \delta_{\s{X}_i}$.  Additionally, we define the resampled empirical process $\s{\G}_{n,b} = \sqrt{n}(\s{\P}_{n,b} - \P_{n,b}^{(j)})$.  Then, for $\F$ a Donsker class of measurable functions such that $\F_\delta$ is measurable for every $\delta > 0$, $$\sup_{h \in \text{BL}_1(l^\infty (\F))} \left| E_{\P_{n,b}^{(j)}} h(\s{\G}_{n,b}) - E h(\G_P) \right| \overset{P^*}{\rightarrow} 0,$$
as $n \rightarrow \infty$, for any sequence $b \rightarrow \infty$, where $E_{\P_{n,b}^{(j)}}$ denotes expectation conditional on the contents of the subscript and $\G_P$ is a $P$-Brownian bridge process.  ``Furthermore, the sequence $\eout{E_{\P_{n,b}^{(j)}} h(\s{\G}_{n,b})} - \ein{E_{\P_{n,b}^{(j)}} h(\s{\G}_{n,b})}$ converges to zero in probability for every $h \in \text{BL}_1(l^\infty (\F))$.  If $P^* \| f - Pf \|_\F^2 < \infty$, then the convergence is also outer almost surely.''~\cite{vdv-wellner}
\end{lemma}
\begin{proof}
Note that $\P_{n,b}^{(j)} \dequal \P_b$.  Hence, applying Theorem 3.6.3 of~\citet{vdv-wellner} with the identification $(n, k_n) \leftrightarrow (b, n)$ yields the desired result.
\end{proof}

Lemma~\ref{lemma:ep-consistency} states that, conditionally on the sequence $\P_{n,b}^{(j)}$, the sequence of processes $\s{\G}_{n,b}$ converges in distribution to the $P$-Brownian bridge process $\G_P$, in probability.  Noting that the empirical process $\G_n = \sqrt{n}(\P_n - P)$ also converges in distribution to $\G_P$ (recall that $\F$ is a Donsker class by assumption), it follows that size $n$ resamples generated from $\P_{n,b}^{(j)}$ behave asymptotically as though they were drawn directly from $P$.  Under standard assumptions, it then follows that $\xi(Q_n(\P_{n,b}^{(j)})) - \xi(Q_n(P)) \pconv 0$:

\begin{lemma}
\label{lemma:oneiter-consistency}
Under the assumptions of Theorem~\ref{thm:consistency}, for any $j$,
$$\xi(Q_n(\P_{n,b}^{(j)})) - \xi(Q_n(P)) \pconv 0$$
as $n \rightarrow \infty$, for any sequence $b \rightarrow \infty$.
\end{lemma}
\begin{proof}
Let $R$ be the random element to which $\sqrt{n}(\phi(\P_n) - \phi(P))$ converges in distribution; note that the functional delta method~\cite{vdv} provides the form of $R$ in terms of $\phi$ and $P$.  The delta method for the bootstrap (see Theorem 23.9 of~\citet{vdv}) in conjunction with Lemma~\ref{lemma:ep-consistency} implies that, under our assumptions, $\sqrt{n}(\phi(\s{\P}_{n,b}) - \phi(\P_{n,b}^{(j)}))$ also converges conditionally in distribution to $R$, given $\P_{n,b}^{(j)}$, in probability.  Thus, the distribution of $\sqrt{n}(\phi(\P_n) - \phi(P))$ and the distribution of $\sqrt{n}(\phi(\s{\P}_{n,b}) - \phi(\P_{n,b}^{(j)}))$, the latter conditionally on $\P_{n,b}^{(j)}$, have the same asymptotic limit in probability.  As a result, given the assumed continuity of $\xi$, it follows that $\xi(Q_n(\P_{n,b}^{(j)}))$ and $\xi(Q_n(P))$ have the same asymptotic limit, in probability.
\end{proof}

The above lemma indicates that each individual $\xi(Q_n(\P_{n,b}^{(j)}))$ is asymptotically consistent as $b, n \rightarrow \infty$.  Theorem~\ref{thm:consistency} immediately follows:

\begin{proof}[Proof of Theorem~\ref{thm:consistency}]
Lemma~\ref{lemma:oneiter-consistency} in conjunction with the continuous mapping theorem~\cite{vdv} implies the desired result.
\end{proof}

\subsection{Higher-Order Correctness}
We first prove two supporting lemmas.

\begin{lemma}
\label{lemma:marginal-var}
Assume that $X_1, \ldots, X_b \sim P$ are i.i.d., and let $\sest{p}_k(X_1, \ldots, X_b)$ be the sample version of $p_k$ based on $X_1, \ldots, X_b$, as defined in Theorem~\ref{thm:higher-order}.  Then, assuming that $E[\sest{p}_k(X_1, \ldots, X_b)^2] < \infty$,
$$\Var(\sest{p}_k(X_1, \ldots, X_b) - p_k) = \Var(\sest{p}_k(X_1, \ldots, X_b)) = O(1/b).$$
\end{lemma}
\begin{proof}
By definition, the $\sest{p}_k$ are simply polynomials in sample moments.  Thus, we can write
\begin{equation}
\label{eq:pk-genform}
\sest{p}_k = \sest{p}_k(X_1,\dots,X_b) = \sum_{\beta=1}^B c_\beta \prod_{\alpha=1}^{A_\beta}\left(b^{-1}\sum_{i=1}^b g^{(\beta)}_\alpha(X_i)\right),
\end{equation}
where each $g^{(\beta)}_\alpha$ raises its argument to some power.  Now, observe that for any $\beta$,
$$V_\beta = \prod_{\alpha=1}^{A_\beta}\left(b^{-1}\sum_{i=1}^b g^{(\beta)}_\alpha(X_i)\right)$$
is a V-statistic of order $A_\beta$ applied to the $b$ observations $X_1, \ldots, X_b$.  Let $h_\beta(x_1, \ldots, x_{A_\beta})$ denote the kernel of this V-statistic, which is a symmetrized version of $\prod_{\alpha=1}^{A_\beta} g^{(\beta)}_\alpha(x_\alpha)$.  It follows that $\sest{p}_k = \sum_{\beta=1}^B c_\beta V_\beta$ is itself a V-statistic of order $A = \max_\beta A_\beta$ with kernel $h(x_1, \ldots, x_A) = \sum_{\beta=1}^B c_\beta h_\beta(x_1, \ldots, x_{A_\beta})$, applied to the $b$ observations $X_1, \ldots, X_b$.  Let $U$ denote the corresponding U-statistic having kernel $h$.  Then, using Proposition~3.5(ii) and Corollary~3.2(i) of~\citet{shao}, we have
$$\Var(\sest{p}_k - p_k) = \Var(\sest{p}_k) = \Var(U) + O(b^{-2}) \leq \frac{A}{b} \Var(h(X_1, \ldots, X_A)) + O(b^{-2}) = O(1/b).$$
\end{proof}

\begin{lemma}
\label{lemma:expectation-rate}
Assume that $X_1, \ldots, X_b \sim P$ are i.i.d., and let $\sest{p}_k(X_1, \ldots, X_b)$ be the sample version of $p_k$ based on $X_1, \ldots, X_b$, as defined in Theorem~\ref{thm:higher-order}.  Then, assuming that $E|\sest{p}_k(X_1, \ldots, X_b)| < \infty$,
$$| E[\sest{p}_k(X_1, \ldots, X_b)] - p_k | = O(1/b).$$
\end{lemma}
\begin{proof}
As noted in the proof of Lemma~\ref{lemma:marginal-var}, we can write
$$\sest{p}_k(X_1,\dots,X_b)=\sum_{\beta=1}^B c_\beta \prod_{\alpha=1}^{A_\beta}\left(b^{-1}\sum_{i=1}^b g^{(\beta)}_\alpha(X_i)\right),$$
where each $g^{(\beta)}_\alpha$ raises its argument to some power.  Similarly,
$$p_k = \sum_{\beta=1}^B c_\beta \prod_{\alpha=1}^{A_\beta} E g^{(\beta)}_\alpha(X_1),$$
and so
\begin{align*}
| E[\sest{p}_k(X_1, \ldots, X_b)] - p_k | &= \left| \sum_{\beta=1}^B c_\beta \prod_{\alpha=1}^{A_\beta}\left(b^{-1}\sum_{i=1}^b g^{(\beta)}_\alpha(X_i)\right) - \sum_{\beta=1}^B c_\beta \prod_{\alpha=1}^{A_\beta} E g^{(\beta)}_\alpha(X_1) \right| \\
&\leq \sum_{\beta=1}^B | c_\beta | \cdot \left| E\left[\prod_{\alpha=1}^{A_\beta}\left(b^{-1}\sum_{i=1}^b g^{(\beta)}_\alpha(X_i)\right) \right] - \prod_{\alpha=1}^{A_\beta} E g^{(\beta)}_\alpha(X_1) \right|.
\end{align*}
Given that the number of terms in the outer sum on the right-hand side is constant with respect to $b$, to prove the desired result it is sufficient to show that, for any $\beta$,
$$\Delta_\beta = \left| E\left[\prod_{\alpha=1}^{A_\beta}\left(b^{-1}\sum_{i=1}^b g^{(\beta)}_\alpha(X_i)\right) \right] - \prod_{\alpha=1}^{A_\beta} E g^{(\beta)}_\alpha(X_1) \right| = O\left(\frac{1}{b}\right).$$
Observe that
\begin{equation}
\label{eq:moment-exp-1}
E\left[\prod_{\alpha=1}^{A_\beta}\left(b^{-1}\sum_{i=1}^b g^{(\beta)}_\alpha(X_i)\right) \right] = b^{-A_\beta} E\left[\sum_{i_1,\dots,i_{A_\beta} = 1}^b \; \prod_{\alpha=1}^{A_\beta} g^{(\beta)}_\alpha(X_{i_\alpha}) \right].
\end{equation}
If $i_1, \ldots, i_{A_\beta}$ are all distinct, then $E \prod_{\alpha=1}^{A_\beta} g^{(\beta)}_\alpha(X_{i_\alpha}) = \prod_{\alpha=1}^{A_\beta} E g^{(\beta)}_\alpha(X_1)$ because $X_1, \ldots, X_b$ are i.i.d..  Additionally, the right-hand summation in~\eq{moment-exp-1} has $b! / (b-A_\beta)!$ terms in which $i_1, \ldots, i_{A_\beta}$ are all distinct; correspondingly, there are $b^{A_\beta} - b! / (b-A_\beta)!$ terms in which $\exists \alpha, \alpha' \text{ s.t.\ } i_\alpha=i_{\alpha'}$.  Therefore, it follows that
$$E\left[\prod_{\alpha=1}^{A_\beta}\left(b^{-1}\sum_{i=1}^b g^{(\beta)}_\alpha(X_i)\right) \right] = b^{-A_\beta} \left[ \frac{b!}{(b-A_\beta)!} \prod_{\alpha=1}^{A_\beta} E g^{(\beta)}_\alpha(X_1) \quad + \sum_{\substack{1 \leq i_1,\dots,i_{A_\beta} \leq b \\ \exists \alpha, \alpha' \text{ s.t.\ } i_\alpha = i_{\alpha'}}} E \prod_{\alpha=1}^{A_\beta} g^{(\beta)}_\alpha(X_{i_\alpha}) \right]$$
and
\begin{align}
\Delta_\beta &= \left| E\left[\prod_{\alpha=1}^{A_\beta}\left(b^{-1}\sum_{i=1}^b g^{(\beta)}_\alpha(X_i)\right) \right] - \prod_{\alpha=1}^{A_\beta} E g^{(\beta)}_\alpha(X_1) \right| \nonumber \\
&= b^{-A_\beta} \left| \left(\frac{b!}{(b-A_\beta)!} - b^{A_\beta}\right) \prod_{\alpha=1}^{A_\beta} E g^{(\beta)}_\alpha(X_1) \quad + \sum_{\substack{1 \leq i_1,\dots,i_{A_\beta} \leq b \\ \exists \alpha, \alpha' \text{ s.t.\ } i_\alpha = i_{\alpha'}}} E \prod_{\alpha=1}^{A_\beta} g^{(\beta)}_\alpha(X_{i_\alpha}) \right| \nonumber \\
&\leq b^{-A_\beta} \left| \frac{b!}{(b-A_\beta)!} - b^{A_\beta}\right| \cdot \left| \prod_{\alpha=1}^{A_\beta} E g^{(\beta)}_\alpha(X_1) \right| \; + \; b^{-A_\beta} \left| b^{A_\beta} - \frac{b!}{(b-A_\beta)!} \right| C, \label{eq:delta-upper-bound}
\end{align}
where
$$C \quad = \max_{\substack{1 \leq i_1,\dots,i_{A_\beta} \leq b \\ \exists \alpha, \alpha' \text{ s.t.\ } i_\alpha = i_{\alpha'}}} \left| E \prod_{\alpha=1}^{A_\beta} g^{(\beta)}_\alpha(X_{i_\alpha}) \right|.$$
Note that $C$ is a constant with respect to $b$.  Also, simple algebraic manipulation shows that
$$\frac{b!}{(b-A_\beta)!} = b^{A_\beta} - \kappa b^{A_\beta - 1} + O(b^{A_\beta - 2})$$
for some $\kappa > 0$.  Thus, plugging into equation~\eq{delta-upper-bound}, we obtain the desired result:
$$\Delta_\beta \leq b^{-A_\beta} \left| \kappa b^{A_\beta - 1} + O(b^{A_\beta - 2}) \right| \cdot \left| \prod_{\alpha=1}^{A_\beta} E g^{(\beta)}_\alpha(X_1) \right| + b^{-A_\beta} \left| \kappa b^{A_\beta - 1} + O(b^{A_\beta - 2}) \right| C = O\left(\frac{1}{b}\right).$$
\end{proof}

We now provide full proofs of Theorem~\ref{thm:higher-order}, Remark~\ref{remark:condvar-rate}, and Theorem~\ref{thm:higher-order-disjoint}.

\begin{proof}[Proof of Theorem~\ref{thm:higher-order}]
Summing the expansion~\eq{sample-expansion} over $j$, we have
$$\numsub^{-1} \sum_{j=1}^{\numsub} \xi(Q_n(\P_{n,b}^{(j)})) = z + n^{-1/2} \numsub^{-1} \sum_{j=1}^{\numsub} \sest{p}^{(j)}_1 + n^{-1} \numsub^{-1} \sum_{j=1}^{\numsub} \sest{p}^{(j)}_2 + o_P\left(\frac{1}{n}\right).$$
Subtracting the corresponding expansion~\eq{pop-expansion} for $\xi(Q_n(P))$, we then obtain
\begin{equation}
\label{eq:edgeworth-dev}
\left| \numsub^{-1} \sum_{j=1}^{\numsub} \xi(Q_n(\P_{n,b}^{(j)})) - \xi(Q_n(P)) \right| \leq n^{-1/2} \left| \numsub^{-1} \sum_{j=1}^{\numsub} \sest{p}^{(j)}_1 - p_1 \right| + n^{-1} \left| \numsub^{-1} \sum_{j=1}^{\numsub} \sest{p}^{(j)}_2 - p_2 \right| + o_P\left(\frac{1}{n}\right).
\end{equation}
We now further analyze the first two terms on the right-hand side of the above expression; for the remainder of this proof, we assume that $k \in \{1,2\}$.  Observe that, for fixed $k$, the $\sest{p}^{(j)}_k$ are conditionally i.i.d.\ given $X_1, \ldots, X_n$ for all $j$, and so
$$\Var\left( \left. \numsub^{-1} \sum_{j=1}^{\numsub} \left([\sest{p}^{(j)}_k - p_k] - E[\sest{p}^{(1)}_k - p_k | \P_n]\right) \right| \P_n \right) = \frac{\Var(\sest{p}^{(1)}_k - p_k | \P_n)}{\numsub},$$
where we denote by $E[\sest{p}^{(1)}_k - p_k | \P_n]$ and $\Var(\sest{p}^{(1)}_k - p_k | \P_n)$ the expectation and variance of $\sest{p}^{(1)}_k - p_k$ over realizations of $\P_{n,b}^{(1)}$ conditionally on $X_1, \ldots, X_n$.  Now, given that $\sest{p}^{(j)}_k$ is a permutation-symmetric function of size $b$ subsets of $X_1, \ldots, X_n$, $E[\sest{p}^{(1)}_k - p_k | \P_n]$ is a U-statistic of order $b$.  Hence, we can apply Corollary~3.2(i) of~\citet{shao} in conjunction with Lemma~\ref{lemma:marginal-var} to find that
$$\Var\left( E[\sest{p}^{(1)}_k - p_k | \P_n] - E[\sest{p}^{(1)}_k - p_k] \right) = \Var\left( E[\sest{p}^{(1)}_k - p_k | \P_n] \right) \leq \frac{b}{n} \Var(\sest{p}^{(1)}_k - p_k) = O\left( \frac{1}{n} \right).$$
From the result of Lemma~\ref{lemma:expectation-rate}, we have
$$| E[\sest{p}^{(1)}_k - p_k] | = O\left(\frac{1}{b}\right).$$
Combining the expressions in the previous three panels, we find that
$$\left| \numsub^{-1} \sum_{j=1}^{\numsub} \sest{p}^{(j)}_k - p_k \right| = O_P\left(\frac{\sqrt{\Var(\sest{p}^{(1)}_k - p_k | \P_n)}}{\sqrt{\numsub}}\right) + O_P\left( \frac{1}{\sqrt{n}} \right) + O\left(\frac{1}{b}\right).$$
Finally, plugging into equation~\eq{edgeworth-dev} with $k=1$ and $k=2$, we obtain the desired result.
\end{proof}

\begin{proof}[Proof of Remark~\ref{remark:condvar-rate}]
Observe that
$$\Var(\sest{p}^{(1)}_k - p_k | \P_n) \leq E[(\sest{p}^{(1)}_k - p_k)^2 | \P_n] = E[(\sest{p}^{(1)}_k - p_k)^2 | \P_n] - E[(\sest{p}^{(1)}_k - p_k)^2] + E[(\sest{p}^{(1)}_k - p_k)^2].$$
Given that $\sest{p}^{(1)}_k$ is a polynomial in the moments of $\P_{n,b}^{(1)}$, $\sest{q}^{(1)}_k = (\sest{p}^{(1)}_k - p_k)^2$ is also a polynomial in the moments of $\P_{n,b}^{(1)}$.  Hence, Lemma~\ref{lemma:marginal-var} applies to $\sest{q}^{(1)}_k$.  Additionally, $\sest{q}^{(1)}_k$ is a permutation-symmetric function of size $b$ subsets of $X_1, \ldots, X_n$, and so $E[\sest{q}^{(1)}_k | \P_n]$ is a U-statistic of order $b$.  Therefore, applying Corollary~3.2(i) of~\citet{shao} in conjunction with Lemma~\ref{lemma:marginal-var}, we find that
$$\Var\left( E[(\sest{p}^{(1)}_k - p_k)^2 | \P_n] - E[(\sest{p}^{(1)}_k - p_k)^2] \right) = \Var\left( E[\sest{q}^{(1)}_k | \P_n] \right) \leq \frac{b}{n} \Var( \sest{q}^{(1)}_k ) = O\left(\frac{1}{n}\right).$$
Now,
$$E[(\sest{p}^{(1)}_k - p_k)^2] = \Var(\sest{p}^{(1)}_k - p_k) + E[\sest{p}^{(1)}_k - p_k]^2.$$
By Lemmas~\ref{lemma:marginal-var} and~\ref{lemma:expectation-rate}, $\Var(\sest{p}^{(1)}_k - p_k) = O(1/b)$ and $E[\sest{p}^{(1)}_k - p_k]^2 = O(1/b^2)$.  Combining with the expressions in the previous three panels, we obtain the desired result:
$$\Var(\sest{p}^{(1)}_k - p_k | \P_n) = O_P\left(\frac{1}{\sqrt{n}}\right) + O\left(\frac{1}{b}\right) + O\left(\frac{1}{b^2}\right) = O_P\left(\frac{1}{\sqrt{n}}\right) + O\left(\frac{1}{b}\right).$$
\end{proof}

\begin{proof}[Proof of Theorem~\ref{thm:higher-order-disjoint}]
As noted in the proof of Theorem~\ref{thm:higher-order},
\begin{equation}
\label{eq:edgeworth-dev-disjoint}
\left| \numsub^{-1} \sum_{j=1}^{\numsub} \xi(Q_n(\P_{n,b}^{(j)})) - \xi(Q_n(P)) \right| \leq n^{-1/2} \left| \numsub^{-1} \sum_{j=1}^{\numsub} \sest{p}^{(j)}_1 - p_1 \right| + n^{-1} \left| \numsub^{-1} \sum_{j=1}^{\numsub} \sest{p}^{(j)}_2 - p_2 \right| + o_P\left(\frac{1}{n}\right).
\end{equation}
Throughout this proof, we assume that $k \in \{1,2\}$.  Under the assumptions of this theorem, the $\P_{n,b}^{(j)}$ are based on disjoint subsets of the $n$ observations and so are i.i.d..  Hence, for any $k$, the $\sest{p}^{(j)}_k$ are i.i.d.\ for all $j$, and so using Lemma~\ref{lemma:marginal-var},
$$\Var\left( \left[ \numsub^{-1} \sum_{j=1}^{\numsub} \sest{p}^{(j)}_k - p_k \right] - E[\sest{p}^{(1)}_k - p_k] \right) = \frac{\Var(\sest{p}^{(1)}_k - p_k)}{s} = O\left(\frac{1}{b \numsub}\right).$$
Additionally, from the result of Lemma~\ref{lemma:expectation-rate}, we have
$$| E[\sest{p}^{(1)}_k - p_k] | = O\left(\frac{1}{b}\right).$$
Combining the expressions in the previous two panels, we find that
$$\left| \numsub^{-1} \sum_{j=1}^{\numsub} \sest{p}^{(j)}_k - p_k \right| = O_P\left(\frac{1}{\sqrt{b \numsub}}\right) + O\left(\frac{1}{b}\right).$$
Finally, plugging into equation~\eq{edgeworth-dev-disjoint} with $k=1$ and $k=2$, we obtain the desired result.
\end{proof}

\section{Appendix: Additional Real Data Results}
\label{sec:appendix-real-data}

\begin{figure}[ht]
\includegraphics[width=0.5\linewidth]{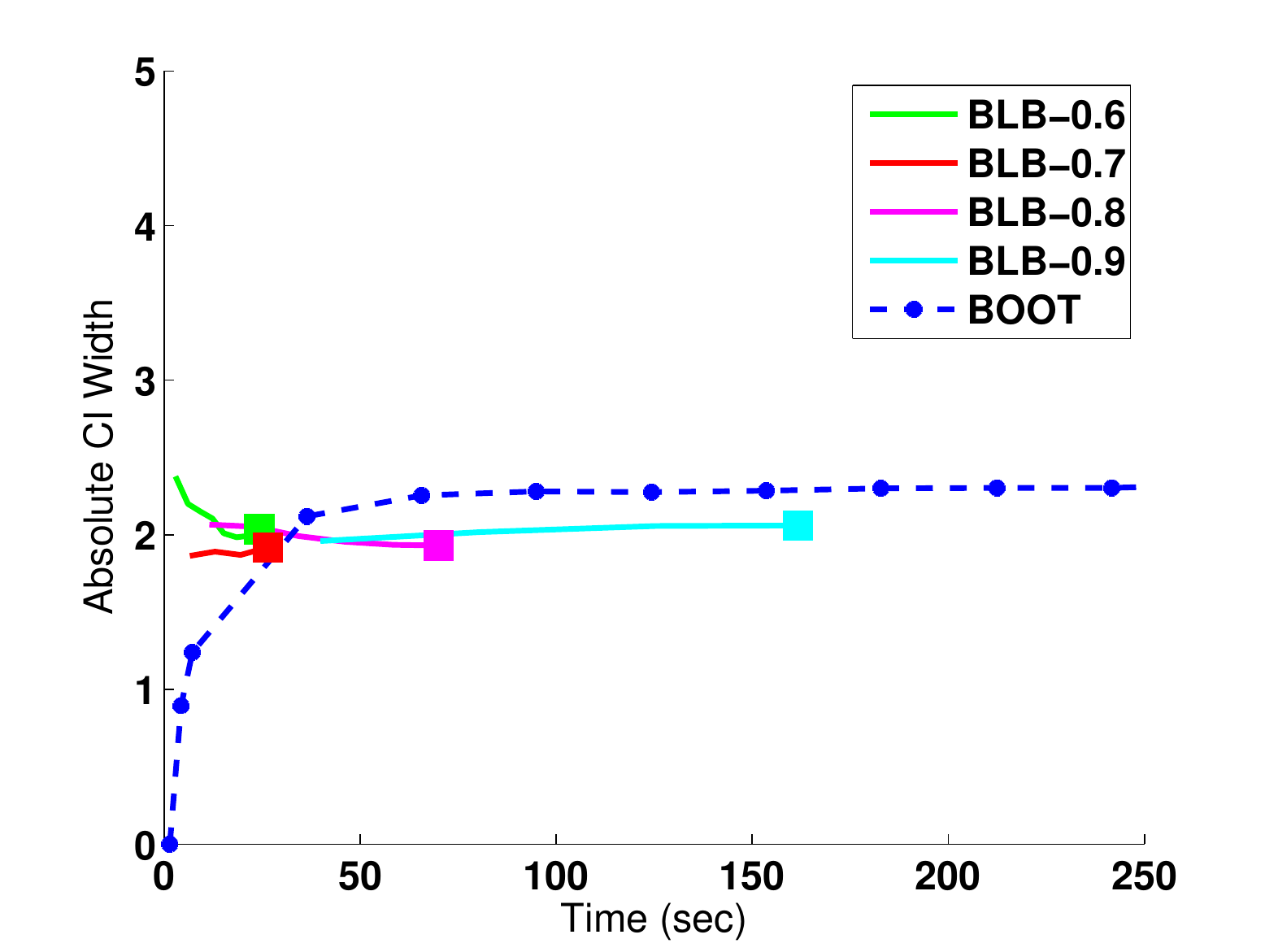}
\includegraphics[width=0.5\linewidth]{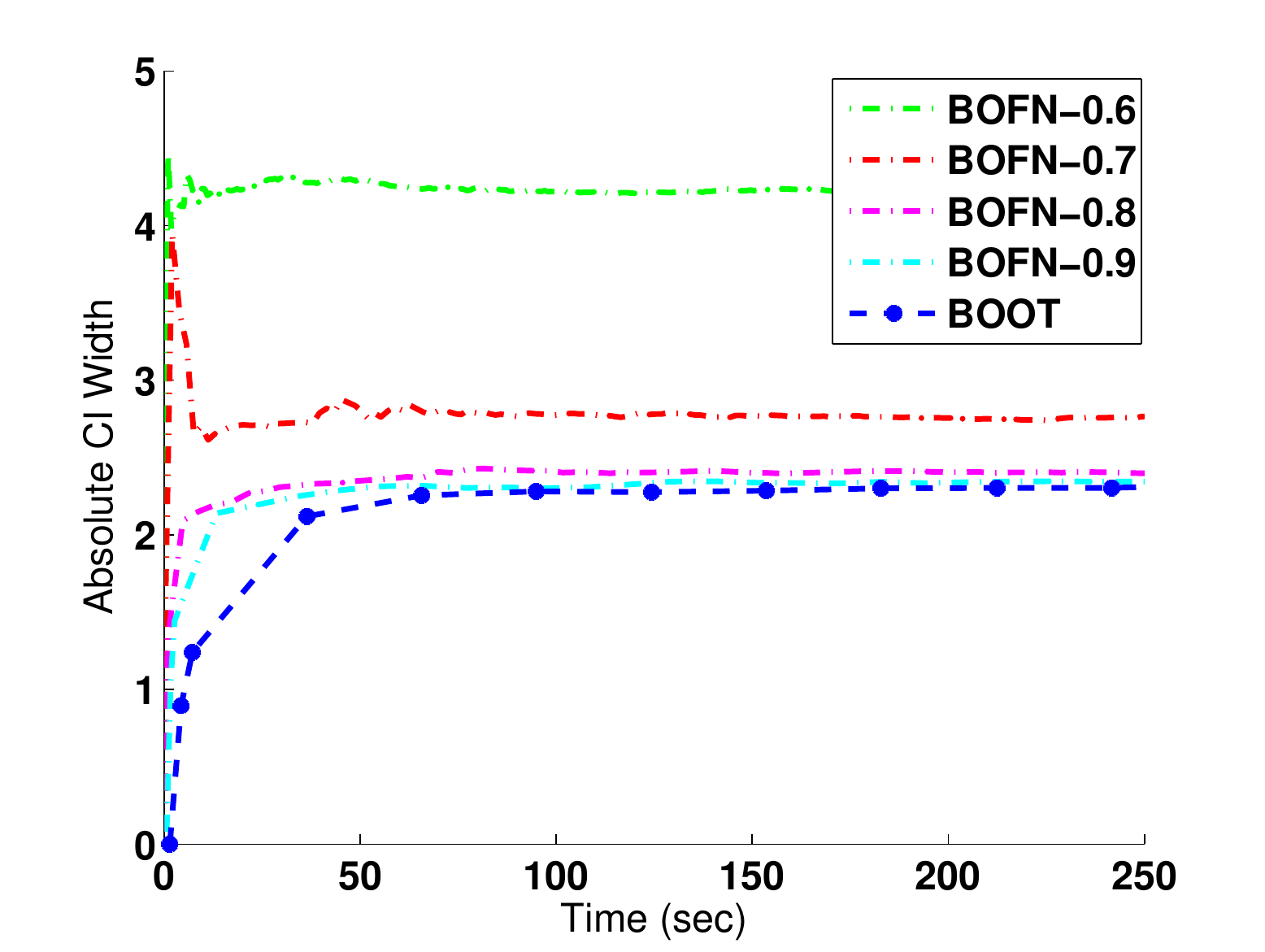}
\caption{Average (across dimensions) absolute confidence interval width vs.\ processing time on the UCI ct-slice dataset (linear regression, $d=385$, $n= 53,500$).  The left plot shows results for \ouralgAbbrev (using adaptive hyperparameter selection, with the output at convergence marked by large squares) and the bootstrap (BOOT).  The right plot shows results for the $b$ out of $n$ bootstrap (BOFN).}
\end{figure}

\begin{figure}[ht]
\includegraphics[width=0.5\linewidth]{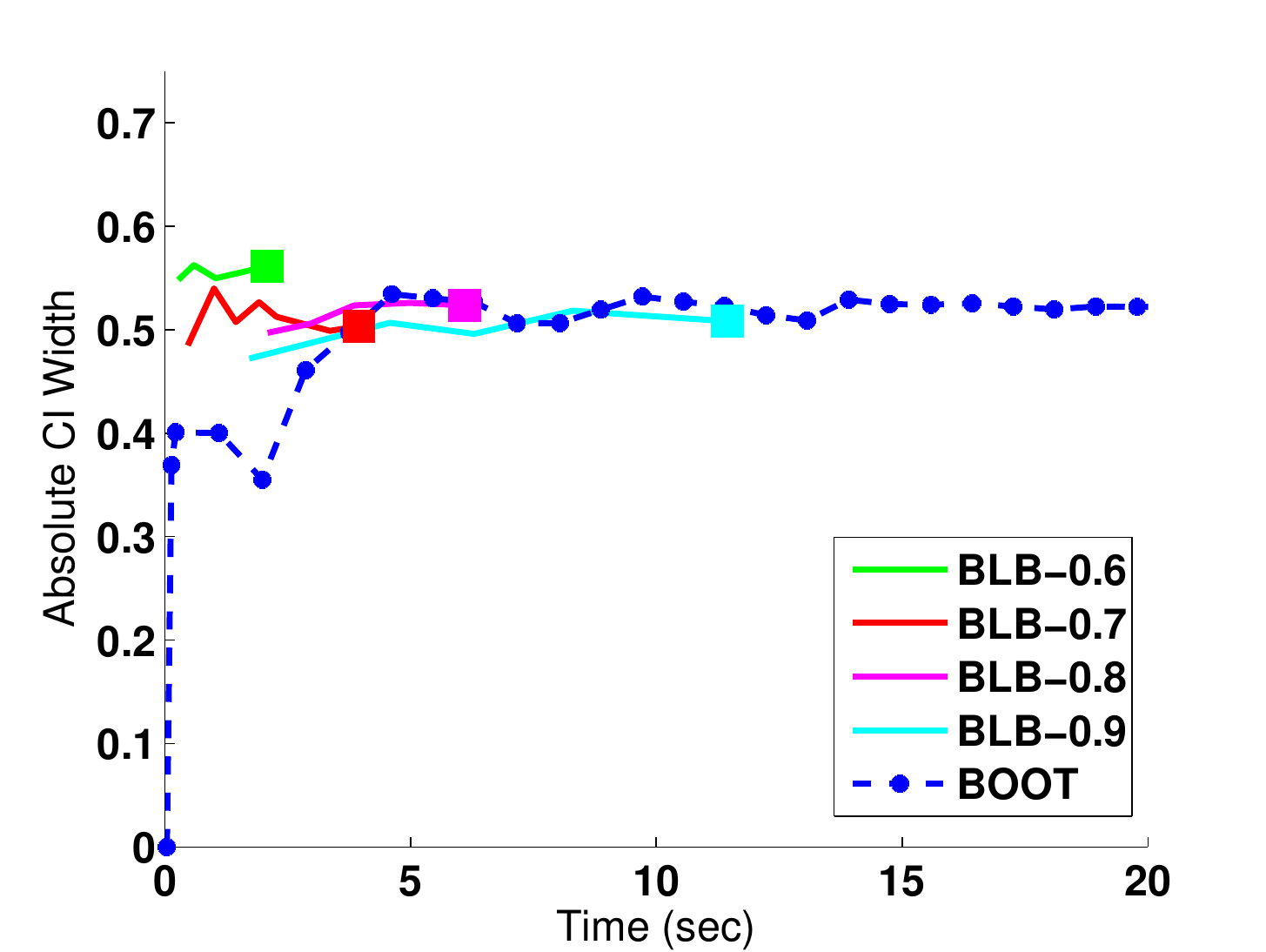}
\includegraphics[width=0.5\linewidth]{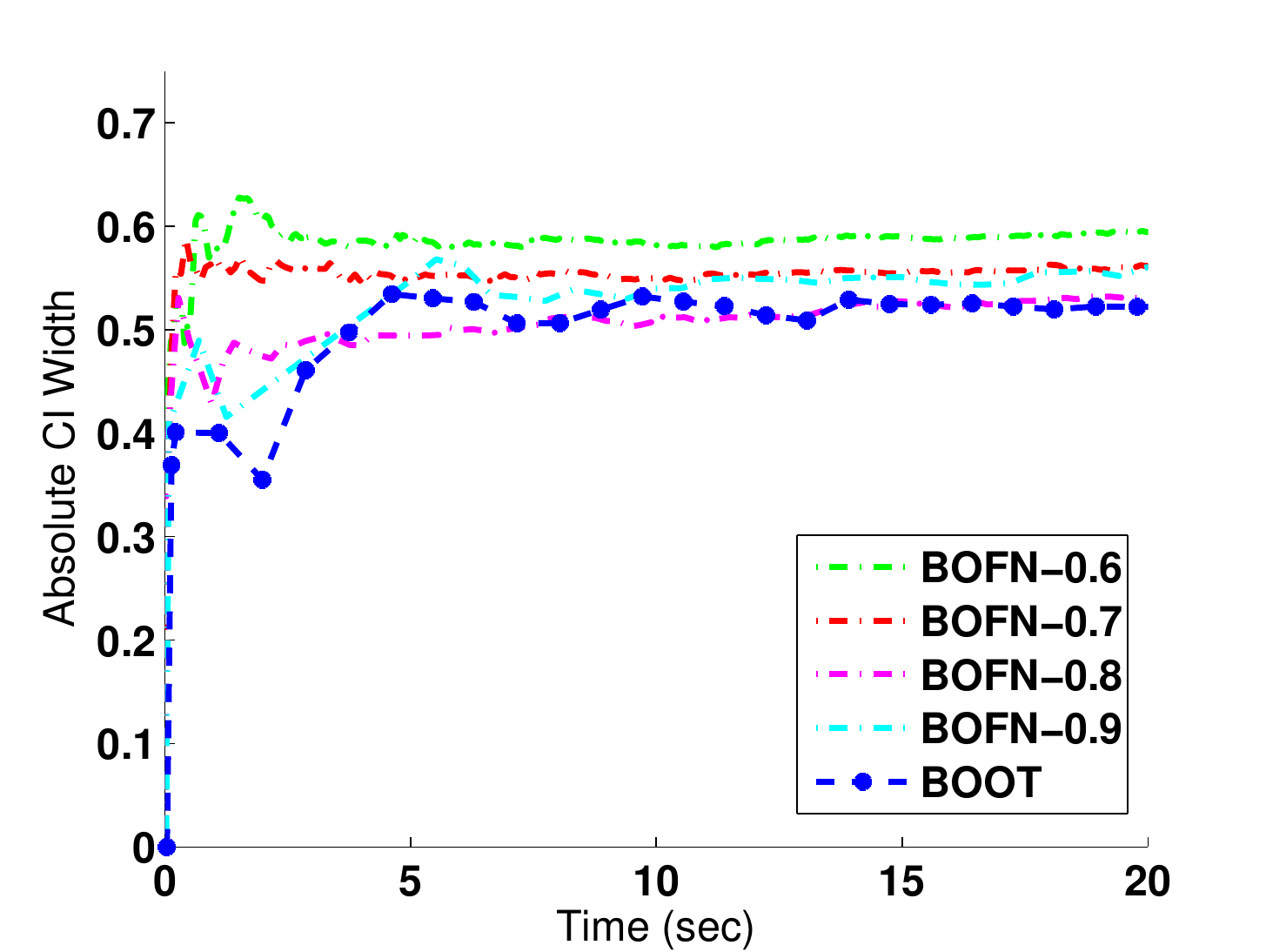}
\caption{Average (across dimensions) absolute confidence interval width vs.\ processing time on the UCI magic dataset (logistic regression, $d=10$, $n= 19,020$).  The left plot shows results for \ouralgAbbrev (using adaptive hyperparameter selection, with the output at convergence marked by large squares) and the bootstrap (BOOT).  The right plot shows results for the $b$ out of $n$ bootstrap (BOFN).}
\end{figure}

\begin{figure}[ht]
\includegraphics[width=0.5\linewidth]{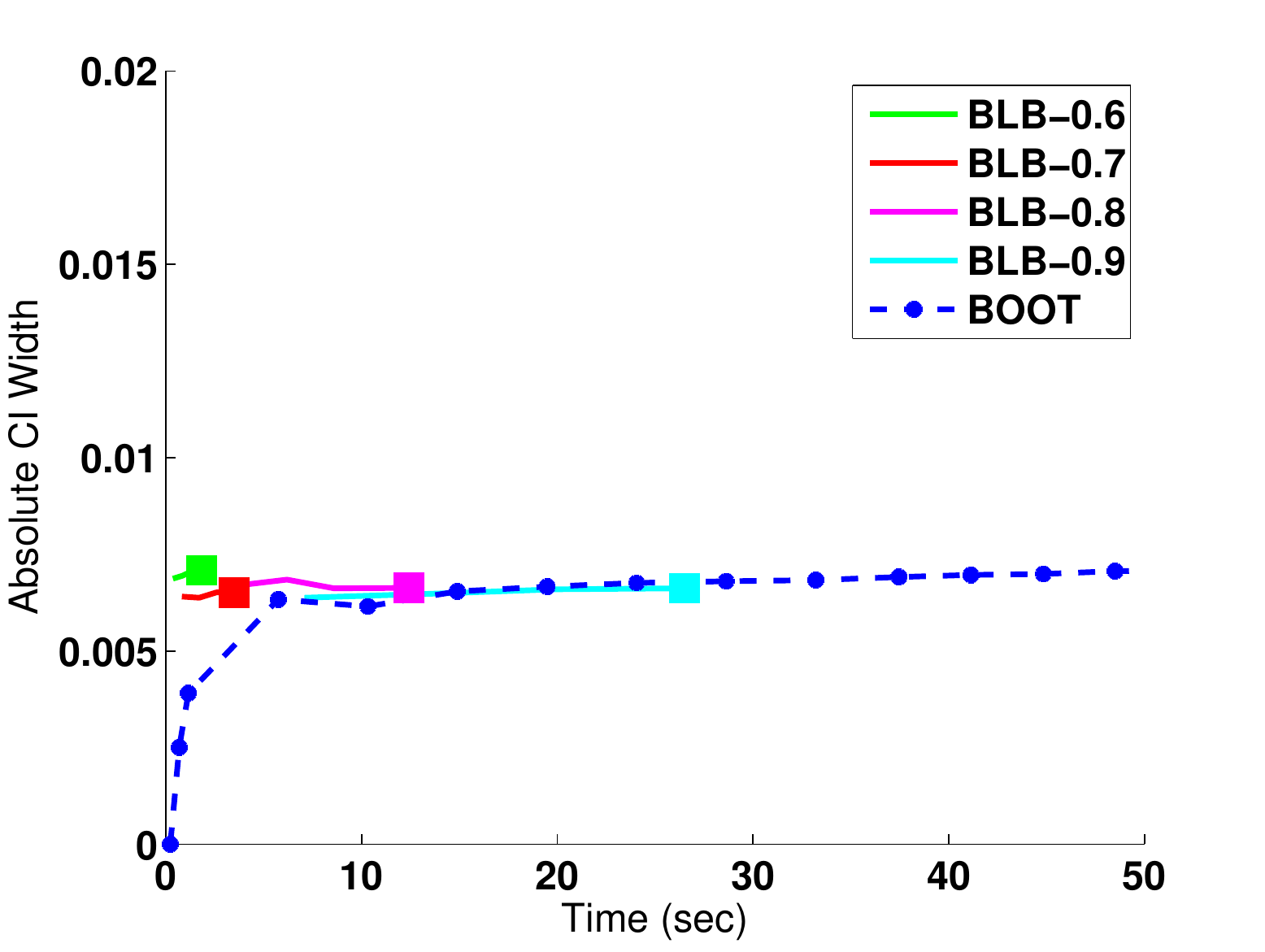}
\includegraphics[width=0.5\linewidth]{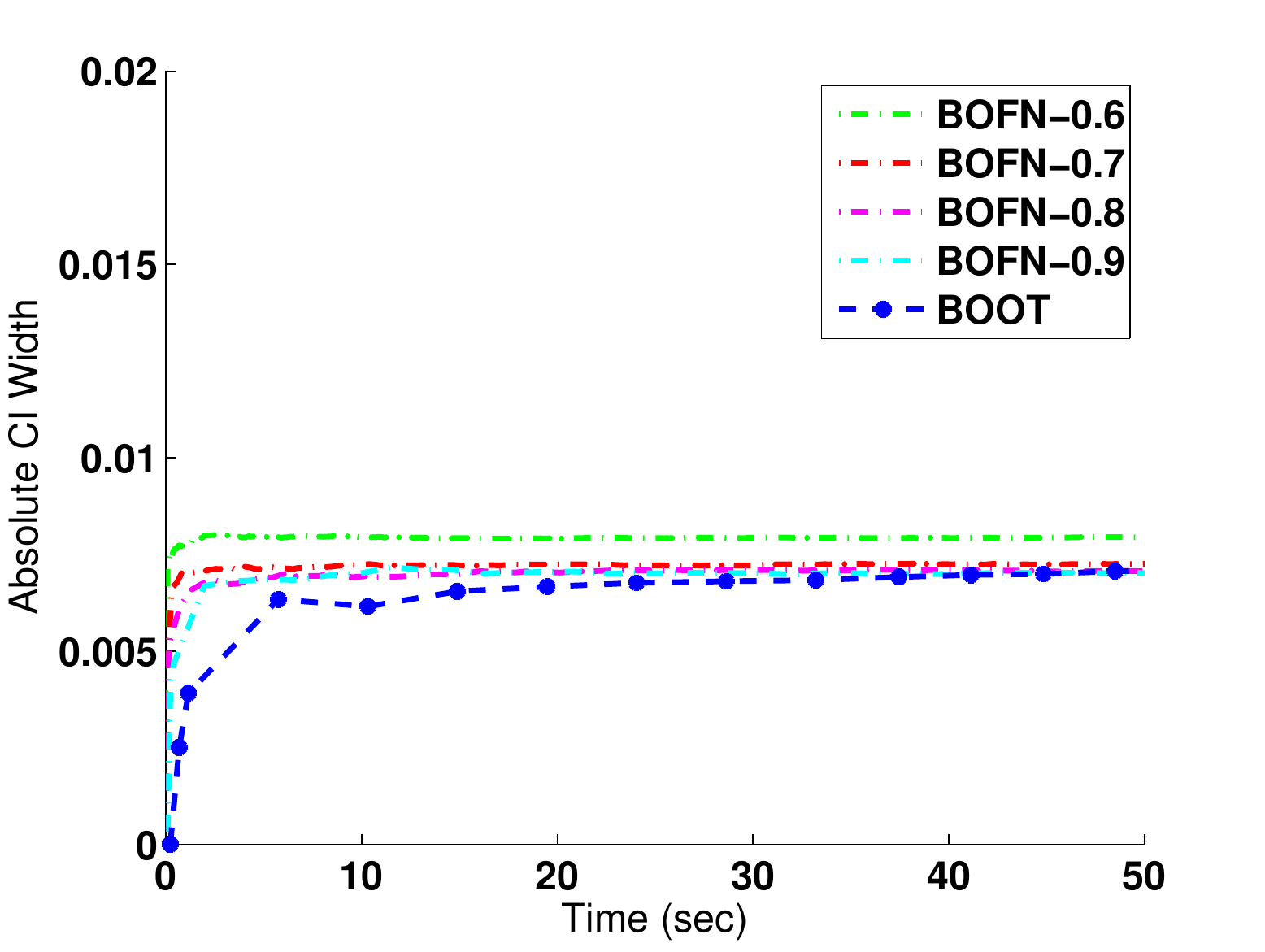}
\caption{Average (across dimensions) absolute confidence interval width vs.\ processing time on the UCI millionsong dataset (linear regression, $d=90$, $n= 50,000$).  The left plot shows results for \ouralgAbbrev (using adaptive hyperparameter selection, with the output at convergence marked by large squares) and the bootstrap (BOOT).  The right plot shows results for the $b$ out of $n$ bootstrap (BOFN).}
\end{figure}

\begin{figure}[ht]
\includegraphics[width=0.5\linewidth]{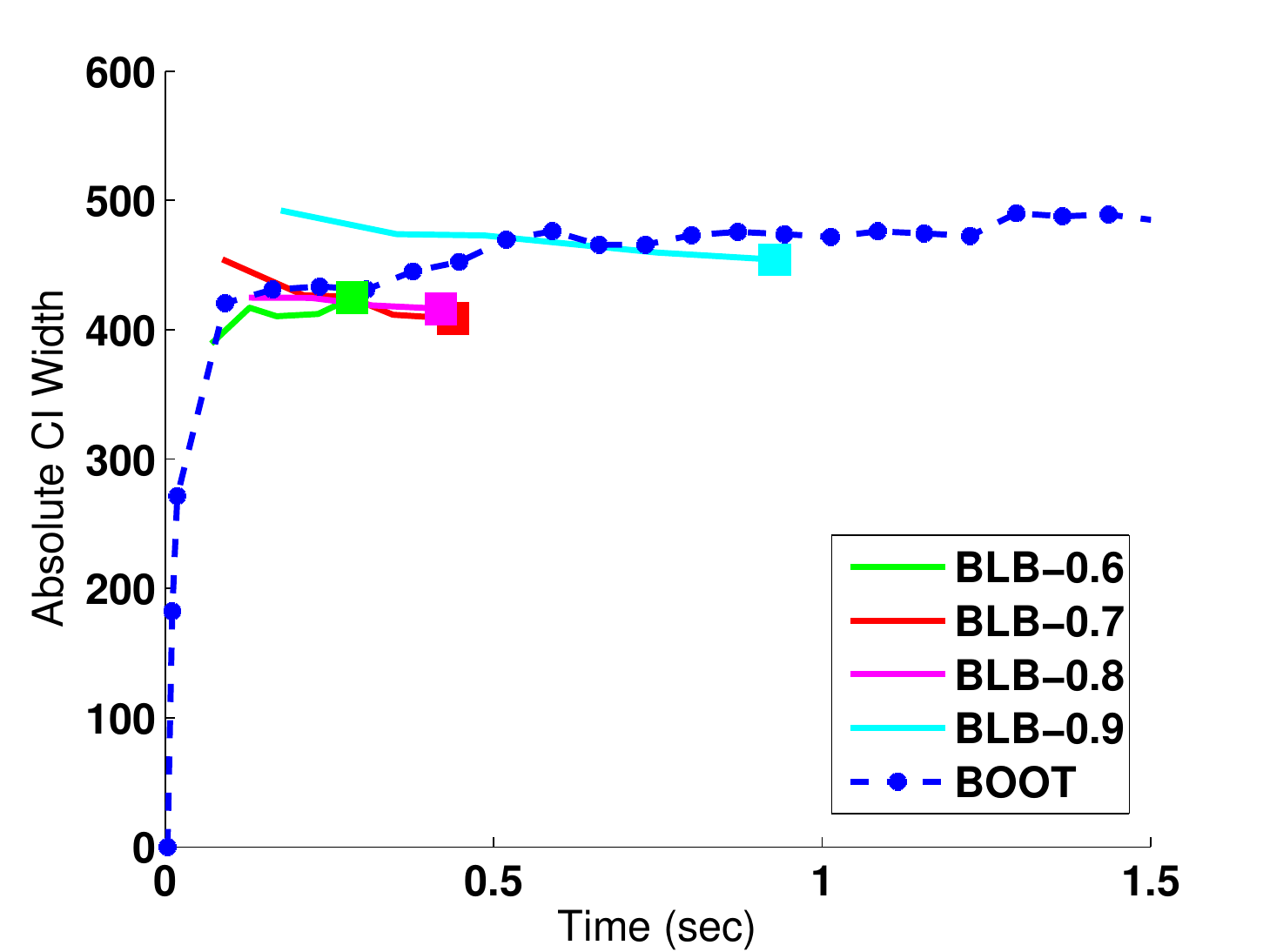}
\includegraphics[width=0.5\linewidth]{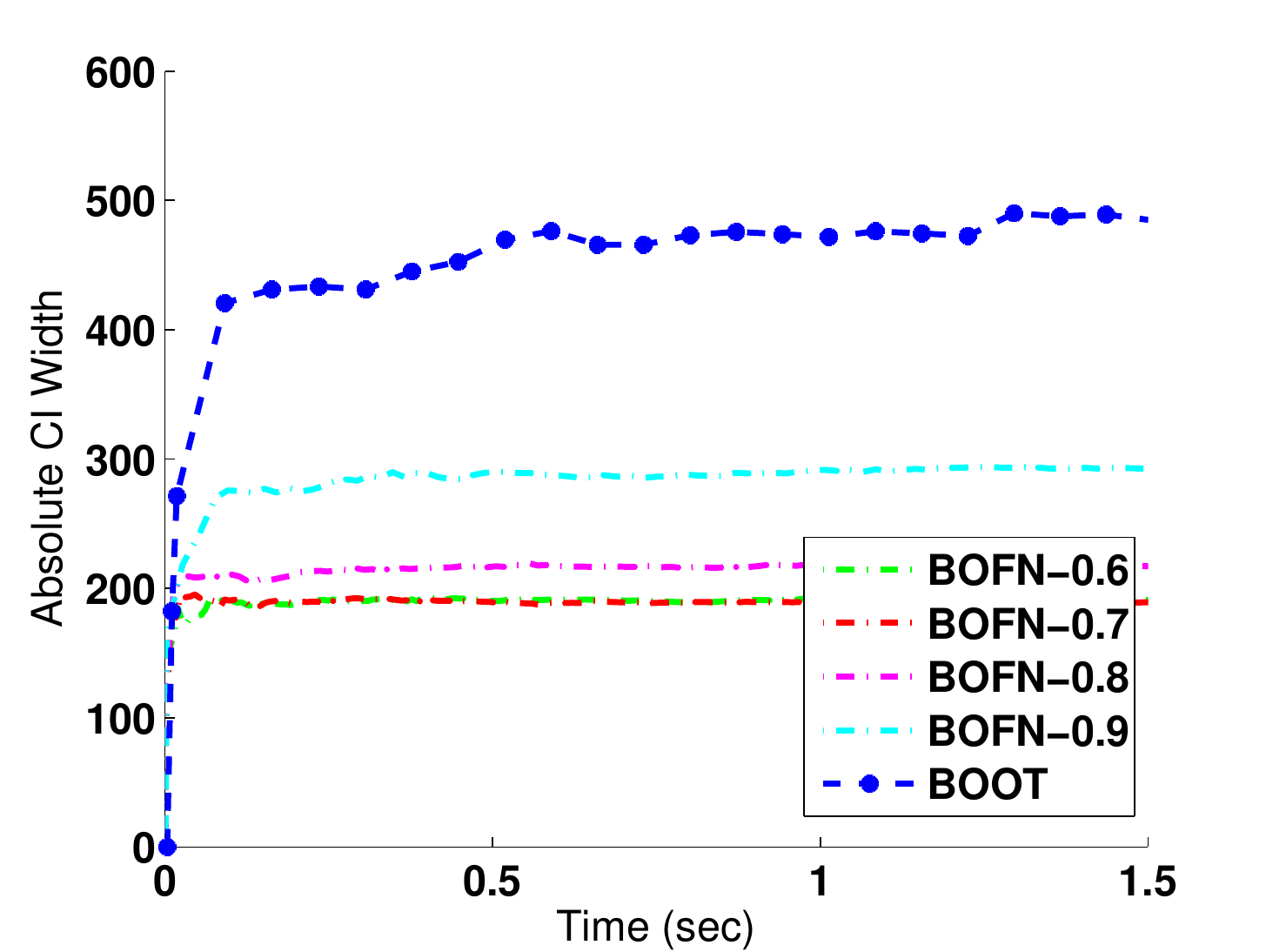}
\caption{Average (across dimensions) absolute confidence interval width vs.\ processing time on the UCI parkinsons dataset (linear regression, $d=16$, $n= 5,875$).  The left plot shows results for \ouralgAbbrev (using adaptive hyperparameter selection, with the output at convergence marked by large squares) and the bootstrap (BOOT).  The right plot shows results for the $b$ out of $n$ bootstrap (BOFN).}
\end{figure}

\begin{figure}[ht]
\includegraphics[width=0.5\linewidth]{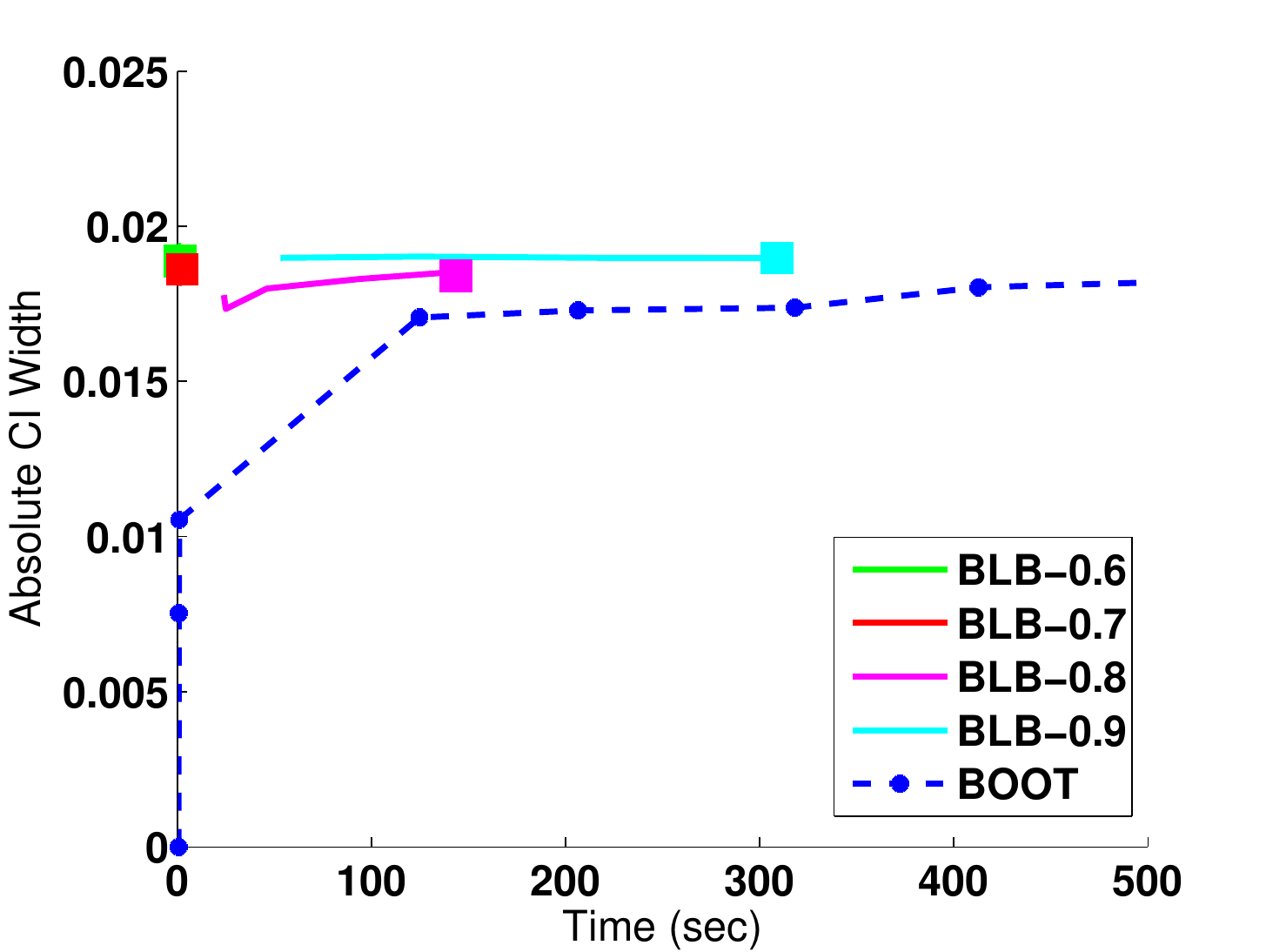}
\includegraphics[width=0.5\linewidth]{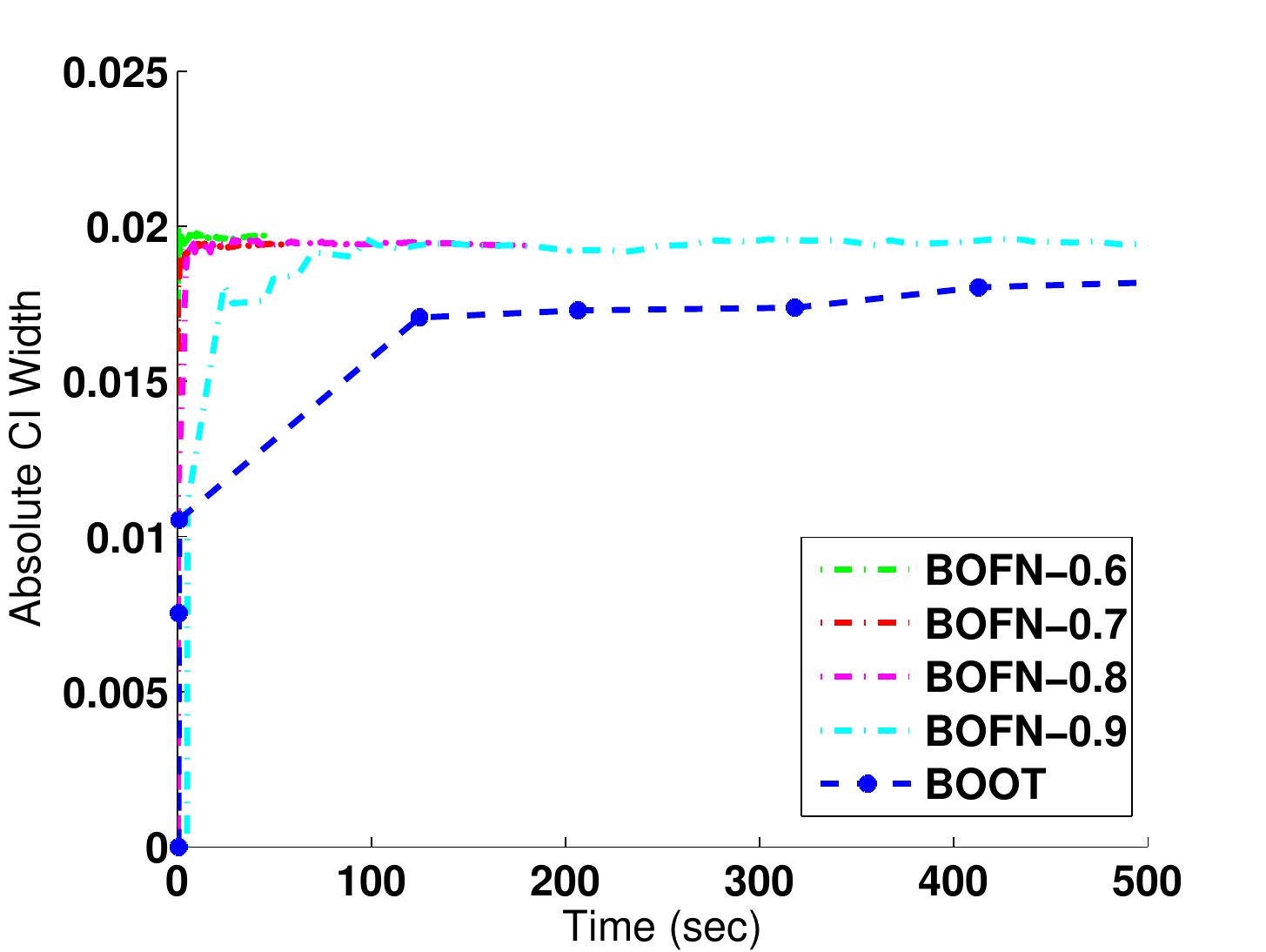}
\caption{Average (across dimensions) absolute confidence interval width vs.\ processing time on the UCI poker dataset (logistic regression, $d=10$, $n= 50,000$).  The left plot shows results for \ouralgAbbrev (using adaptive hyperparameter selection, with the output at convergence marked by large squares) and the bootstrap (BOOT).  The right plot shows results for the $b$ out of $n$ bootstrap (BOFN).}
\end{figure}

\begin{figure}[ht]
\includegraphics[width=0.5\linewidth]{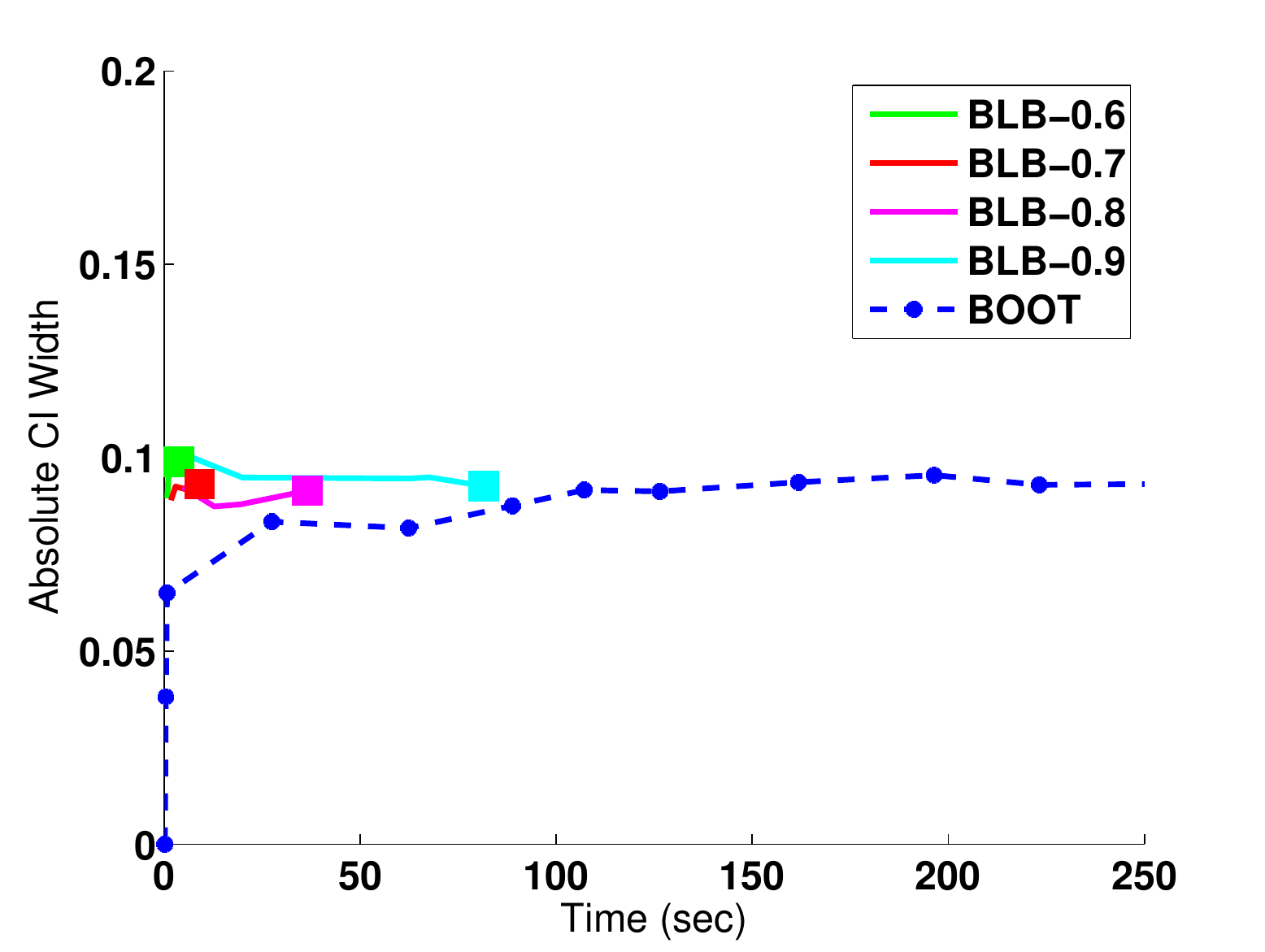}
\includegraphics[width=0.5\linewidth]{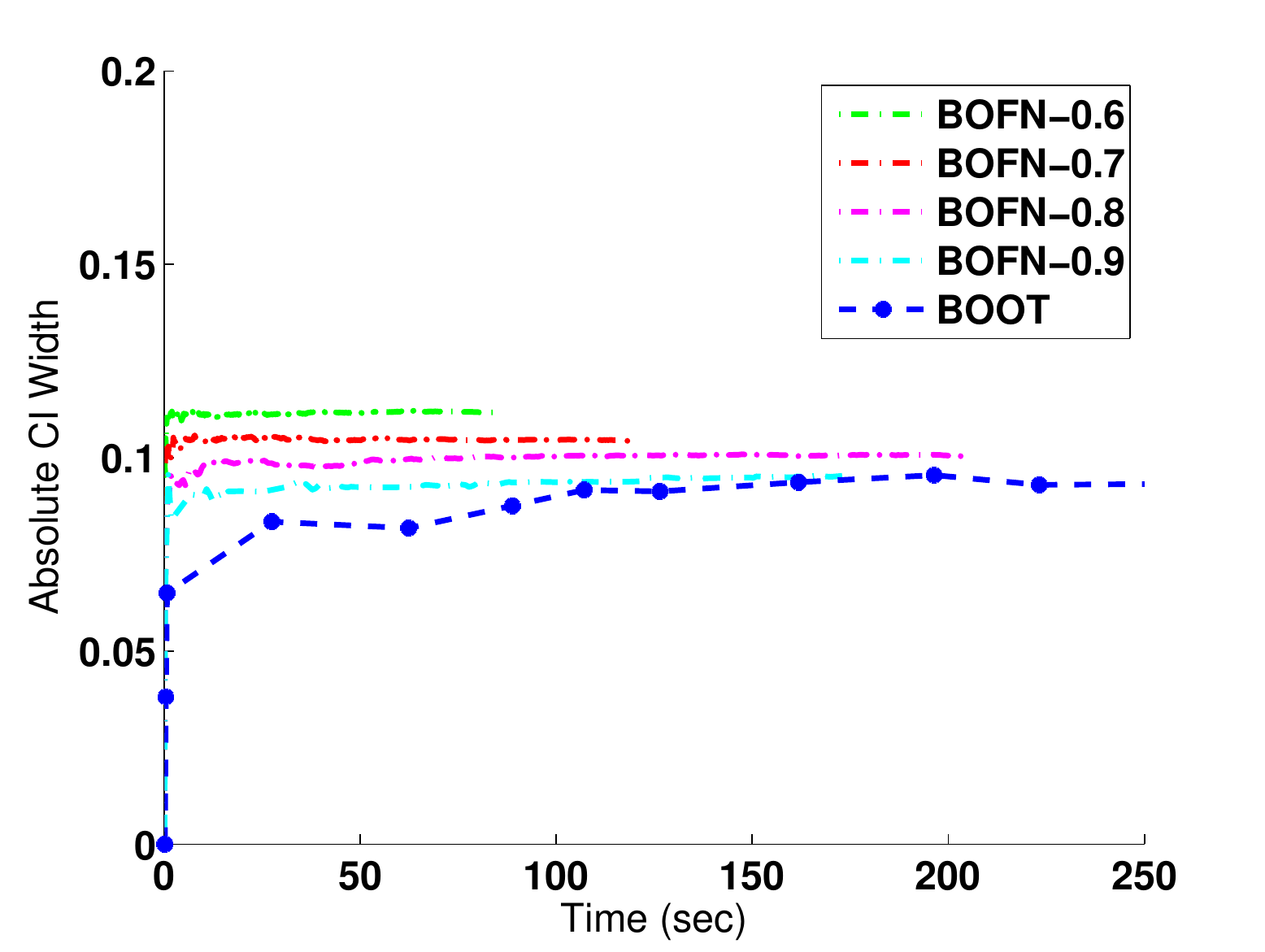}
\caption{Average (across dimensions) absolute confidence interval width vs.\ processing time on the UCI shuttle dataset (logistic regression, $d=9$, $n= 43,500$).  The left plot shows results for \ouralgAbbrev (using adaptive hyperparameter selection, with the output at convergence marked by large squares) and the bootstrap (BOOT).  The right plot shows results for the $b$ out of $n$ bootstrap (BOFN).}
\end{figure}

\end{document}